\documentclass{amsart}
\usepackage{amsmath}
\usepackage{amssymb}
\usepackage{amscd}

\usepackage[utf8]{inputenc}
\usepackage[T1]{fontenc}

\def\dj{d\kern-.30em\raise1.25ex\vbox{\hrule width .3em height .03em}}
\def\Dj{D\kern-.70em\raise0.75ex\vbox{\hrule width .3em height .03em}
\kern.03em}

\def\ginv{\Gamma_{ \mathrm{inv} }} 

\def\ddisp{\lambda}

\def\bla#1{$(${\it #1\/{}}$)$}

\def\FWP{\widehat{F}}
\def\hatphi{\widehat{\phi}}
\def\dP{d_{\!P}}

\def\cal{\mathcal}
\def\Bbb{\mathbb}
\def\frak{\mathfrak}

\newcommand{\id}{\mathrm{id}}
\newcommand{\ad}{\mathrm{ad}}
\newcommand{\Sum}{\displaystyle{\sum}}

\newcommand{\hor}{\frak{hor}}

\def\1{\varnothing}

\def\e{\varepsilon}
\def\S{S}

\def\St{\S}

\def\w{\mathrm{w}}

\def\wact{\blacktriangleleft}

\def\Guni{\Gamma^{\wedge}}
\def\Gbr{\Gamma^{\vee}}
\def\sw{\mu}

\newtheorem{theorem}{Theorem}[section]
\newtheorem{corollary}{Corollary}[section]

\newtheorem{prop}{Proposition}[section]
\theoremstyle{definition}
\newtheorem{definition}{Definition}[section]
\theoremstyle{remark}
\newtheorem{remark}[definition]{Remark}
\numberwithin{equation}{section}

\begin{document}

\title{Dunkl Operators for Arbitrary Finite Groups}

\author{Micho \Dj ur\dj evich}
\address{Instituto de Matem\'aticas,  
Universidad Nacional Aut\'onoma de M\'exico,\\
Area de la Investigaci\'on Cient\'{\i}fica, 
Circuito Exterior, \\
Ciudad Universitaria, CP 04510, Mexico City, MEXICO.}
\email{micho@matem.unam.mx}
\author{Stephen Bruce Sontz}
\address{Centro de Investigaci\'on en Matem\'aticas, A.C., 
(CIMAT)\\
Jalisco S/N, Mineral de Valenciana, CP 36023, 
Guanajuato, MEXICO.}
\email{sontz@cimat.mx}

\begin{abstract}
The Dunkl operators associated to a necessarily 
finite Coxeter group acting on 
a Euclidean space are generalized to any finite group 
using the techniques of 
non-commutative geometry, as introduced by the authors to view the usual 
Dunkl operators as covariant derivatives in a quantum principal 
bundle with a quantum connection. 
The definitions of Dunkl operators and 
their corresponding Dunkl connections 
are generalized to quantum principal bundles 
over quantum spaces which possess 
a classical finite structure group. 
We introduce {\em cyclic Dunkl connections} 
and their {\em cyclic Dunkl operators}.
Then we establish a number of interesting 
properties of these  structures, 
including the characteristic zero curvature property. 
Particular attention is given to 
the example of complex reflection groups, and their naturally 
generalized siblings called groups of {\em Coxeter type}. 
\end{abstract}


\maketitle

\section{Introduction}
\label{sec-1}
\parskip=2pt plus3pt minus1pt

In our previous article \cite{DS} we showed how one can 
view the Dunkl differential-difference operators as covariant 
derivatives in a specific quantum principal bundle endowed with 
a specific quantum connection, 
whose (quantum!) curvature turns out to be zero. 
The zero curvature then implies 
that the Dunkl operators 
associated to a Dunkl
connection commute among themselves. 
While this commutativity result dates back to 
Dunkl's original paper \cite{dunkl-1989}
on this topic, the viewpoint established in~\cite{DS} 
and continued here gives that result a 
geometric meaning. 
The Dunkl operators are associated 
to a finite Coxeter group which acts
on a finite dimensional Euclidean space 
as orthogonal transformations. 

But the general procedure used in \cite{DS} 
need not be restricted to 
Coxeter groups. 
Rather the theory can be developed, as we show here, 
in much the same detail on the
side of non-commutative geometry for 
arbitrary finite groups 
acting freely on a $C^\infty$ manifold 
which satisfies certain 
additional properties. 
This is the main topic of this paper, even though certain secondary topics, most 
especially cyclic structures as 
discussed later, will be necessary for
our presentation. 

This paper and \cite{DS} have 
opened up a bridge between theories 
that had been previously studied 
independently. On one side of this bridge
there is the non-commutative 
(or quantum) geometry of quantum principal bundles 
and their quantum connections. 
On the other side there is 
the harmonic analysis which 
originally motivated and then continued to 
grow out of the Dunkl
operators (see \cite{dunkl-1989} and \cite{rosler}), 
but also their applications in 
probability, Segal-Bargmann analysis 
and other areas of mathematical 
physics such as Sutherland-Calogero-Moser models.  
(See \cite{etingof}, \cite{rosler-voit}, \cite{sontz-heat} 
and references therein.) 
One can always cross the bridge back 
to other side of analysis 
in order to define and study the corresponding 
structures in harmonic analysis 
(such as a new generalization of the Fourier, or Dunkl,
transform) and in mathematical physics, 
though this is more 
straightforward in the case of curvature zero. 
However, we leave these interesting and 
important topics for 
consideration in our future investigations. 

Our paper is organized as follows. 
In the next section we begin by reviewing for the reader's convenience 
some background material on quantum differential calculus 
with an emphasis on finite groups,  quantum principal bundles, 
quantum connections and quantum covariant derivatives. 
We also give a new, quite general definition 
of quantum Dunkl connections, 
that generalizes the definition in \cite{DS}.  
In Section~3 we discuss the motivating example from \cite{DS} 
of the Dunkl operators associated with a Coxeter group acting 
on $\mathbb{R}^n$. 
In Section~4 we introduce a discrete geometry 
of points and lines in a finite set. 
This {\em cyclic geometry} is at the heart of the structure 
we wish to study. 
Our results come in Section~5 where we 
present a new construction, which enables us to define 
{\em cyclic Dunkl connections}  
for a quite large class of quantum principal bundles, 
including those with a quantum "total space" manifold 
having a finite structure group. 
We compute the curvature of these connections 
and then, under certain general conditions, we prove 
that the curvature is zero in many cases of interest 
or, more generally, is equal  
to the curvature of the initial background geometry. 
We also prove that these cyclic Dunkl connections possess 
an important multiplicative property. 

As our main illustration, we explain in Section~6 
how complex reflection groups 
and their associated Dunkl operators 
as introduced in \cite{dunkl-opdam} 
are included in this picture. In Section~7, we sketch a construction of a simple yet instructive 
class of examples of quantum principal bundles and their cyclic Dunkl connections, based on Cuntz algebras.
These examples work with arbitrary finite groups, and are such that both the base manifold and the bundle 
are truly quantum objects. 

The paper ends with three Appendices. 
In Appendix~A, having an independent interest, 
we analyze properties of an underlying geometrical structure, 
a kind of primitive cyclic geometry.  
This structure is associated to the space of 
left-invariant elements of a 
bicovariant, $*$-covariant first-order 
differential calculus over a finite group, 
and it is closely related to the properties of quantum Dunkl connections. 
Appendix~B presents some definitions, while
Appendix~C is a technical proof.

\section{Background Material}
\label{sec-2}

In this section we review in some detail the results on which the rest of
the paper is based. 
The references for this material are \cite{MichoQPB2}, 
\cite{MichoQPB3}, 
\cite{DS} and \cite{Part-II}. 
All vector spaces are over the complex numbers. 
All maps are linear over the field of complex 
numbers, unless otherwise indicated. 
For example, 
if $V$ is a non-zero vector space with 
an involutive, additive conjugation
$C : v \mapsto C(v) \equiv v^*$ satisfying 
$(\alpha v)^* = \alpha^* v^*$ for
all $v \in V$ and $\alpha \in \mathbb{C}$, then $C$ is not linear, 
but rather is {\em anti-linear}. 
Also, a linear map $T : V \to W$ that satisfies $T(v ^*) =  T(v)^*$,
where $V$ and $W$ are vector spaces with conjugation,  is called
a {\em $*$-morphism} or a {\em hermitian map}. 
The tensor product symbol $ \otimes $ without a subscript 
means the context appropriate tensor product over 
the complex numbers $ \mathbb{C} $. 
If $ \S $ is a (finite) set, we let $ \mathrm{Card} (\S) $ 
denote its (finite) cardinal number of elements. 
We assume familiarity with Sweedler's notation, 
which we sometimes use without explicit comment.

\subsection{Quantum Differential Calculus on Finite Groups}

We let $G$ denote a finite group. 
We put 
$$
\mathcal{A}:= 
\mathcal{F} (G) := \{ f : G \to \mathbb{C} \},
$$
the set of all complex valued functions with domain $G$. 
Then $\mathcal{A}$ is a $*$-Hopf algebra 
with identity element $1_{ \cal{A} } =1$, 
the constant function. 
The multiplication in $\cal{A}$ is defined point-wise, that is as
$(f_1 f_2) (g) := f_1(g) f_2(g)$ for $f_1, f_2 \in\cal{A}$ and $g \in G$,  
and so is commutative.
A complex number $\alpha \in \mathbb{C}$ will be used at times to 
denote the element $\alpha \, 1_{ \cal{A} } \in \cal{A}$. 

The antipode $\kappa : \cal{A} \to \cal{A}$, which is defined as 
$(\kappa (f) ) (g) := f(g^{-1})$ for $f \in \cal{A}$ and $g\in G$,  
satisfies 
$\kappa^2 = \mathrm{id}_{\cal{A}}$, 
the identity map on $\mathcal{A}$. 
So $\kappa$ is a bijection. 
We let $\phi : \cal{A} \to \cal{A} \otimes \cal{A}$ denote the 
co-product of $\cal{A}$, which is defined as the pull-back of the group 
multiplication function $G \times G \to G$. 
The co-product is co-commutative if and only 
if the group $G$ is abelian. 
The $*$-operation on $\mathcal{A}$ is given by 
point-wise complex conjugation, 
$f^* (g):= \overline{f (g)}$.  
The co-unit $\e : \cal{A} \to \mathbb{C}$ is defined by 
$\e (f) := f (e)$, where $e \in G$ is the identity element in~$G$. 

The right adjoint co-action 
$ \ad : \cal{A} \to \cal{A} \otimes\cal{A}$
is defined for all $ a \in \cal{A} $ by 
\begin{equation}
\label{define-right-adjoint}
  \ad (a) := a^{(2)} \otimes \kappa(a^{(1)}) \, a^{(3)}
\end{equation}
using Sweedler's notation 
for the double co-product of $a \in \mathcal{A}$ and 
the co-associativity of the co-product $ \phi $, 
namely
\begin{equation}
\label{phi-1-2-3-notation}
      (\phi \otimes \mathrm{id}) \, \phi (a) =
      (\mathrm{id} \otimes \phi) \, \phi (a) = 
      a^{(1)} \otimes a^{(2)} \otimes a^{(3)} 
      \in \cal{A} \otimes \cal{A} \otimes \cal{A}. 
\end{equation}

For each $g \in G$ we let $\delta_g \in \mathcal{A}$ denote the 
Kronecker delta function with value $1$ on $g$ and value $0$
on all other elements of~$G$. 
Then the set $\{ \delta_g~|~ g \in G \}$ is a vector space 
basis of $\cal{A}$. 
We suppose that $\S$ is a subset of $G$ satisfying 
these properties: 
\begin{enumerate}

\item
$\S^{-1} = \S$.

\item
$g^{-1} \S g = \S$ for all $g \in G$.

\item
$e \notin S$, where $e$ denotes the identity element of $G$.

\item 
$\S$ is non-empty. 

\end{enumerate}

At one extreme we could take $S = G \setminus \{ e \}$, 
where $G$ has at least two elements. 
At the other extreme we could have $ S = \{ g_0 \}$, 
where $ g_0 \in G \setminus \{ e \}$ is central and has 
order $2$. 
An example of the latter is the multiplicative 
group of $8$ quaternions 
$ \{ \pm 1, \pm i, \pm j, \pm k \} $ with $g_0 = -1$. 

Then $\S$ determines a unique 
{\em first-order differential calculus (fodc)}, 
$d : \mathcal{A} \to \Gamma$, which is 
{\em bicovariant} 
and {\em $*$-covariant}. 
The vector space $\Gamma$ is an $\cal{A}$-bimodule, and 
the {\em differential} $d$ satisfies the Leibniz rule 
$$
 d (a b ) = (d a ) b + a ( d b) 
$$ 
for $a,b \in \cal{A}$. 
Also $ \Gamma $ is generated as a vector space by all the 
elements of the form $ a (d b) $ for $a,b \in \cal{A}$. 
The bicovariance condition means that there is a 
{\em canonical left co-action} 
$ \Phi_L : \Gamma \to \cal{A} \otimes \Gamma $ 
and a {\em canonical right co-action 
$ \Phi_R : \Gamma \to \Gamma \otimes \cal{A}$} of $\cal{A}$ 
that co-act on $\Gamma$ in a compatible way
(meaning that these two co-actions commute); 
we say that $ \Gamma $ is an {\em $ \cal{A} $-bicomodule}. 
Also, the $*$-covariance means that the differential $ d $ as well as 
that the canonical right and left co-actions 
are $*$-morphisms. 

Without going into all the details 
we note that $\Gamma$ is 
isomorphic as 
a left $\cal{A}$-module to the free left $\cal{A}$-module 
$\cal{A} \otimes ( \ker \e / \cal{R} )$ 
for the (two-sided) ideal $\cal{R} \subset \ker \e \subset \cal{A}$  
that is determined by $\S$ by 
\begin{equation}
\label{define-cal-R}
         \cal{R} := \{ f \in \ker \, \e ~|~ f(s) = 0 
         \quad \mathrm{for~all~} s \in \S\}. 
\end{equation}
(In the general theory of fodc's $\cal{R}$ is a right ideal, but since
$\cal{A}$ is commutative every right ideal is a two-sided ideal.) 
Here $\e$ is the co-unit of $\cal{A}$. 
Under the isomorphism 
$$
\Gamma \cong \cal{A} \otimes ( \ker \e / \cal{R} )
$$ 
the right $\cal{A}$-module structure on $\Gamma$ is represented as 
\begin{equation}\label{Gamma-R}
(a\otimes [b])c=(ac^{(1)})\otimes [bc^{(2)}]. 
\end{equation}
The canonical left co-action of $\cal{A}$ on $\Gamma$ 
is identified with the map 
\begin{equation}
\phi \otimes id : \cal{A} \otimes ( \ker \e / \cal{R} ) \to 
\cal{A} \otimes \cal{A} \otimes ( \ker \e / \cal{R} ),
\end{equation}
where $\phi$ is the co-product of $\cal{A}$. 
The notation $ id $ here and always hereafter 
means the context appropriate identity map. 

The subspace of left-invariant elements 
of $ \Gamma $ under its canonical left co-action $ \Phi_L $, 
defined as 
$$ 
\ginv :=
\{ \omega \in \Gamma ~|~ 
\Phi_L (\omega) = 1_{\cal{A}} \otimes \omega \}, 
$$
is then identified with the subspace
$ 1_{\cal{A}}\otimes ( \ker \e / \cal{R} )$ 
of $\Gamma$.  
So we have 
$$
\ginv \cong \ker \e / \cal{R} \quad  \mathrm{and} \quad 
 \Gamma \cong \cal{A} \otimes \ginv.  
$$

The projection $ \pi_\mathrm{{inv}} : \Gamma \to \ginv $ is 
defined by $\pi_\mathrm{{inv}} ( a \otimes \omega ) := \e (a) \omega$ 
for every 
$ 
a \otimes \omega \in \cal{A} \otimes ( \ker \e / \cal{R} ) $. 
See \cite{Part-II} or \cite{Wdiff} for more details about fodc's, 
but notice that in those references $ \ginv $ is denoted 
as $ _{\mathrm{inv}}\Gamma $. 

Moreover, the condition $S \ne \emptyset$ 
is equivalent to the condition that 
the $\mathcal{A}$-bimodule $\Gamma$ is non-zero. 
It is important to point out that the finite group $G$ is uniquely a 
$C^\infty$  differential manifold in the discrete 
topology, in which case it has dimension zero 
and so its de~Rham fodc (that is, the 
space of differential $1$-forms on $G$) is the zero vector space. 
Therefore the fodc we have constructed from $\S$ is not the de~Rham fodc 
of classical differential geometry but rather is a quantum object. 
So, the condition $ \S \ne \emptyset $ is included here 
only to exclude the trivial case $ \Gamma = 0 $. 
This condition is not included in Theorem 13.2 in 
\cite{Part-II} where the notation $ J $ 
corresponds to our notation $ \S $. 

There is a surjective linear map 
$\pi : \mathcal{A} \to \,\ginv$ 
defined by $ \pi := \pi_\mathrm{inv} \, d $. 
This map, which is called the {\em quantum germs map},  
provides us with a basis $ \{ \pi(\delta_s) ~|~ s \in \S \}$
of the vector space $\ginv$. 
In particular, $\dim (\ginv) = \mathrm{card} (\S) \ne 0$, which is 
one way to understand the fact that 
$\Gamma  \cong \cal{A} \otimes \ginv \ne 0$. 
The quantum germs map $\pi$ is also closely connected to the 
differential $d : \cal{A} \to \Gamma$ by the identity  
for all $a\in\cal{A}$ 
\begin{equation}
\label{diff-id}
d a = a^{(1)} \pi (a^{(2)} ), 
\end{equation}
where $\phi (a) = a^{(1)} \otimes a^{(2)}$ 
in Sweedler's notation. 
See Section~6.4 of \cite{Part-II} for this identity. 

It is important to note that just the group $ G $ 
does not determine a 
unique bicovariant, $ * $-covariant fodc. 
In general, different choices of the subset $ \S $ lead 
to different (i.e., non-isomorphic) fodc's. 
For example, as noted above even the dimension of $ \ginv $
depends on $ \S $. 
 
The actual formula for $d$ has its own interest. 
The action of $d$ on the basis
$\{ \delta_g ~|~ g \in G \}$ of 
$ \cal{A} = \cal{F} (G) $ is given by 
\begin{equation}
\label{d-of-delta-g}
       d (\delta_g) = \sum_{s \in S} (\delta_{g s^{-1}} 
       - \delta_g ) \pi (\delta_s). 
\end{equation}
We remark that this is quite similar 
to a nearest neighbor formula 
as frequently used in mathematical physics. 
However, here the nearest neighbor differences appear as 
the coefficients of distinct elements of a basis of $\ginv$ 
rather than being added directly together. 
So the formula \eqref{d-of-delta-g} looks more like 
a "gradient at $g$" associated with a family 
of "directional derivatives" 
$\delta_{g s^{-1}} - \delta_g$ 
indexed by the elements $s\in\S$. 

\begin{remark} In what follows, we shall use the same symbol 
for diverse mutually naturally identifiable objects. 
In particular, we shall 
commonly switch between $g$ to denote an element 
of the group $G$ and to denote 
its associated Kronecker delta function $\delta_g$. 
We shall also write $[g]$, or simply $g$ 
when there is no risk of ambiguity, 
to denote the quantum germ $\pi(g)$. 
In much the same spirit we let $\e$ denote the identity element 
in $G$ as well the co-unit of $\cal{A}$.  
\end{remark}

\subsection{Quantum Principal Bundles}\label{QPB}

In general, we say that a triple $P = (\cal{B}, \cal{A}, F)$ is 
a {\em quantum principal bundle} 
if $\cal{B}$ is a $*$-algebra 
with identity element $ 1_{\mathcal{B}} $, 
$\cal{A}$ is a $*$-Hopf algebra 
with identity element $ 1_{\mathcal{A}} $
and $F : \cal{B} \to \cal{B} \otimes \cal{A}$ 
is right co-action of $\cal{A}$ on $\cal{B}$ that is also a unital, 
multiplicative, $*$-morphism of algebras. 
One also requires of $F$ a technical property that is dual to that 
of an action being free. 
Next, the $ * $-algebra $\mathcal{V}$ of the "base space" 
is defined to be the right invariant 
elements in $\mathcal{B}$ under the right co-action by $F$, 
namely,
\begin{equation}
\label{define-base-space}
       \mathcal{V}:= \{ b \in \mathcal{B} ~|~ 
       F(b) = b \otimes 1_{\mathcal{A}} \}. 
\end{equation}

Suppose the finite group $G$ acts freely 
on the right on $E$, a $C^\infty$ manifold. 
Then the quotient map $E \to E/G$ is a principal bundle with structure 
group $G$, a zero dimensional Lie group. 
So far this is purely a construction in classical differential geometry. 

We now modify this construction in order to get  
a quantum principal bundle 
in non-commutative geometry with the $ * $-Hopf algebra 
$\mathcal{A}= \mathcal{F} (G)$ being its structure "quantum group".  
In this construction we follow the method used in \cite{DS} closely, 
though there are some details here which will be different. 
We construct a {\em quantum principal bundle (QPB) 
with finite structure group $G$}  
to be the triple $P = (\mathcal{B}, \mathcal{A}, F)$, 
where 
$\mathcal{A} = \cal{F}(G)$ as already defined above and 
$\mathcal{B}:= C^\infty (E)$, the $*$-algebra of 
{\em complex valued} $C^\infty$ functions $f : E \to \mathbb{C}$, 
with $E$ as above. 
The $*$-operation on $\mathcal{B}$ is given by 
pointwise complex conjugation 
of such a function $f$. 
Lastly the linear map (which is a unital, multiplicative, 
$*$-morphism of algebras) 
$F :  \mathcal{B} \to \mathcal{B} \otimes \mathcal{A}$ 
is defined as the pull-back of the right action map 
$\alpha : E \times G \to E $. 
So, $F$ is a right co-action of $\mathcal{A}$ on $\mathcal{B}$. 
Since $\alpha$ is a free action, the co-action $F$ satisfies the 
technical condition dual to being a free action. 
This completely describes the QPB $P = (\mathcal{B}, \mathcal{A}, F)$ 
with finite structure group $G$. 

Of course, a  QPB with finite structure group $G$ 
is a special case of a QPB for which an extensive theory
has been developed. 
Much of what follows holds in the context of QPB's in general, though 
eventually we focus on the special case of a QPB with a finite 
structure group. 

But in general just a QPB in and of itself does not have enough 
structure for our purposes. 
We also need to associate differential calculi to the algebras. 
We do this first in the case 
of {\em any} algebra $\mathcal{A}$
with identity element $ 1_{\cal{A}} $, not necessarily 
$ \cal{B} $ or $\cal{F} (G)$. 
We start with a given fodc $d : \cal{A} \to\Gamma$. 
This is then associated to a graded 
algebra $\Lambda = \bigoplus _{k=0}^\infty \Lambda^k$, 
meaning that $ \Lambda^k \cdot \Lambda^l \subset \Lambda^{k+l} $. 
Each {\em homogeneous subspace} $ \Lambda^k $ 
is also required to be an $ \cal{A} $-bimodule. 
We require that $\Lambda$ have a {\em differential}  
$d : \Lambda^k \to \Lambda^{k+1}$, meaning $d^2 =0$, 
that also satisfies 
for all $\varphi_1 \in \Lambda^n$ the {\em graded Leibniz rule} 
$$
    d (\varphi_1 \varphi_2) =   
    d (\varphi_1) \, \varphi_2  + (-1)^n \varphi_1 \, d(\varphi_2). 
$$
We require that $\Lambda$ is an {\em extension} of the fodc $\Gamma$ 
in the sense that  $\Lambda^0 = \cal{A}$,  
$\Lambda^1 = \Gamma$ and 
the map $d : \Lambda^0 \to \Lambda^1$ 
coincides with the differential 
$d : \cal{A} \to \Gamma$ of the given fodc, 
and hence this justifies using the same notation $ d $. 
We require as well that the identity element 
$ 1_{\cal{A}} \in \cal{A}$ is 
also the identity element of $ \Lambda $. 
Any such object $\Lambda$ 
is called a 
{\em higher-order differential calculus (hodc)} that
{\em extends} the given fodc. 
Moreover, in the case of an fodc over $ G $ we require that the hodc 
$\Lambda$ must be generated as a differential algebra 
by $\Lambda^0$ or, in other words, every element in 
$\Lambda^k$ can be written as a finite sum of 
elements of the form $a_0 \, d a_1 \cdots d a_k$, where 
$a_0, a_1, \dots, a_k \in \Lambda^{0} = \cal{A}$ or,
equivalently, of the form  $d b_1 \cdots d b_k \, b_0$, 
where $b_1, \dots, b_k, b_0 \in \Lambda^{0} = \cal{A}$.  
If $ \cal{A} $ is a Hopf algebra, we also require bicovariance 
of each $ \Lambda^k $. 
So far the discussion in this paragraph does not 
involve $*$-operations and the corresponding $ * $-covariance. 
However, these are also to be included as presented 
in Chapter~11 of \cite{Part-II}. 
For the special case of a QPB with finite structure group 
the algebras $ \cal{A} = \cal{F} (G) $ and 
$ \cal{B} = C^\infty (E) $ are actually $ * $-algebras, 
and so it is natural to require a compatible 
$ * $-operation on all the associated 
vector spaces $ \Lambda^k $ introduced above. 
     
As an aside we note that a graded 
algebra $\Lambda = \bigoplus _{k=0}^\infty \Lambda^k$
is said to be {\em graded commutative} if 
$\varphi_1 \varphi_2 = 
(-1)^{n m} \varphi_2 \varphi_1 \in \Lambda^{n+m} $
whenever $\varphi_1 \in \Lambda^n$ and 
$\varphi_2 \in \Lambda^m$. 
In particular, this implies that $ \Lambda^0 $ is 
commutative. 
We do not require this property of an arbitrary hodc, 
although the classical de Rham calculus is graded commutative. 

A {\em higher-order differential calculus (hodc) 
for a QPB $P=(\cal{B}, \cal{A}, F)$} is a 
triple $ ( \Omega(P), \Lambda, \FWP )$, 
where $ \Omega(P) $ is an hodc for $ \cal{B} $ and 
$ \Lambda $ is a bicovariant, $*$-covariant hodc 
for $ \cal{A} $.  
Finally, $ \FWP : \Omega(P) \to \Omega(P) \otimes \Lambda $
is a grade preserving, 
right co-action of $ \Lambda $ on $ \Omega (P) $
which also extends the right co-action 
$F : \cal{B} \to \cal{B} \otimes \cal{A}$ and  
is a multiplicative, differential, unital $ * $-morphism. 
See Section~12.4 of \cite{Part-II} 
for the definition of what it means for 
$\FWP$ 
to be a right~co-action. 
The notation $ \FWP $ is used in the papers of the first author 
such as \cite{MichoQPB2},  
while the more cumbersome notation $_{ \Omega(P) }\Psi$ 
is used for this right co-action 
in the book \cite{Part-II} of the second author.

\subsection{A Special Case: The Hopf Algebra of 
	a Finite Group}

We now study the case $ \cal{A} = \cal{F} (G) $. 
For this case when the fodc 
$d : \cal{A} 
\to\Gamma$ is 
bicovariant and $*$-covariant 
there is a particular hodc 
$\Guni$ called the {\em universal hodc}
associated to the fodc~$\Gamma$, 
which itself has left and right co-actions of 
$\cal{A}$ co-acting on it. 
Every such hodc extending the fodc 
$d: \cal{A} \to \Gamma$ is a quotient of $\Guni$, 
thereby justifying the qualifier `universal.' 
(See \cite{MichoQPB1} or \cite{Part-II} 
for the construction of $\Guni$.) 
The theory also can be developed using another hodc, called 
the {\em braided differential calculus} of Woronowicz. 
(See \cite{Part-II} or \cite{Wdiff}.) 
The latter hodc, denoted as $\Gbr$, 
is based on the braided exterior algebra, 
which is associated to a canonical 
$\cal{A}$-bimodule braid-automorphism 
$\sigma\colon\Gamma\otimes_{\cal{A}}\Gamma\rightarrow 
\Gamma\otimes_{\cal{A}}\Gamma$. 
Its restriction 
$\sigma\colon\ginv\otimes\ginv\rightarrow\ginv\otimes\ginv$ 
is given for $\eta, \vartheta \in \ginv$ by 
\begin{equation}
\label{sigma-eta-vartheta} 
\sigma(\eta\otimes\vartheta)= 
\vartheta^{(0)} \otimes (\eta \circ \vartheta^{(1)}), 
\end{equation}
where $\circ$ denotes a canonical right action of 
$\mathcal{A}$ on $\ginv$ defined 
for $b, c \in \cal{A}$ by 
\begin{equation}
\label{define-circ-action}
\pi(b) \circ c := \pi (b c - \e (b) c ), 
\end{equation}
where $\pi$ is the quantum germs map. In terms of the identification $\ginv$ with $\ker \e /\cal{R}$ this is 
the factor right $\cal{A}$-module structure. 
Also, 
$\ad(\vartheta) = \vartheta^{(0)} \otimes \vartheta^{(1)}$ 
in Sweedler's notation, where  
$\ad\colon\ginv\rightarrow \ginv\otimes \cal{A}$ 
denotes the right adjoint co-action 
of $\cal{A}$ on $\ginv$. 
The relation of this $ \ad $ with the right adjoint co-action 
$ \ad $     defined in \eqref{define-right-adjoint} 
is given by the commutative diagram 
\begin{equation*}
\begin{CD}
\cal{A} @>{\mbox{$\ad$}}>> \cal{A} \otimes \cal{A} \\
@V{\mbox{$\pi$}}VV @VV{\mbox{$\pi$}}V\\ 
 \ginv @>>{\mbox{$\ad$}}> \ginv \otimes \cal{A}
\end{CD}
\end{equation*}
that is, the previously defined right adjoint co-action  
$ \ad $ passes
to the quotient by the quantum germs map $ \pi $. 
As a formula this diagram reads as
\begin{equation}
\label{ad-ad-relation}
  \ad \, \pi = (\pi \otimes \mathrm{id}_\cal{A}) \, \ad. 
\end{equation}
Using \eqref{define-right-adjoint}, 
\eqref{phi-1-2-3-notation} and 
\eqref{ad-ad-relation} 
we obtain 
the formula 
$\ad \, \pi(a)\!=\!\pi(a^{(2)})\otimes 
\kappa(a^{(1)}) \, a^{(3)}$ 
for all $ a \in \cal{A} $. 
The long, technical proof of \eqref{sigma-eta-vartheta} 
is given in Appendix~C.

\subsection{A Special Type of Hodc}

Let us here present an important `intermediary' context for an hodc.
This will be highly relevant in our considerations 
of general Dunkl connections, 
because all of the considered fodc's 
over finite groups are of this form. 
So we continue studying $ \cal{A} = \cal{F}(G) $. 
Namely, let us assume that there exists an element 
$q\in\ker(\e) \setminus \cal{R}$ 
that satisfies 
$$
(q-1)\ker(\e)\subseteq\cal{R} 
$$ 
where $\e : \cal{A} \to \mathbb{C}$ is the 
co-unit of the Hopf algebra $\cal{A}$ 
and $\cal{R}$ is the ideal in $\ker(\e)$ that determines 
the bicovariant, $ * $-covariant fodc 
$\Gamma$ on $\cal{A}$ given by the subset $S$, 
as defined in \eqref{define-cal-R}. 
Next, let us define 
$$
\tau:=\pi(q)\in\ginv. 
$$
Recalling that $\ker \, \pi \cap \ker \e = \cal{R}$, 
we see that $\tau \neq 0$. 
Also, we remind the reader that 
$\ker \, \pi = \cal{R} + \mathbb{C} \, 1_{\cal{A}}$. 

First, we note for all $a\in\cal{A}$ that 
\begin{equation*}
\pi \Big( q  \big( a - \e(a) \big) \Big) =
\pi \Big( (q-1)  \big( a - \e(a) \big) \Big) 
+ \pi \big( a - \e(a) \big) = \pi (a), 
\end{equation*}
where we used 
$(q-1)  \big( a - \e(a) \big) \in (q-1) \ker \, \e 
\subseteq\cal{R}\subset\ker \, \pi$ and $\pi(1) = 0$. 
On the other hand we have
$$
\pi \Big( q  \big( a - \e(a) \big) \Big) = \pi (q a) - \e(a) \pi(q) = 
\pi(q a) - \e(a) \tau. 
$$

Combining the last two equalities we arrive at 
\begin{equation}
\label{arrive-at}
\pi(q a) = \pi(a) + \e(a) \tau.
\end{equation}

Then it is easy to see that 
\begin{equation}
\label{int-circ}
\tau\circ a = \pi(q) \circ a = \pi \bigl(q a - \e(q) a \bigr) 
= \pi (q a) = \pi(a)  +  \e(a) \tau
\end{equation}
where we have used the definition 
\eqref{define-circ-action}
of $\circ$, 
$q \in \ker \, \e$ and \eqref{arrive-at}. 

We will also be using the identity 
\begin{equation}
\label{curious-identity}
\theta \, a = a^{(1)} ( \theta \circ  a^{(2)} ) 
\end{equation}
for $a \in \cal{A}$ and $\theta \in \ginv$, 
where $\phi (a) =  a^{(1)} \otimes a^{(2)}$
is Sweedler's notation for the co-product $\phi(a)$ of $a\in\cal{A}$. This is actually  
another way to describe the right $\cal{A}$-module structure \eqref{Gamma-R} on $\Gamma$. 
For the sake of completeness, we include a proof of this. 
So we have that 
\begin{equation*}
a^{(1)} ( \theta \circ  a^{(2)} ) 
= a^{(1)} \bigl( \kappa ( a^{(2)}) \, \theta \, a^{(3)} \bigr)  
= \e(a^{(1)}) \theta a^{(2)}
= \theta \, \e(a^{(1)}) a^{(2)} 
= \theta \, a. 
\end{equation*}
Here we used 
the co-associativity of the co-product expressed in 
Sweedler's notation and basic
Hopf algebra properties of the 
antipode $\kappa$ and the co-unit $\e$ as well as 
this  property of the action $\circ \,$: 
$$
   \theta \circ b = \kappa (b^{(1)}) \, \theta \, b^{(2)} 
$$
for $ \theta \in \ginv $ and $ b \in \cal{A} $, 
where $ \phi (b) = b^{(1)} \otimes b^{(2)} $. 
See Section~6.4 and Appendix~B of \cite{Part-II} 
for more details. 

Next, for any $a \in \cal{A}$ we compute
\begin{align}
\label{d-is-inner}
    d a &= a^{(1)} \pi ( a^{(2)} ) 
    = 
    a^{(1)} \bigl( \tau \circ a^{(2)} - \e(a^{(2)}) \tau \bigr)\\ &= 
    a^{(1)} \bigl( \tau \circ a^{(2)} \bigr) - a^{(1)}\e(a^{(2)}) \tau 
    = \tau a - a \tau, \nonumber
\end{align}
where we used the identity \eqref{diff-id}
for the differential $d$, 
the Hopf algebra 
property $ a = a^{(1)} \e (a^{(2)}) $ of the co-unit $\e$ 
as well as formulas \eqref{int-circ} and \eqref{curious-identity}.  

This result, $d a = \tau a  -  a \tau$, 
is a sort of commutation relation, 
because it measures the difference between the left and right 
$\cal{A}$-module structures acting on the element 
$\tau\in\ginv\subset\Gamma$ in the 
$\cal{A}$-bimodule $\Gamma$. 
This justifies saying that $d$ is an {\em inner derivation} 
with respect to the element $\tau$. 
Then any graded algebra 
$\Omega(G) = \bigoplus_{k=0}^\infty \, \Omega^k(G)$ 
which extends the fodc 
$d : \cal{A} = \cal{F}(G) \to \Gamma$ 
and in which $\tau^2=0\in\Omega^2(G)$ 
will automatically become an hodc 
with its extended differential being defined by the graded commutator 
with $\tau$, namely 
$$ 
d \psi := \tau\psi-(-1)^k  \, \psi\tau\in\Omega^{ (k+1) } (G), 
$$
where $\psi\in\Omega^k(G)$ is any 
homogeneous element of degree $k$. 
The point is that one can readily verify that $d^2 = 0$ and 
that $d$ satisfies the graded Leibniz rule. 
As we shall see, this is precisely what happens for 
bicovariant, $*$-covariant 
differential calculi 
over finite groups, where we can take $q=1_{\cal{A}}-\e$. 
Notice the similarity of this construction of an hodc with 
the {\em extended bimodule method} as used in \cite{Wtwisted} 
and \cite{Wdiff}. 
That method is also presented in Section~10.1 of  \cite{Part-II}. 

The mere existence of the element 
$q\in\ker(\e) \setminus \cal{R}$ implies that 
$\cal{R} \subsetneq \ker \e$ and so 
$ \ginv \cong \ker \e / \cal{R} \ne 0 $, which
in turn implies that 
$ \Gamma \cong \cal{A} \otimes \ginv \ne 0 $. 
So the fodc is non-zero in this situation. 
Recall that $ \Gamma = 0 $ for the de Rham theory 
of the zero-dimensional manifold $G$ of classical differential geometry. 
So the above discussion is only for the quantum case. 

\subsection{The Hodc of the Total Space}

Next we will define an hodc for the "total space"   
$\cal{B}$ of the QPB $P = (\cal{B}, \cal{A}, F)$ 
with finite structure group $G$ and with fodc $d : \cal{A} \to \Gamma$
associated to $G$ as above. 
More specifically, $ E $ is a $ C^\infty $~manifold on which 
$ G $ acts freely from the right and $ \cal{B} = C^\infty (E) $. 
As a graded vector space this is defined by this tensor 
product of graded vector spaces: 
\begin{equation}
\label{define-Omega-P}
    \Omega (P) := \mathcal{D} \otimes \,\ginv^\wedge 
    \quad \mathrm{or~explicitly} \quad  
    \Omega^m (P) :=  \bigoplus_{k + l = m} 
     \mathcal{D}^k \otimes \,\ginv^{\wedge \, l} 
\end{equation}
for integers $k,l,m \ge 0$. 

There is a lot to explain here. 
First, $\cal{D} := \Omega_{dR} (E) \otimes \mathbb{C}$ 
is the complexified 
de~Rham exterior calculus of the $C^\infty$ manifold $E$. 
Note 
$\Omega^{0}(P)=  \cal{D}^0 = C^\infty (E)=\cal{B}$ 
as it must be in order to have an extension. 
(As we shall see later $ \ginv^{\wedge 0} = \mathbb{C} $). 
Second, $\ginv^\wedge$ denotes the space 
of the left invariant elements 
in the universal hodc $\Guni$ under its canonical 
left co-action by $\mathcal{A} = \cal{F}(G)$. 

Instead of using  $ \Guni $, 
we could have used an {\em acceptable} hodc $\Omega(G)$, 
which is defined to be an hodc  
extending $ d : \cal{A} \to \Gamma $  
that is generated as a differential algebra by 
$ \Omega^0 (G) =\cal{A} = \cal{F}(G)$ 
and such that the co-multiplication  
$\phi : \cal{A} \to \cal{A} \otimes \cal{A}$ 
has a necessarily unique extension 
$\hatphi : \Omega(G) \to \Omega(G) \, \otimes \, \Omega(G)$
that is a differential algebra morphism. 
Examples of acceptable hodc's 
are the universal enveloping hodc $\Guni$ 
and the braided hodc $\Gbr$. 
Also, using the partial order of hodc's (defined 
by $ \Omega^\prime (G) \preceq \Omega (G) $ 
if $ \Omega^\prime (G) $ 
is a quotient of $ \Omega(G)$), 
all other acceptable hodc's are 
intermediates lying between the two extremes cases 
of $\Gbr$ and $\Guni$, which satisfy 
$ \Gbr \preceq \Guni $.
This opens up a richness in the quantum theory 
that is not available with the unique, functorial de Rham 
hodc of classical differential geometry. 
Also, we claim that the very existence 
of such an extension $ \hatphi $ 
implies that the fodc being extended, namely 
$d : \cal{A} = \Omega^0 (G) \to \Gamma = \Omega^1 (G) $, 
is bicovariant, since the map 
$$
\hatphi : \Gamma = 
\Omega^1 (G) \to (\Omega(G) \otimes \Omega(G))^1 
= (\Gamma \otimes \cal{A}) \oplus (\cal{A} \otimes \Gamma)  
$$ 
has projections to 
the first (resp., second) summand which give a right
(resp., left) co-action of $\cal{A}$ on $\Gamma$. 
Moreover, these two co-actions are compatible, 
thereby making the fodc $\Gamma$ bicovariant as claimed. 
However, hereafter $\Omega(P)$ means \eqref{define-Omega-P} 
with the universal hodc $ \Guni $ unless 
stated otherwise. 
 
Moreover, we remark that $\Omega(P)$, 
defined in \eqref{define-Omega-P}, 
can be given four operations:  
a multiplication, a $*$-operation, 
a differential $d_P$ of degree $ +1 $ and a right co-action 
$$
    \FWP : \Omega (P) \to \Omega (P) \otimes \Gamma^\wedge
$$
which extends $ F : \cal{B} \to \cal{B} \otimes \cal{A}$. 
All this gives us 
a graded $*$-algebra such that $d_P$ satisfies 
the graded Leibniz rule and is a $*$-morphism. 
Consequently, $\Omega(P)$ is an hodc. 
The definition \eqref{define-Omega-P} of $\Omega(P)$ seems 
to be a trivial product, but this is not 
so because these four operations 
involve a non-trivial `twisting' using the right co-action 
$_{\cal{D}}\Phi$ defined below in
\eqref{define-Coxeter-co-action}. 
For the definitions of these four operations 
see Appendix~B. 
Finally, we remark that $ \Omega (P) $ is generated 
as a differential algebra by $ \Omega^0 (P) = \cal{B} $. This is a kind of minimality condition for the calculus 
and ensures that the geometric and algebraic properties 
related to the differential calculus are uniquely defined 
by their restrictions to $\cal{B}$.  
 
Then $\Omega(P)$ is taken to be the hodc for the algebra 
$\mathcal{B} = C^\infty(E)$ of the total space. 
Notice again that this is much more than 
simply the hodc $\mathcal{D}$ of classical 
differential geometry.

\subsection{The Hodc of the Base space}

The hodc $\Omega(M)$ of the $ * $-algebra 
$\mathcal{V}$ of the base space 
is defined to be the right invariant elements of $\Omega(P)$, 
defined as the graded vector space 
\begin{equation}
\label{define-omega-M}
\Omega(M):= \{ \, \omega \in \Omega(P) ~|~ \FWP (\omega) = 
\,  _{ \Omega(P) }\Psi (\omega) = \omega \otimes 1_{\cal{A}} \, \}. 
\end{equation}
It turns out that $\Omega(M) \subset \Omega(P)$ is closed 
under the multiplication, the $*$-operation 
and the differential of $\Omega(P)$ and so, with those 
operations, is an hodc in its own right. 
However, this hodc is not necessarily generated as 
a differential algebra by its elements in degree $0$. 
But do notice that 
the degree $0$ elements of $\Omega(M)$ give exactly 
the algebra $\cal{V}$ of the "base space" 
as defined in \eqref{define-base-space}. 
Also beware that the notation can be misleading. 
The hodc $\Omega(M)$ is well defined for any QPB $P$
with an hodc $\Omega(P)$, 
even though there is no "base space" $M$.

\subsection{Horizontal Forms}

In classical differential geometry a principal bundle 
has canonically associated 
vertical tangent vectors of the total space. 
But the horizontal vectors of such a bundle 
are not uniquely determined, 
though any `reasonable' choice of them is (or is equivalent to) 
an Ehresmann connection on the principal bundle. 
(See \cite{spivak} for much more about this.)  
In the theory of non-commutative geometry 
we typically have a space corresponding to $1$-forms
(namely, the fodc~$\Gamma$) instead of a space of tangent vectors. 
So the quantum situation is dual to the classical situation, that is, 
a QPB has a canonically associated horizontal space  
$\mathfrak{hor} (P) \subset \Omega(P)$, 
which does not require the existence of a 
(quantum) connection for its definition. 
Specifically, we use the definition in \cite{MichoQPB2}, namely
\begin{equation}
\label{define-horizontal-forms}
    \mathfrak{hor} (P) :=  
    \{ \, \omega \in \Omega(P) ~|~ \FWP (\omega) 
     \in \Omega(P) \otimes \cal{A} \, \} 
    = \FWP^{-1} ( \Omega(P) \otimes \cal{A} ), 
\end{equation}
which is a $*$-subalgebra of $\Omega(P)$. 
We say that the elements of $ \mathfrak{hor} (P) $ 
are the {\em horizontal forms} 
of $\Omega (P)$. 
So $ \Omega(M) \subset \mathfrak{hor} (P)$, 
possibly with a proper inclusion, and
$ \mathfrak{hor} (P)$ inherits the structure 
of a graded vector space from $\Omega(P)$, that is 
$$
\mathfrak{hor} (P) = \bigoplus_{k = 0}^\infty \mathfrak{hor}^k (P). 
$$ 
The space $ \mathfrak{hor} (P)$ has many nice properties that justify 
calling it the {\em horizontal space} in $\Omega (P)$.
In this regard, see \cite{MichoQPB2}. 
One such nice property is that $\omega \in \mathfrak{hor} (P)$ implies 
that $\FWP (\omega) \in \mathfrak{hor} (P) \otimes \cal{A}$. 
So, $\FWP$ restricted to $\mathfrak{hor} (P)$ gives 
a right co-action of $\cal{A}$ on the horizontal space 
$\mathfrak{hor} (P)$. 
We want to emphasize that $  \mathfrak{hor} (P) $ is not always 
invariant under $d_P$, the differential of $\Omega(P)$. 
We will come back to this point when we discuss 
the covariant derivative. 

For the case when we have a QPB with 
finite structure group $G$ 
equipped with the hodc $\Omega(P)$, as defined 
in \eqref{define-Omega-P}, it turns out that 
$$
      \mathfrak{hor} (P) = 
      \mathcal{D} \otimes 1_{\mathcal{A}} \cong  \mathcal{D}, 
$$
that is, the horizontal forms are identified with 
the complexified de~Rham forms of the $C^\infty$ 
manifold $E$. 
(See Theorem~14.3 in \cite{Part-II}.) 
This fundamental fact serves as a motivation and 
justification for the definitions in 
\eqref{define-Omega-P} and Appendix~B.

\subsection{Quantum Connections}

A {\em (quantum) connection} on a general QPB $P$ 
with an hodc $\Omega(P)$ is a linear map 
$\omega : \,\ginv \to \Omega^1 (P)$ satisfying 
these two conditions for all $\theta \in \,\ginv $: 
\begin{align}
             \label{omega-is-star-morphism}
             \omega (\theta^*) &= \omega (\theta)^* 
             \qquad (\mathrm{optional})
             \\
             \label{characteristic-property-of-a-connection}
             \FWP \big( \omega (\theta) \big) 
             &= (\omega \otimes \mathrm{id}_{ \mathcal{A} }) \, 
             \ad (\theta) 
             + 1_{ \mathcal{B} } \otimes \theta.
\end{align}
As already mentioned, 
$\ad : \!\,\ginv \to \,\ginv\otimes \mathcal{A}$ 
is the right adjoint co-action of $\mathcal{A}$ on $\ginv$. 
We note that these conditions are {\em not} linear in $\omega$ due to the 
inhomogeneous term $1_{ \mathcal{B} } \, \otimes \, \theta$ 
in \eqref{characteristic-property-of-a-connection}. 
The second condition \eqref{characteristic-property-of-a-connection} 
 is an encoding of the properties of an 
Ehresmann connection in the dual context of non-commutative 
geometry. 
The property \eqref{omega-is-star-morphism} is a reality condition 
that is usually satisfied in classical differential geometry, 
since in that context the underlying scalar field is 
typically the reals $\mathbb{R}$. 
Unlike \eqref{characteristic-property-of-a-connection}, the 
property \eqref{omega-is-star-morphism} is devoid of 
interesting geometric content. 
And this is why \eqref{omega-is-star-morphism} 
is an optional condition. 

The object defined by the corresponding homogeneous conditions is
known as a {\em quantum connection displacement} or briefly a {\em QCD}. 
So a QCD is a linear map 
$\lambda : \,\ginv\to \Omega^1 (P)$ satisfying 
for all $\theta \in \!\,\ginv$ these two conditions: 
\begin{align*}
             \lambda (\theta^*) &= \lambda (\theta)^* 
             \quad (\mathrm{optional})
             \\ 
             \FWP \big( \lambda (\theta) \big) 
             &= 
             (\lambda \otimes \mathrm{id}_{ \mathcal{A} }) \, 
             \ad (\theta). 
\end{align*}
The second condition is simply 
the covariance of $\lambda$ with respect
to the two right co-actions, 
namely $\FWP$ and $\ad$ respectively, 
co-acting on the appropriate two spaces, namely 
$\Omega^1(P)$ and $\ginv$ respectively. 
This second condition plus the definition 
of the horizontal elements 
in $\Omega(P)$ immediately implies using 
\eqref{ad-ad-relation} that $\lambda (\theta)$ is 
horizontal for all $\theta \in \!\,\ginv$ 
or, in other words, 
$\lambda (\theta) \in \mathfrak{hor}^1 (P) \subset \Omega^1(P)$. 
The first condition is again an optional reality condition. 

Given a QPB, the set of its quantum connections, 
if it is non-empty, is an affine space 
associated to the vector space (which always exists, 
though it may be zero) of
quantum connection displacements. 

The existence of a quantum connection on a QPB 
is a delicate question in general, 
but in this case we can immediately define 
a quantum connection on the QPB $P$ with finite structure 
group and with hodc $\Omega(P)$ defined as above. 
(See \cite{MichoQPB2} for a proof of the existence of 
a quantum connection on a QPB with compact
matrix quantum group, which includes the present case.) 
We simply define 
$
     \tilde{\omega} : \,\ginv\to \Omega^1 (P) 
$
for $\theta \in \,\ginv$ to be 
\begin{equation}
\label{define-omega-tilde}
     \tilde{\omega}  (\theta) := 1_{\mathcal{B}} \otimes \theta.  
\end{equation}
Note that 
$1_{\mathcal{B}} \otimes \theta \in \mathcal{B} \otimes \,\ginv =
\mathcal{D}^0 \otimes \,\ginv \subset \Omega^1 (P)$.  
First, $ \tilde{\omega} $ satisfies 
\eqref{omega-is-star-morphism} by the definition of 
the $ * $-operation as given in Appendix~B. 
And second, \eqref{characteristic-property-of-a-connection} 
is shown in Theorem 14.4 in \cite{Part-II}. 

\subsection{A Right Co-action}

In order to establish notation and a definition 
for the next theorem 
we note that there is a right co-action of 
$\mathcal{A} = \cal{F}(G)$ on $\mathcal{D}$, denoted by 
$$
 _{\mathcal{D}} \Phi : \mathcal{D} \to \mathcal{D} \otimes \mathcal{A}
$$
and defined for $ \varphi \in \mathcal{D} $ 
as the finite sum 
\begin{equation}
\label{define-Coxeter-co-action}
   _{\mathfrak{D} }\Phi (\varphi) := 
   \sum_{g \in G} (g \cdot \varphi) \otimes \delta_g 
   \in \mathfrak{D} \otimes \mathcal{A},
\end{equation}
where $g \cdot \varphi$ denotes the left action 
of $g$ on a classical, complexified 
differential $k$-form $\varphi \in \mathcal{D}^k$.
This left action comes via pull-back from the 
right action of the group $G$ on
the $C^\infty$ manifold $E$.
We also write $\varphi_g := g \cdot \varphi$. 

In the following we also use Sweedler's notation 
\begin{equation}
\label{swee-note}
      _{\mathcal{D}} \Phi (\varphi) = 
      \varphi^{(0)} \otimes \varphi^{(1)} 
      \in \mathcal{D}^k \otimes \mathcal{A} 
      \quad \mathrm{for} \,\,  \varphi \in \mathcal{D}^k = 
      \frak{hor}^k (P).  
\end{equation}

Note that these two expressions 
\eqref{define-Coxeter-co-action} and 
\eqref{swee-note} for 
$_{\mathcal{D}} \Phi (\varphi)$ allow us 
to make these identifications: 
$\varphi^{(0)} = g \cdot \varphi = \varphi_g$ 
and $\varphi^{(1)} = \delta_g$, provided that we include an 
explicit sum over $ g \in G $.

\subsection{The Covariant Derivative}

\begin{definition}
Let $\omega$ be a quantum connection on the QPB $P$ 
with finite structure group $G$ and 
equipped with the 
hodc $\Omega(P)$ as introduced above in \eqref{define-Omega-P}. 
Then we define 
the {\em covariant derivative} $D_\omega$ of the 
quantum connection $\omega$ by 
\begin{align}
\label{define-covariant-derivative}
    D_\omega (\varphi) := 
     d_P \, (\varphi) - (-1)^k \varphi^{(0)} 
     \omega (\pi ( \varphi^{(1)} )) 
     \in \Omega^{k+1}(P) 
\end{align} 
for each $\varphi \in \mathfrak{hor}^k(P)$,  
where we use Sweedler's notation in \eqref{swee-note}. 
\end{definition}

One important result here is that 
$D_\omega : \frak{hor}^k (P) \to \frak{hor}^{k+1} (P)$. 
As we remarked earlier, 
in general $d_P (\varphi)$ is not a horizontal vector 
for $\varphi \in \mathfrak{hor}(P)$ 
or, in other words, the $*$-subalgebra $\frak{hor}(P)$ of 
$\Omega(P)$ is not necessarily invariant under $d_P$, 
the differential of $\Omega(P)$. 
So what is happening in the above  definition of the covariant 
derivative is that a term is being 
subtracted off of $d_P (\varphi)$ 
to give us the horizontal vector $ D_\omega (\varphi) $. 

Continuing with this notation we recall 
the following key result which is 
proved in Theorem~14.4 in \cite{Part-II}. 

\begin{theorem}
\label{D-tilde-omega-is-de-Rham}
The map  $\tilde{\omega} : \,\ginv\to \Omega^1 (P)$ 
defined in equation \eqref{define-omega-tilde} 
is a quantum connection on $P$, and its associated 
covariant derivative $D_{\tilde{\omega}}$ is given by 
$$
         D_{\tilde{\omega}} (\varphi) 
         = d_P  (\varphi) - (-1)^k \varphi^{(0)} \otimes 
         \pi (\varphi^{(1)}) 
$$
for $\varphi \in \mathfrak{hor}^k(P) \cong \mathcal{D}^k$. 

Also, $D_{\tilde{\omega}} = D_{dR}$, 
the complexified de~Rham differential in 
$\mathcal{D} =  \Omega_{dR} (E) \otimes \mathbb{C}$, 
the complexified 
de~Rham exterior differential calculus 
of the smooth manifold~$E$. 
\end{theorem}

Let $\omega = \tilde{\omega} + \lambda$, where $\lambda$ is an 
arbitrary quantum connection displacement (QCD) on $P$. 
So $\omega$ is an arbitrary quantum connection on $P$. 
Then we can compute the covariant derivative of $\omega$ as follows:

\begin{theorem}
For any $\varphi \in \mathcal{D}^k$ we have
\begin{equation}
\label{D-omega-varphi}
        D_{\omega} (\varphi) = D_{\tilde{\omega} + \lambda} (\varphi) = 
         D_{ \tilde{\omega} } (\varphi) 
         + (-1)^k \sum_{s \in S} (\varphi 
         - \varphi_s) \lambda ( \pi(\delta_s) ) \in \cal{D}^{k+1}. 
\end{equation}
\end{theorem}

This formula is Equation~(14.6) in \cite{Part-II} 
and is an immediate precursor to the 
formula for the Dunkl operators, 
which is the case when $k=0$ and the 
QCD $\lambda$ is chosen adequately. 

\subsection{The Quantum Connection Displacement}

Now the quantum connection displacement
(QCD) $\lambda$ is a linear map
$$
    \lambda :  \,\ginv\to \hor^1(P) \cong \mathcal{D}^1
$$
satisfying two conditions.  
But we drop the first condition (namely, that $\lambda$ 
is a $*$-morphism) as being inessential 
to the desired geometric structure. 
Being linear, $\lambda$  is completely determined 
by its values on the basis
elements $\pi (\delta_s)$ of $\ginv$, where $s \in S$. 
So for each $s \in S$ we define 
$$
   \eta_s := \lambda (\pi (\delta_s)) \in \mathcal{D}^1,
$$
the complexification of the space 
of de Rham $1$-forms on $E$. 
The second condition on $\lambda$ for it to be 
a QCD translates directly into
a condition on the $1$-forms $\eta_s \in \mathcal{D}^1$. 
So for all $s \in S$ we must have 
\begin{align}
             \label{second-condition}
             \FWP (\eta_s) 
             &= (\lambda \otimes 
             id_{ \cal{A} } ) \mathrm{ad} (\pi (\delta_s) ) 
             \\
             &= (\lambda \otimes id_{ \cal{A} } )
              (\pi \otimes id_{ \cal{A} } ) \mathrm{ad} (\delta_s) \nonumber  
              \\             
             &= (\lambda \pi \otimes id_{ \cal{A} } ) 
             \sum_{k \in G} \delta_{k s k^{-1} } \otimes \delta_k \nonumber 
             \\
             &= \sum_{k \in G} \eta_{k s k^{-1} } \otimes \delta_k 
             \in \mathcal{D}^1 \otimes \mathcal{A}
             \subset \Omega^1 (P) \otimes \mathcal{A}. \nonumber              
\end{align}
Here we used the identity \eqref{ad-ad-relation} 
as well as the explicit formula for $\mathrm{ad} (\delta_g)$ 
for all $ g \in G $: 
\begin{equation}
\label{ad-delta-g}
      \ad (\delta_g) 
      = \sum_{k \in G} \delta_{k g k^{-1}} \otimes \delta_k.
\end{equation}
For a proof of \eqref{ad-delta-g} see 
equation (13.3) in \cite{Part-II}. 
Note that (\ref{second-condition}) is a condition 
on each conjugacy class of $\S$. 
For example, 
if $G$ is abelian or even if $s$ is in the center of $G$, 
then this condition becomes 
\begin{equation*}
            \FWP (\eta_s) = 
            \sum_{k \in G} \eta_{ s } \otimes \delta_k = 
            \eta_{ s } \otimes \sum_{k \in G} \delta_k = 
            \eta_s \otimes 1_{\mathcal{A}}, 
\end{equation*}
that is, $\eta_s \in  \mathcal{D}^1$ 
is right invariant with respect to the right co-action $ \FWP $. 
So the covariance condition (\ref{second-condition}) 
can be thought of as a generalization of right invariance.

\subsection{Dunkl (Quantum) Connections}

So far we have the general condition on $\lambda$ so that 
$\omega = \tilde{\omega} + \lambda$ is a quantum connection, namely 
that $\lambda$ is a QCD. 
We now would like to add extra conditions  
on $\lambda$ in order to restrict to a class 
of quantum connections that is more closely related to the class
of Dunkl connections as described in \cite{DS}. 
To do this we first suppose from now on the following 
assumption and notations: 
\begin{itemize}
	
	\item Since the $ C^\infty $ manifold $E$ 
	is a covering space of its quotient $ E/G $, 
	it has 
	an open coordinate neighborhood $ U \subset E $ 
	which projects
	diffeomorphically down to the quotient 
	space $ E /  G $
    to a coordinate neighborhood $ U / G $. 
    We fix a set of 
    coordinate functions $ x_{i} $ for 
    $ i= 1, \dots, n = \dim \, E \ge 1  $
    in $ U/G $ and use the same notation for the 
    corresponding local coordinate functions
    defined on $ U $. 
    These coordinates define local vector fields 
    $ \partial/ \partial x_{i}  $ on the (possibly
    non-trivial)
    tangent bundle
    of $ E/G $ and local forms $ d x_{i} $ 
    on the cotangent bundle
    of $ E/G $. 
    In particular, the local forms $ d x_{i} $ are
    exact and thus closed. 
    Also, the local vector fields 
    $ \partial/ \partial x_{i}  $ commute among 
    themselves. 
    From now on we work in such a neighborhood $ U/G $
    with its distinguished coordinates functions 
    $ x_{i} $. 
    In particular, we will use the complexification 
    of these local trivializations 
    so that $ d x_{i} $ is a holomorphic form and
    $ d \overline{x}_{i} $ is the associated 
    anti-holomorphic form. 
    Similarly, we have the complex vector fields
    $ \partial / \partial x_{i} $ and 
    $ \partial / \partial \overline{x}_{i} $. 
    We emphasize that we are not requiring globally 
    defined coordinates on $ E/G $.  
    
\end{itemize}

An alternative setting is to assume that
the $ C^\infty $ manifold $E$ has a trivial tangent bundle 
and use globally defined vector fields 
and forms on it that are associated to a trivialization. 
In that case local coordinates are not needed.  
As an example, we could take $ E $ to be a Lie group 
with $ G $ being any finite subgroup of $ E $ that 
acts by right multiplication on $ E $. 
We plan to present this situation in a future paper, 
since the associated Dunkl operators have quite 
different properties. 

Next is a definition that imposes 
a covariance condition 
on a QCD $ \ddisp $. 
\begin{definition}
\label{define-dunkl-qcd}

Suppose the following: 
\begin{itemize}
\item 
For each $s \in S$ we have a non-zero element 
$\alpha_s \in \mathcal{D}^1 \cong \hor^{1}(P)$ 
such that $ \alpha_s $ depends only on 
the conjugacy class of $ s \in \S $, that is, 
$ \alpha_{g s g^{-1}} = \alpha_{s} $ for all 
$ g \in G $ and $ s \in \S $. 

\item
For each $s \in S$ 
there exists a $C^\infty$ function 
$h_s : E \to \mathbb{C}$ 
such that
$$
      h_{g s g^{-1}} (x) = h_s (x g)
$$
holds for all $x \in E$ and $g \in G$. 

\end{itemize}
(Note that since 
$\S$ is closed under conjugation, $g s g^{-1} \in S$ and so 
$ \alpha_{g s g^{-1}} $ and 
$h_{g s g^{-1}}$ make sense.  
Also $ x g $ denotes the right action of $g$ on $x$.) 
We then define 
$\lambda :  \,\ginv\to \hor^1 (P) \cong \mathcal{D}^1$ 
for  $s \in S$ and $x \in E$ by 
\begin{equation}
\label{formula-dunkl-qcd}
    \lambda( \pi(\delta_s) ) (x) = 
    h_s (x) \, \alpha_s \in \hor^1 (P) \cong \mathcal{D}^1 
\end{equation}
on the basis 
$\{ \pi(\delta_s) ~|~ s \in S \}$ of $\ginv$ 
and then extend linearly to $\ginv$. 
We say that $\lambda$ is a 
{\em Dunkl (quantum) connection displacement} 
or a {\em Dunkl QCD}. 

In this case, we also say that 
$\omega = \tilde{\omega} + \lambda$ 
is a {\em Dunkl (quantum) connection}, where 
$ \tilde{\omega} $ is defined 
in \eqref{define-omega-tilde}.
\end{definition}

We remark that choosing $ h_s \equiv 0$  for every 
$ s \in \S $ shows that $ \ddisp \equiv 0 $ is a Dunkl 
QCD and hence $ \tilde{\omega} $ is a Dunkl connection. 

This generalizes the definition given in 
\cite{Part-II} of a Dunkl connection 
in the context of finite Coxeter groups acting on Euclidean space, 
since there Definition~14.3 of a Dunkl QCD 
refers only to the Coxeter group case with a specific choice 
of the set $ \S $. 
Consequently, Definition~14.4 of a quantum Dunkl connection 
in \cite{Part-II} also only 
refers to the Coxeter case. 
But in this paper we are dealing with an arbitrary finite group. 
The present theory also includes Coxeter groups, but with 
other choices for the set $ \S $. 
For example, we can consider $G = \Sigma_{n}$,   
the symmetric group on $ n \ge 3$ letters 
acting on $ \mathbb{R}^n $ 
by permuting its coordinates  
(which {\em is} a Coxeter group), 
with $ \S $ being the set of all $ k $-cycles for 
$3 \le k \le n$, which does define a bicovariant,
$ * $-covariant fodc on $\Sigma_{n}$.
We remind the reader that the case when $ \S $ is the 
set of all $ 2 $-cycles in $\Sigma_{n}$ 
is a special case of the theory in \cite{DS}. 

Of course, this discussion would be besides the point 
without the following result whose proof is essentially
the same as that of Proposition~14.2 of \cite{Part-II}. 

\begin{theorem}
Every Dunkl QCD is a QCD. 
Consequently, every Dunkl quantum connection is a 
quantum connection.
\end{theorem}

\subsection{Dunkl Gradients}

\noindent 
Using the local coordinates introduced above, we 
write the covariant derivative $ D_\omega $ as a 
{\em Dunkl gradient} in $ U $, namely as  
\begin{equation}
\label{define-Dunkl-gradient}
    D_\omega ( \varphi ) = 
    \sum_{j = 1}^n D_\omega^j ( \varphi ) \, d \, x_j 
\end{equation}
for $ \varphi \in \cal{D}^0 = C^\infty (E) $, 
where the components 
$ D_\omega^j $ 
are called the {\em generalized Dunkl operators} 
associated to the local forms $ d x_{j} $ defined above. 
It follows from \eqref{D-omega-varphi} that  
for $ \varphi \in \cal{D}^0 = C^\infty(E) $ and 
for $ 1 \le j \le n $ the equation in $ U $
\begin{equation}
\label{gen-dunkl-operators}
  D_\omega^j ( \varphi ) = 
  \dfrac{\partial \varphi}{\partial x_j} 
  + \sum_{s \in \S} (\varphi - \varphi_s) \lambda_j (\pi(\delta_s)) 
\end{equation}
(where $ \lambda_j $ is the $ j $-th component of 
$ \lambda $ with respect to the same local forms) 
gives us 
an explicit formula for the generalized Dunkl operator. 
Here $ \partial \varphi / \partial x_j $ means the above
defined local vector field $ \partial / \partial x_j $ 
acting on the $ C^\infty $ function $ \varphi \in \cal{D}^0 $. 

\subsection{How We Build on This Background}

Nonetheless, in order to get interesting results we can impose 
various other conditions, such as those 
in Definition \ref{define-dunkl-qcd}, on the QCD $ \lambda $. 
An example of this is given in the next section. 
In Section~\ref{general-picture-section}
we will introduce a new condition on a QCD
for the case of a QPB with finite structure group.  
This is the central, new concept of this paper. 
Nonetheless, we could take the main result of \cite{DS} 
as motivation for saying that a generalized Dunkl 
operator is the covariant derivative of any quantum 
connection in any QPB, including the case when all the 
spaces are quantum spaces and the structure group is 
a quantum group. 
We take the point of view here that a generalized 
Dunkl operator should be of this type plus some 
extra structure such as a curvature zero condition,  
to name one possibility.

\section{Basic Example: The Dunkl Operators}
\label{sec-3}

The motivating example for this paper is the case 
discussed in \cite{DS},
where one takes $G$ to be a (necessarily finite) Coxeter group 
acting as orthogonal transformations on 
a finite dimensional Euclidean space $V$ over $ \mathbb{R} $. 
In that case one takes
$$
      \S = \{ s \in G ~|~ s \ne e \quad \mathrm{and} \quad s^2 =e \},
$$
which turns out to be the same as 
the set of reflections $\sigma_\alpha$
for $\alpha \in R \subset V$, 
where $ R $ is
the non-empty set of non-zero roots in $ V $
which are used to define $G$. 
The orthogonal map $ \sigma_\alpha : V \to V $ fixes point-wise 
the hyperplane perpendicular to $ \alpha \ne 0 $ 
and maps $ \alpha $ to $ - \alpha $. 
(See \cite{grove-benson} or \cite{hump} for more definitions 
and basic properties of a Coxeter group.) 
One should proceed with a bit of caution 
when comparing the results of
this paper with those of \cite{DS}, since the mapping 
$\alpha \mapsto \sigma_\alpha$
is $2$ to $1$ from the set of roots $R$ (as used in \cite{DS}) 
onto the set $\S$ used here. 
In this case one also lets 
$$
E = V \setminus \bigcup_{ \alpha \in R } H_\alpha,
$$
where $H_\alpha \subset  V$ is the hyperplane of fixed points 
of the reflection $\sigma_\alpha$. 
Finally, one defines the QCD $\ddisp$ by 
\begin{equation}
\label{define-coxeter-ddisp}
        \eta_s (x) = \ddisp (\pi (\delta_s)) (x) = 
        \dfrac{\nu (\alpha)}
        {\langle \alpha , x \rangle} \alpha^{\star}
\end{equation}
where $s = \sigma_\alpha = \sigma_{-\alpha}$, $\alpha \in R$, 
$x \in E$ and $\nu : R \to [0, \infty)$ is a 
$G$-invariant function, meaning that $ \nu $ is constant 
on each orbit of $ G $. 
Then $ \nu $ called a {\em multiplicity function}. 
In particular, $\nu (-\alpha) = \nu (\alpha)$, 
since $\sigma_\alpha (\alpha) = - \alpha$ and so 
$\alpha$ and $-\alpha$ lie in the same $G$-orbit. 
Also $\alpha^{\star}$ is the $1$-form on $E$ defined 
via the inner product with the vector $\alpha$. 
With these identifications we obtain the Dunkl operators 
as originally introduced in \cite{dunkl-1989} and then 
the main result of \cite{DS} is a special case 
of this paper. 

This section gives an example of the 
theory developed in Section~2. 
It also is an example of the new theory in Section~5  
as we shall see.

\section{Cyclic Geometry on $\S$}
\label{sec-4}

Our results for Dunkl operators admit elegant 
and far reaching generalizations 
to geometrical contexts where finite groups act 
on classical or quantum spaces. 
We shall start by extracting the key algebraic property 
of root vectors; 
that property is what enabled us to prove in \cite{DS} 
that the curvature of any Dunkl connection 
associated to a Coxeter group is always zero. 
In this paper
this topic has a auxiliary, though important, 
role. 
But we find it to be of independent interest. 
More details on this are in Appendix~A.  

Let us begin by considering an interesting geometrical 
structure on the space $\S$, 
any subset of a finite group $ G $
closed under conjugation by all the elements of $G$. 
(The motivating example, of course, is the set $\S$ defining a 
$*$-covariant, bicovariant fodc for $\cal{F} (G)$ as described earlier. 
We will come back to this example, but for now we could even take $\S$ 
to be a finite subset of any group, provided that $\S$ is closed 
under conjugation just by elements in $\S$.) 
The canonical {\em flip-over operator} $\sigma$ 
acts naturally on $\S\times\S$. 
Explicitly, by definition $\sigma$ and its inverse $\sigma^{-1}$ 
act on $(g,h) \in \S \times \S$ by 
\begin{equation}
\label{define-flip-over}
\begin{aligned}
\sigma&\colon(g,h)\mapsto (ghg^{-1},g) \in \S \times \S\\  
\sigma^{-1}&\colon (g,h)\mapsto (h, h^{-1}g h) \in \S \times \S.
\end{aligned} 
\end{equation}
For example, if $g$ and $h$ commute, then $\sigma (g,h) = (h,g)$. 
So we can think of $\sigma$ as a twisted interchange. 
Moreover, $\sigma^2 (g,h ) = (g,h)$ 
if and only if $g$ and $h$ commute. 
Another immediate observation is that $\sigma (g,h ) = (g,h)$
if and only if $g =h$. 
When it is convenient, 
the elements $(g,h) \in \S \times \S$ will be identified with the 
tensor products 
$g\otimes h \equiv \delta_g \otimes \delta_h 
\in \cal{A} \otimes \cal{A}$, 
where $\cal{A} = \cal{F} (G)$. 

Accordingly, the finite space $\S\times\S$ splits 
into finite orbits of this action. 
It is obvious from the above formula 
\eqref{define-flip-over}
for $\sigma$ 
that for all pairs in an orbit 
the product of their two entries is the same element in $G$, 
which is then an invariant of the orbit.  
Each orbit consists of 
$n \ge 1$ distinct pairs in $\S \times \S$, say 
\begin{equation}
\label{typical-orbit}
(q_1,q_2), ( q_2, q_3 ), \dots , (q_{n-1}, q_n), (q_n, q_{n+1}), 
\end{equation}
where each pair starting with $ ( q_2, q_3 )$ is the image under 
$\sigma^{-1}$ of the previous pair and 
$\sigma^{-1}(q_n, q_{n+1}) = (q_1,q_2)$. 
This last equality implies that $q_{n+1} = q_1$. 
By the above remarks we see that there are always orbits with 
$n=1$ and, if $ \mathrm{card}(G) \ge 2 $, for $n=2$ as well. 

We claim that 
$q_1,\dots, q_n$ in the orbit (\ref{typical-orbit}) 
are $n$ distinct elements of $\S$. 
Suppose to the contrary that there exists 
$k,l$ with $1 \le k < l \le  n$ 
such that $q_k = q_l$. 
Then by the invariance of the product along an orbit we have that 
$q_k  q_{k+1} = q_l  q_{l+1}$. 
This together with $q_k = q_l$ implies that $q_{k+1} =  q_{l+1}$ 
and therefore $(q_k , q_{k+1} ) = ( q_l , q_{l+1} )$, 
that is to say, the pairs 
in (\ref{typical-orbit}) are not all distinct.  
But this contradicts that (\ref{typical-orbit}) 
is an orbit with $n$ distinct pairs. 
Therefore the elements $q_1,\dots, q_n$ of $\S$ must be distinct. 
The fact that an orbit $ \cal{O} $ lies in $\S \times \S$ implies 
that $ \mathrm{card}(\cal{O}) \le k^2 $, where 
$ k = \mathrm{card}(\S) $. 
This is a very weak upper bound, since the result of 
this paragraph tells us that $ \mathrm{card}(\cal{O}) \le k $. 

In other words, every orbit of $ \sigma $ 
in $\S\times \S$ projects into 
a certain ordered sequence $(q_1, q_2,  \dots, q_n)$ of
$n$ distinct elements of $G$, where $n$ is 
the cardinality of the orbit. 
Since the same orbit is also associated with 
the ordered sequences  $(q_2, q_3, \dots, q_n, q_1)$, 
$(q_3, \dots, q_n, q_1, q_2)$ and so forth 
for all cyclic permutations, 
we then define 
a {\em cyclic set} to be the set whose
elements are all of those ordered sequences 
arising from one orbit and which are so related 
by cyclic permutations. 
If we now consider $\S$ together with the family 
$\cal{T}$ of all of these cyclic sets, 
we get a kind of discrete geometry on $\S$ 
with the cyclic sets being 
interpretable as `lines' and the elements of $\S$ as `points'. 
Cyclic sets are preferably called {\em cyclic lines} in order to 
emphasize this geometrical interpretation. 
Notice that the cyclic line associated to $(q_1, q_2,  \dots, q_n)$ 
uniquely determines the orbit (\ref{typical-orbit}), since 
the second entries in those pairs are  $(q_2, q_3, \dots, q_n, q_1)$. 
We write the cyclic line $\ell$ associated to the ordered sequence 
$(q_1, q_2,  \dots, q_n)$ as 
$\ell \leftrightarrow (q_1, q_2,  \dots, q_n)$. 
This means that each orbit $ \cal{O} $ 
of $ \sigma $ is also associated with a unique cyclic line 
$\ell \leftrightarrow (q_1, q_2,  \dots, q_n)$. 
We write this relation as 
$\cal{O} \leftrightarrow (q_1, q_2,  \dots, q_n)$.

\begin{definition}
\label{define-group-invariant}
We define the {\em group invariant} of an orbit 
$ \cal{O} \leftrightarrow (q_1, q_2, \dots, q_n)$ of $ \sigma $ 
for $ n \ge 2 $ by
$$
   \mathrm{Inv} (\cal{O}) := q_k q_{k+1} \in G,  
$$
where $ k $ is any integer satisfying $1 \le k \le n$. 
As usual we let $ q_{n+1} = q_1 $. 

For $ n = 1 $ we define $  \mathrm{Inv} (\cal{O}) := q_1^2$. 
\end{definition}

By our remarks above we have that 
$  \mathrm{Inv} (\cal{O})$ is well-defined, 
that is, its value does not depend on 
the particular choice of $ k $. 
Also, it does not change value under cyclic 
permutations of $ (q_1, q_2, \dots, q_n) $ and so
is a function of the orbit~$ \cal{O} $. 

For their own independent interest these 
{\em cyclic geometries} are discussed 
in more detail in Appendix~A, 
where we have in particular defined the concepts of a 
{\em cyclic space} and its {\em representations}. 

We also wish to note that the flip-over operator 
$ \sigma $ induces braidings. 
(Compare with the Woronowicz braid operator 
in \eqref{Woronowicz-sigma-quadratic-formula}.) 
We recall that the {\em braid group} $ B_{n} $ is defined 
for each integer $ n \ge 2 $ as the group generated 
by $ \{ g_{1}, \dots , g_{n-1}  \} $ with 
the relations $ g_{j} g_{k} = g_{k} g_{j} $ 
whenever $ | j-k | \ge 2 $ and 
$ g_{j} g_{j+1} g_{j} = g_{j+1} g_{j} g_{j+1} $ 
for $ j=1, \dots, n-2 $. 
We define the bijection 
$ \sigma_{12} $ of $ \S \times \S \times \S $ 
as $ \sigma $ acting on the first two factors, 
that is $ \sigma_{12} (a,b,c) := (\sigma(a,b), c) $. 
Here $ a,b,c \in \S $. 
Similarly, we define 
$ \sigma_{23} (a,b,c) := (a,\sigma(b,c)) $. 
One readily checks that the braid equation holds, 
namely 
\begin{equation}
\label{braid-equation}
   \sigma_{12}\sigma_{23}\sigma_{12}= \sigma_{23}\sigma_{12}\sigma_{23}.
\end{equation}
Next, letting $ \S^{n} = \S \times \cdots \times \S  $ with 
$ n $ factors, we define the bijection $ \tau_{k} $ 
of $ \S^{n} $ for $ 1 \le k < n $ to be $ \sigma $ 
acting on the factors $ k $ and $ k+1 $ and to be the 
identity on the remaining factors. 
This defines an action of the braid group $ B_{n} $ 
on $ \S^{n} $ for every integer $ n \ge 2 $, where 
we use \eqref{braid-equation} when $ n \ge 3 $.   
Since $ \Gamma \cong \mathcal{A} [\S] $, the free 
$ \mathcal{A} $-module with basis $ \S $, 
and $ \ginv \cong \mathbb{C}[\S] $, the vector space 
with basis $ \S $, we also have induced representations 
of $ B_{n} $ on 
$ \Gamma^{\otimes n} = 
\Gamma \otimes_{\mathcal{A}} \cdots 
\otimes_{\mathcal{A}} \Gamma \cong \mathcal{A} [\S^n]$ 
and on 
$\! \ginv^{\otimes n} \!=\! 
\ginv \otimes_{\mathbb{C}} \cdots 
\otimes_{\mathbb{C}} \ginv \! \cong \!
\mathbb{C}[\S^n]$,   
coming from the action of $ B_{n} $ on the space $ \S^n $.

\section{A General Picture}
\label{general-picture-section}

Let us incorporate all this into the context of 
quantum principal bundles. 
We shall assume that 
a quantum principal bundle $P = (\cal{B}, \cal{A}, F)$ 
with a classical finite structure group $G$ is given, 
that is, $\cal{A} = \cal{F}(G)$. 
We shall also assume here that 
a bicovariant, $*$-covariant  fodc $\Gamma \neq 0$ on $G$ is fixed, 
and as explained in Section~2 it is associated to a non-empty set 
$\S\subseteq G\setminus\{\e\}$ which is 
invariant under all conjugations by elements of $G$ 
and under the inverse operation on $G$. 
This is the set $\S$ which we will consider from now on. 
We let $\sigma$ and $\cal{T}$ be as above for this particular subset $S$ 
of the finite group $G$. 
In addition, we shall assume that $ \cal{B} $ has been 
extended to the higher-order 
differential calculus (hodc) $\Omega(P)$ defined 
in \eqref{define-Omega-P}. 
Then $\Omega(P)$ has its connection independent, 
horizontal forms $\hor(P)$ as defined 
in \eqref{define-horizontal-forms}.

\subsection{A Regular `Initial' Connection}

We also assume the existence of 
a regular `initial' connection $\tilde{\omega}$ on $\Omega(P)$ that 
has a covariant derivative 
$D=D_{\tilde{\omega}}\colon \hor(P)\rightarrow\hor(P)$ 
which acts as a hermitian, degree~$+1$ map satisfying 
the graded Leibniz rule and that also intertwines 
the right co-action $ \FWP $ of $\cal{A}$ 
on $ \hor(P) $ or, colloquially speaking, 
that "intertwines the action of $G$". 
Also, a connection $\omega$ said to be {\em regular} 
if for all $\varphi \in \frak{hor}^k (P)$  
and all $\vartheta \in \ginv$ we have 
\begin{equation}
\label{define-regular-connection}
    \omega(\vartheta) \varphi = 
    (-1)^k \varphi^{(0)} \omega (\vartheta \circ \varphi^{(1)} ) 
    \in \frak{hor}^{k+1} (P),
\end{equation}
where we have used Sweedler's notation for 
$\FWP ( \varphi ) 
= \varphi^{(0)} \otimes \varphi^{(1)} 
\in \frak{hor}^k (P) \otimes \cal{A}$. 

\begin{remark}
We remark that a connection $ \omega $ is regular if and only if 
its covariant derivative $D_\omega$ satisfies the graded 
Leibniz rule. 
This is proved in \cite{MichoQPB3}. 
Also see Theorem 12.14 in \cite{Part-II} where it is shown that
regularity of $ \omega $ implies the graded Leibniz rule 
for $ D_\omega $. 
\end{remark}

\begin{theorem}
The connection $ \tilde{\omega} $ defined in 
\eqref{define-omega-tilde} is regular. 
\end{theorem}
\begin{proof}
By the previous remark this is equivalent to showing that 
$ D_{\tilde{\omega}} $ satisfies the  graded Leibniz rule. 
But by Theorem \ref{D-tilde-omega-is-de-Rham} 
$ D_{\tilde{\omega}} = D_{dR}$, 
the complexified de Rham differential, which indeed does
satisfy the graded Leibniz rule. 
This is an abstract way to prove this result. 

A computational proof starts with 
$ \varphi \in \hor^k (P) $ and uses 
Sweedler's notation for the right co-action $ _{\cal{D}}\Phi $ 
as follows: 
$$
_{\cal{D}}\Phi(\varphi) = \varphi^{(0)} \otimes \varphi^{(1)} 
\in \hor^k (P) \otimes \cal{A}.  
$$
Therefore for $\vartheta \in \ginv$ and 
$ \varphi \in \hor^k (P) $ we obtain 
\begin{align*}
\tilde{\omega}(\vartheta) \, \varphi 
&= (1_{\cal{B}} \otimes \vartheta) (\varphi \otimes 1_{A} ) 
= (-1)^k 1_{\cal{B}} \, \varphi^{(0)} 
\otimes (\vartheta \circ \varphi^{(1)} ) \, 1_{\cal{A}} 
\\
&= (-1)^k \varphi^{(0)} 
\otimes (\vartheta \circ \varphi^{(1)} ) 
= (-1)^k (\varphi^{(0)} \otimes 1_{\cal{A}}) 
          (1_{\cal{B}} \otimes \vartheta \circ \varphi^{(1)}) 
\\
&=(-1)^k \varphi^{(0)}\,\tilde{\omega} (\vartheta \circ \varphi^{(1)} )
\end{align*}
by using the definitions of $ \tilde{\omega} $ 
in \eqref{define-omega-tilde} and 
of the product in $ \Omega (P) $ 
as given in Appendix~B. 
We also used the isomorphism  
$ \cal{D} \cong  \cal{D} \otimes 1_{\cal{A}} $. 
\end{proof}

\subsection{Cyclic Dunkl Displacements}

In Definition \ref{define-dunkl-qcd} 
we defined a Dunkl QCD. 
The next definition is a similar, but different, concept. 

\begin{definition}
\label{define-cyclic-dunkl}
A {\em cyclic Dunkl displacement} 
with respect to a given (that is, `initial') 
regular connection $\tilde{\omega}$ 
on a quantum principal bundle $P$ with a finite 
structure group $ G $
is a set function 
$\ddisp\colon \S\rightarrow \hor^1(P)$, where $ \S \subset G$ 
has a cyclic geometry (with respect 
to the flip-over operator $ \sigma $) 
whose set of lines is  $\cal{T}$ and  
which satisfies the following three properties:  

\vskip 0.2cm \noindent 
\bla{i} Cyclic Property: 
\begin{equation}
\label{cyclic-property}
\ddisp (\ell) := 
\ddisp(s_1)\ddisp(s_2) + 
\cdots 
+ \ddisp(s_{n-1})\ddisp(s_n)+\ddisp(s_n)\ddisp(s_1) 
= 0 \in \frak{hor}^2 (P) 
\end{equation}
for every cyclic line 
$\ell \leftrightarrow (s_1, \dots, s_n) \in \cal{T}$, 
where $s_1, \dots, s_n \in \S$. 
Notice that this condition is invariant under cyclic 
permutations of the chosen ordered sequence, 
that is, it depends only on the cyclic line  $ \ell $ 
defined by $(s_1, \dots, s_n)$. 

\vskip 0.2cm \noindent 
\bla{ii} Covariance Property: 
$$
F^\wedge_g(\ddisp(s)) = \ddisp(gsg^{-1}) \in \hor^1(P)
$$
for every $g\in G$ and $s\in \S$. 
Here $F^\wedge_g=(\id\otimes g)F^\wedge$ 
and we identify points $g \in G$ with their 
associated characters $\chi_g$ on $\cal{A}$, 
where $\chi_g (f) := f (g)$ for all $f \in \cal{A}$. 
So, $\chi_g : \cal{A} \to \mathbb{C}$ is a 
unital, multiplicative, $*$-morphism. 

Here 
$F^\wedge := (\mathrm{id} \otimes p_0) \, \FWP : 
\Omega(P) \to \Omega(P) \otimes \cal{A}  $, 
where $ \FWP : \Omega(P) \to \Omega(P) \otimes \Gamma^\wedge  $ 
is the right co-action of $\Gamma^\wedge  $ on $\Omega(P)$ 
as introduced earlier. 
Also,  
$p_0: \Gamma^\wedge \to \cal{A}$ is the projection 
that maps all homogeneous elements 
of positive degree to $0$ and is the identity 
on all homogeneous elements of degree $0$. 
It follows that $F^\wedge$ is a right co-action 
of $\cal{A}$ on $ \Omega(P) $. 
Furthermore, $F^\wedge_g : \hor^1(P) \to \hor^1(P)$. 
For more details, such as why 
$F^\wedge_g[\ddisp(s)] \in \hor^1(P)$ holds, 
see \cite{MichoQPB2}.  

The reader should be very aware that $\FWP$ and $F^\wedge$ 
are not the same, but by definition are related by 
$F^\wedge = (\mathrm{id} \otimes p_0)  \, \FWP$. 
It turns out that the map $ _{\cal{D}}\Phi $ defined in
\eqref{define-Coxeter-co-action} 
is equal to $F^\wedge$ restricted to 
$ \cal{D} = \hor(P) \subset \Omega(P) $.

\vskip 0.2cm \noindent 
\bla{iii} Closed-ness Property: 
$$
D (\ddisp(s) ) = D_{\tilde{\omega}} (\ddisp(s) ) =0 \in  \hor^2(P)
$$  
for every $s\in\S$. 
This property depends on the choice of the
`initial' quantum connection $\tilde{\omega}$. 
When $\ddisp$ is extended linearly to $L(\S) \cong \ginv$, as 
discussed immediately below, 
then we will have $D_{\tilde{\omega}} (\ddisp( [s] ) ) =0$ 
for all the basis elements $[s]$ of $\ginv$ and therefore 
$D_{\tilde{\omega}} \ddisp = 0$, 
the zero map $\ginv \to \frak{hor}^2 (P)$. 
Notice that in this definition this is the only property 
which refers to the `initial' 
connection~$ \tilde{\omega} $. 
\end{definition}

\begin{remark}
We remark that the constant function $ \ddisp \equiv 0 $ 
is a cyclic Dunkl displacement. 
\end{remark}

\begin{remark} 
As we shall see shortly, every cyclic Dunkl 
displacement is indeed a QCD.    
It is worth observing that our main `cyclicity' condition 
\eqref{cyclic-property} is 
non-linear in $ \ddisp $.  
It picks out an interesting intersection of a quadratic conic 
with the vector space of all Dunkl QCD's. 
\end{remark}

The cyclic property \bla{i} says that $\ddisp$ is a 
representation of the cyclic space $(\S, \cal{T})$ in 
the $*$-algebra $\hor(P)$. 
(See Appendix~A for the appropriate definitions.) 
Since $S$ labels the basis $\{ \delta_s ~|~ s \in \S \}$ of $\ginv$,  
any set function $\ddisp\colon\S\rightarrow \hor(P)$ admits, 
by linearity, a unique linear extension to 
$\ddisp\colon L(\S) \cong\ginv\rightarrow \hor^1(P)$ 
(still denoted as $\ddisp$), 
where $L(\S)$ is the abstract vector space whose elements are formal 
linear combination of elements in $\S$ with complex coefficients. 
This extended $\ddisp$ itself extends further 
to a unique unital algebra morphism 
$\ddisp^{\otimes}\colon\ginv^\otimes\rightarrow \hor(P)$, where 
$ \ginv^\otimes $ is the space of left invariant elements in 
the tensor algebra 
$\Gamma^\otimes= \bigoplus_{k=0}^\infty \, \Gamma^{\otimes k}$. 
(See Definition~7.2 and Exercise~7.4 in \cite{Part-II} for 
the definition and properties of the left co-action
of $ \cal{A} $ on $ \Gamma^\otimes $.) 
Here $\Gamma^{\otimes k} = \Gamma \otimes \cdots \otimes \Gamma$ 
with $k$ factors for $ k \ge 1 $ and 
$ \Gamma^{\otimes 0} = \cal{A} $. 
It turns out that 
$ \ginv^{\otimes k} := 
(\Gamma \otimes \cdots \otimes \Gamma)_{\mathrm{inv}}
= \ginv \otimes \cdots \otimes \ginv$ for $ k \ge 1 $ 
and that 
$ \ginv^{\otimes 0} = \mathbb{C} $. 
Explicitly, 
$ \ddisp^{\otimes k} (\theta_1 \otimes \cdots \otimes \theta_k)
:= \ddisp (\theta_1) \otimes \cdots \otimes \ddisp(\theta_k)
$ 
for $ k \ge 1 $, where 
$ \theta_1, \dots , \theta_k \in \ginv $.
Also, we put 
$ \ddisp^{\otimes 0} (\alpha)  := \alpha \, 1_{\cal{B}} $ 
for $ \alpha \in \mathbb{C} $, since this must be 
an identity preserving map. 

The next result is an essential step in this theory. 
\begin{theorem}
\label{cyclic-dunkl-qcd}
Every cyclic Dunkl displacement is a QCD.
\end{theorem}
\begin{proof}
We have to show that any cyclic Dunkl displacement 
$$
\ddisp\colon L(\S) \cong\ginv\rightarrow \hor^1(P)
$$ 
is covariant with respect to the right co-actions of $\cal{A}$ 
on $\ginv$ and $\hor^1 (P)$. 
These right co-actions are $ \ad $ and $ \FWP $, respectively. 
Of course, this is a consequence of only the Covariance Property 
in Definition \ref{define-cyclic-dunkl}. 

Since $ \ddisp (s) $ is horizontal and 
$ \{ \delta_k ~|~ k \in G \} $ is a basis
of $ \cal{A} $, we have
$$
   \FWP (  \ddisp (s) ) = F^\wedge (  \ddisp (s) ) =
   \sum_{k \in G} \omega_k \otimes \delta_k 
   \in \hor^1(P) \otimes \cal{A}
$$
for unique elements $ \omega_k \in \hor^1(P) $. 
Then by the Covariance Property we obtain
$$
    \ddisp ( g s g^{-1}) = F^\wedge_g ( \ddisp(s) )  
    = (\id\otimes g)F^\wedge ( \ddisp(s) ) 
    =  \sum_{k \in G} \omega_k \otimes \delta_k (g) 
    = \omega_g. 
$$

So we conclude that 
\begin{equation}
\label{F-hat-of-ddisp}
   \FWP (  \ddisp (s) ) 
  = \sum_{k \in G} \ddisp ( k s k^{-1}) \otimes \delta_k.  
\end{equation}

On the other hand 
$$
    ( \ddisp \otimes \mathrm{id}_{ \cal{A} }) \, \ad ( s) = 
    ( \ddisp \otimes \mathrm{id}_{ \cal{A} }) \,  
    \sum_{k \in G} [ k s k^{-1} ]\otimes \delta_k 
    = \sum_{k \in G} \ddisp ( k s k^{-1}) \otimes \delta_k, 
$$
using our standard notation conventions and the 
identities \eqref{ad-ad-relation} and \eqref{ad-delta-g}.
This proves the desired covariance.  
\end{proof}

\subsection{Cyclic Dunkl Connections}

This theorem allows us to define a new generalization 
of a Dunkl connection as was originally defined in \cite{DS}. 
This is a central aspect of this paper. 
\begin{definition}

Let $\ddisp$ be a cyclic Dunkl displacement with respect to 
the `initial' regular connection $\tilde{\omega}$ 
as given above on the QPB $ P $ with finite structure group. 
By Theorem~\ref{cyclic-dunkl-qcd} 
$\omega:= \tilde{\omega} + \ddisp$ is a quantum 
connection (dropping the optional hermitian condition).
Then we say that $ \omega $ is a {\em cyclic Dunkl connection} 
(with respect to $\tilde{\omega}$). 

\end{definition}

In particular, since $ \ddisp \equiv 0 $ is 
a cyclic Dunkl displacement with respect to $\tilde{\omega}$, 
it follows by this definition 
that $\tilde{\omega}$ is a cyclic Dunkl connection 
(with respect to itself).  

Another crucial fact is established in the next result.
\begin{theorem}
Let $ \ddisp : \ginv \to \hor^1 (P) $ be defined by 
formula \eqref{formula-dunkl-qcd}, where the non-zero 
co-vectors satisfy   $ \alpha_{g s g^{-1}} = \alpha_{s} $ 
for all $ g \in G $ and $ s \in \S $. 
Then $ \ddisp $ is a Dunkl QCD if and only if $ \ddisp $ 
satisfies the Covariance Property (ii). 
\end{theorem}
\begin{proof}
On the one hand we calculate 
$$
   \ddisp( g s g^{-1}) (x) 
   = \ddisp( \pi (\delta_{ g s g^{-1} }) ) (x) 
   = h_{g s g^{-1}} (x) \alpha_{g s g^{-1}} 
$$
for all $ x \in E $, $ g \in G $ and $ s \in \S $. 
On the other hand, we have 
\begin{align*}
&F_{g}^{\wedge} (\ddisp(s)) 
= (\mathrm{id} \otimes g) F^{\wedge} ( \ddisp(s) )
= (\mathrm{id} \otimes g) 
\sum_{h \in G} h \cdot ( \ddisp(s) ) \otimes \delta_{h} 
\\
&= 
\sum_{h \in G} h \cdot ( \ddisp(s) ) \otimes \delta_{h} (g) 
= g \cdot ( \ddisp(s) ) \otimes 1 
\cong g \cdot ( \ddisp(s) ). 
\end{align*} 
Evaluating at $ x \in E $ we see that 
$$
F_{g}^{\wedge} (\ddisp(s)) (x) =  g \cdot ( \ddisp(s) ) (x) 
=  \ddisp(s) ( x g) = h_{s} (x g) \, \alpha_s. 
$$
Now we assume that $ \ddisp $ is a Dunkl QCD. 
So we use  
Definition~\ref{define-dunkl-qcd} to obtain that 
$$
 \ddisp( g s g^{-1}) (x) 
 = h_{g s g^{-1}} (x) \, \alpha_{g s g^{-1}} 
 = h_{s} (x g) \, \alpha_{s} 
 = F_{g}^{\wedge} (\ddisp(s)) (x). 
$$
And this shows that $ \ddisp $ 
satisfies the Covariance Property (ii). 

Conversely, suppose that 
the Covariance Property (ii) holds for $ \ddisp $. 
Then from the calculations above we get that
$$
 h_{g s g^{-1}} (x) \alpha_{g s g^{-1}} 
 = h_{s} (x g) \, \alpha_s 
$$
holds for all $ x \in E $, $ g \in G $ and $ s \in \S $. 
By the hypothesis on the $ \alpha_{s} $'s we have 
$$
 h_{g s g^{-1}} (x) \, \alpha_{s} 
 = h_{s} (x g) \, \alpha_s. 
$$
Since $ \alpha_{s} \ne 0 $ this implies 
$ h_{g s g^{-1}} (x) = h_{s} (x g) $, 
proving $ \ddisp $ is a Dunkl QCD.  
\end{proof}

\subsection{The Fodc of $ G $}

Next we present a more detailed description of the fodc on $G$. 
The following rather nice identity 
\begin{equation}\label{sum-S}
[\e]=-\sum_{s\in\S}[s]
\end{equation}
holds for every 
bicovariant, $*$-covariant fodc on $\cal{A} = \cal{F} (G)$. 
(See Section~13.4 and especially Corollary 13.2 
in \cite{Part-II}.) 
Recall that $[s] = \pi (s) = \pi (\delta_s)$  
and that the identity element 
$\e \in G$ also denotes its associated 
character $f \mapsto f(\e)$ for $f\in \cal{A}$, 
which is also the co-unit for the Hopf algebra $\cal{A}$. 

It is also worth recalling that 
$[\e] \in \Gamma$ is invariant under the 
right adjoint co-action of $\cal{A}$. 
We now prove this. 
Using $\e = \delta_e\in\cal{A}$ 
where $e\in G$ is the identity element in $G$ 
and \eqref{ad-delta-g}, 
we get 
\begin{equation*}
  \ad (\e) = \ad (\delta_e) 
  = \sum_{k \in G} \delta_{k e k^{-1}} \otimes \delta_k 
  = \sum_{k \in G} \delta_{ e } \otimes \delta_k 
  = \delta_{ e } \otimes \sum_{k \in G}  \delta_k 
  = \e \otimes 1_{\cal{A}}. 
\end{equation*}
Next, using this and the identity \eqref{ad-ad-relation}, 
it follows that 
\begin{equation}
\label{e-right-ad-invariant}
        \ad (  [ \e ] )
        \!=\! \ad ( \pi (\e) ) \!=\! ( \pi  \otimes \mathrm{id} ) \ad (\e)
        \!=\!  ( \pi  \otimes \mathrm{id} ) ( \e \otimes 1_{\cal{A}} ) 
        \!=\! \pi ( \e ) \otimes 1_{\cal{A}}
        \!=\! [ \e ] \otimes 1_{\cal{A}}, 
\end{equation}
thereby showing that $[ \e ]$ is right adjoint invariant 
as claimed. 

We also remark that for every $a\in \cal{A}$ we have 
\begin{equation}
\label{e-circle-a}
[\e]\circ a=\e(a)[\e]-[a],
\end{equation}
as follows from the definition \eqref{define-circ-action} 
when we simplify after putting $b = \delta_e$ and $c=a$ there. 
In particular, the identity 
$ \delta_e a = \e (a) \delta_e $ is used. 

Thus $[\e]$ is a kind of "vacuum state", that is, it is   
a right adjoint invariant, cyclic vector for the right $\cal{A}$-module $\ginv$. 

Next we note that the derivation 
$d\colon \cal{A}\rightarrow \Gamma$ 
of the fodc satisfies 
\begin{align*}
d a &= a^{(1)} \pi (a^{(2)}) 
\\
&= a^{(1)} [ a^{(2)} ] 
\\
&= a^{(1)} \bigl( \e( a^{(2)}) [\e] - [\e] \circ  a^{(2)} \bigr) 
\\
&= a^{(1)}  \e( a^{(2)}) [\e] - a^{(1)} \bigl( [\e] \circ  a^{(2)} \bigr) 
\\
&= a [\e] -   [\e] a. 
\end{align*}
Here we used the identity \eqref{diff-id} for $d a$, 
equations \eqref{e-circle-a} 
and \eqref{curious-identity} as well as 
the Hopf algebra identity 
$a = a^{(1)} \e ( a^{(2)} )$ 
for the co-unit $\e$. 
This result, written as 
\begin{equation}
\label{d-as-anti-inner-derivation} 
d a = [- \e]  a - a [-\e], 
\end{equation}
is the same sort of commutation relation as we saw earlier 
in \eqref{d-is-inner} 
since it measures the difference between the left and right 
$\cal{A}$-module structures acting on the element $[-\e] \in \ginv$ 
in the $\cal{A}$-bimodule $\Gamma$. 
Again, we say that $d$ is an {\em inner derivation}, 
but now with respect to $[-\e]$, which also can be
considered as a "vacuum state". 

As noted above (taking $ 1 $ to mean $ 1_{\cal{A}} $) 
we can put
$q = 1 - \e = 1 - \delta_e \in\cal{A}$ and define
$\tau := \pi(q) = [ q ] = [ -\e ] \ne 0$. 
Hence 
$(q - 1) \ker (\e) = -\delta_e \ker (\e) = \{ 0 \} \subset \cal{R}$ 
and $q \in \ker (\e)$. 
We also claim that $q \notin \cal{R}$.  
This is so 
since for any $s \in\S$ we have 
$$
q(s) = (1 - \delta_e) (s) = 1 -\delta_e(s) = 1 - 0 =1 \neq 0. 
$$
Here we used that $s \neq e$ for all $s\in\S$. 
But $\cal{R}$ by definition is the ideal in $\ker (\e)$ of those 
functions that annihilate $\S$. 
So $q \notin \cal{R}$ as claimed. 
This shows that this differential $d$ satisfies the general properties
that gave us \eqref{d-is-inner}. 
Even though the identity
\eqref{d-as-anti-inner-derivation} 
seems to be at odds with \eqref{d-of-delta-g}, 
both are identities for $d : \cal{A} \to \Gamma$. 
All of this concerns the fodc $d\colon \cal{A}\rightarrow \Gamma$. 
We will next investigate how this analysis can be extended to 
an hodc that extends this fodc.

\subsection{Extending to the Hodc}

We start out by remarking that $\sigma$, 
the Woronowicz braid operator, satisfies 
\begin{equation}
\label{slight-switch}
\sigma(\vartheta\otimes [\e])=[\e]\otimes \vartheta
\end{equation}
for every $\vartheta\in\ginv$, the space of left invariant elements 
of $\Gamma$ with respect to its canonical left co-action by $\cal{A}$. 
To see that this is so 
we note that the canonical right co-action of 
$\cal{A}$ on $[\e] \in \ginv$ 
is equal to the right adjoint co-action of $\cal{A}$ on $[\e]$. 
This is the content of diagram (6.22) in \cite{Part-II}.
Since $[\e]$ is right adjoint invariant, 
as shown in \eqref{e-right-ad-invariant}, it follows that 
$[\e]$ is right canonical invariant.  
However, 
$\sigma(\omega_1 \otimes \omega_2) = \omega_2 \otimes \omega_1$ 
provided that $\omega_1$ is left invariant 
and $\omega_2$ is right invariant
with respect to the canonical left and right co-actions, respectively.  
(See Section~7.2 in \cite{Part-II}.)  
And thus we have proved \eqref{slight-switch}. 
Then, by taking the special case $\vartheta = [\e] \in \ginv$, we obtain 
\begin{equation}
\label{sigma-of-e-tensor-e}
\sigma([\e] \otimes [\e])=[\e]\otimes [\e]. 
\end{equation}

The action of the Woronowicz braid operator 
$\sigma : \ginv^{\otimes 2} \to  \ginv^{\otimes 2}$ 
is calculated for $g,h \in S$ by using 
\eqref{sigma-eta-vartheta}, \eqref{ad-ad-relation} 
and \eqref{ad-delta-g} to be
\begin{equation}
\label{Woronowicz-sigma-quadratic-formula}
\sigma ([g] \otimes [h]) 
= \sum_{k \in G} [ k h k^{-1} ] \otimes ([g] \circ k) 
= [g h g^{-1}] \otimes [g]. 
\end{equation}
Here we also used $[g] \circ k = 0$ for $k \neq g$ 
and $[g] \circ g = [g]$, which follow 
from the definition \eqref{define-circ-action}. 
Notice that the set $\{ [g] \otimes [h] \, | \, g,h \in \S \}$ is a basis 
of $\ginv^{\otimes 2}$. 
So the Woronowicz braid operator corresponds to the permutation 
of the finite set $\S \times \S$ known as 
the flip-over operator in \eqref{define-flip-over}. 
This justifies using the same notation $\sigma$ 
for both of them. 

There is a very interesting characterization 
of the cyclic property 
in terms of a natural higher-order differential calculus on $G$, 
which we are going to describe now. 
Let us define the algebra $\Gamma^{\sim}$ 
as the quotient algebra of 
the tensor algebra 
$\Gamma^{\otimes} = \bigoplus_{k=0}^\infty \, \Gamma^{\otimes k}$ 
over the ideal $\cal{ I}_{\sigma \mathrm{inv} }$ in $\Gamma^{\otimes}$ 
generated by $\ker(\id-\sigma)$, the vector space 
of all the $\sigma$-invariant 
elements in the degree-$2$ subspace $\Gamma^{\otimes 2}$, 
where $\id$ means the identity map of $\Gamma^{\otimes 2}$. 
An ideal, such as  $\cal{ I}_{\sigma \mathrm{inv} }$, 
generated by elements of degree~$ 2 $ is called 
a {\em quadratic ideal}. 
In this quotient algebra $[\e]^2=0$, 
because of \eqref{sigma-of-e-tensor-e}. 
Thus we can consistently extend the differential 
$d\colon\cal{A}\rightarrow\Gamma$ to 
$d\colon\Gamma^{\sim}\rightarrow \Gamma^{\sim}$ 
by setting 
$$
d \psi := -\bigl([\e]\psi-(-1)^{\partial \psi}\psi[\e]). 
$$
(In general, $ \partial \psi $ denotes the degree of  
a homogeneous element $ \psi $.) 
This says the extended differential is a graded inner derivation 
with respect to the element $[-\e]$. 
Our construction is covariant, and so we have that 
$\Gamma^{\sim}\leftrightarrow\cal{A}\otimes \ginv^{\sim}$ 
where we define 
$$
  \ginv^{\sim} := \ginv^{\otimes} / \cal{ I}_{\sigma \mathrm{inv} }. 
$$
So, any bicovariant, $ * $-covariant fodc 
$ d : \cal{A} \to \Gamma $ has a functorially associated quadratic 
hodc $ \Gamma^\sim $.

\begin{remark} 
It is easy to see from the covariance and *-interrelation 
properties of the flip-over operator $\sigma$,  that this 
calculus has both a compatible bicovariance structure and 
compatible $*$-operations. 
Moreover, 
the calculus admits (a necessarily unique) extension of the 
coproduct to 
$\phi^\sim\colon \Gamma^\sim\rightarrow 
\Gamma^\sim \otimes \Gamma^\sim$, 
and therefore it is an acceptable hodc for the 
structure group $G$ when considering quantum principal 
$G$-bundles. 
Indeed, if we consider a generic quadratic relation 
represented by 
$$
\sum_{\vartheta\eta}\vartheta\otimes \eta 
\in  \ker (\id - \sigma), 
$$ 
then a straightforward calculation shows  
\begin{multline*}
0=\sum_{\vartheta\eta}\vartheta\eta
\mapsto \sum_{\vartheta\eta}\vartheta^{(0)}\eta^{(0)}\otimes \vartheta^{(1)}\eta^{(1)}
+\sum_{\vartheta\eta}\vartheta^{(0)}\otimes \vartheta^{(1)}\eta
-\sum_{\vartheta\eta}\eta^{(0)}\otimes \vartheta \eta^{(1)}\\
{}+\sum_{\vartheta\eta}1\otimes\vartheta\eta
=\sum_{\vartheta\eta}\vartheta^{(0)}\eta^{(0)}\otimes \vartheta^{(1)}\eta^{(1)}
+\sum_{\vartheta\eta}1\otimes\vartheta\eta=0. 
\end{multline*}
The second and third terms cancel out because of the 
$\sigma$-invariance of the initial quadratic expression, 
and the first and last terms are both zero. This proves the 
compatibility property of the ideal, and thus the existence 
of $\phi^\sim$. 

Therefore 
the hodc $\Gamma^\sim$ is situated, in terms of the 
richness of its generating relations, between the universal 
differential envelope $ \Guni $ of $\Gamma$ and the braided exterior 
calculus $ \Gbr $ over $\Gamma$. 
As we know from the general theory  
these two are the maximal and  minimal hodc, respectively, 
for  acceptable hodc's over $\Gamma$. 
\end{remark}

\subsection{The Cyclic Property Related to the Hodc}

We are ready to present the characterization 
of the cyclic property in terms 
of the hodc $ \Gamma^{\sim} $. 

\begin{theorem} 
The cyclic property for $\ddisp$ is equivalent to 
$$
( m_{\Omega(P)} \, \ddisp^{\otimes 2} ) ( \ker(\id-\sigma) ) = 0 . 
$$
In other words, in this case $\ddisp^\otimes$ is 
projectable---or $\ddisp$ is extendible---to a 
unital algebra morphism
$\ddisp^{\sim}\colon\ginv^{\sim}\rightarrow\hor(P)$. 
Here $ m_{\Omega(P)} $ denotes the multiplication in 
$ \Omega(P) $. 
\end{theorem}

\begin{remark}
Recall that $ \ddisp^{\otimes} $ was defined just before 
Theorem~\ref{cyclic-dunkl-qcd}. 
\end{remark}

\begin{proof} 
This follows from the fact that the $\sigma$-invariant 
elements of $\ginv\otimes\ginv$~are 
precisely those expressible as linear combinations 
of certain elements 
associated to orbits or, equivalently, 
to cyclic lines as we now describe. 
Indeed, given a cyclic line $\ell$ associated to the ordered sequence 
$(q_1,\dots, q_n)$ of $n$ distinct elements of $\S \subset \ginv$,  
we can consider the tensor defined by 
$$
\tau(\ell) :=q_1\otimes q_2+ q_2 \otimes q_3 + \cdots 
+q_{n-1} \otimes q_n+q_n\otimes q_1 
\in \ginv\otimes\ginv, 
$$
which does not change under cyclic permutation of the elements 
in the ordered sequence $(q_1,\dots, q_n)$, 
that is, it only depends on $\ell$.  
In other words $ \tau(\ell) $ is the sum over 
the pairs of the corresponding orbit (\ref{typical-orbit}), 
where each such pair is interpreted as a tensor product. 
Each $\tau(\ell)$ is clearly $\sigma$-invariant, 
that is, $ \tau(\ell) \in \ker(\id-\sigma) $. 
Moreover the set of all $\tau(\ell)$
provides us with a basis for $ \ker(\id-\sigma) $ 
as can be easily checked.  
As we remarked above the cyclic property says that 
$ m_{\Omega(P)} \, \lambda^{\otimes 2}$ annihilates every $\tau(\ell)$. 
So $ \lambda $ satisfies the cyclic property 
if and only if $ m_{\Omega(P)} \, \lambda^{\otimes 2}$ 
annihilates $ \ker(\id-\sigma) $.  
\end{proof}

\subsection{A Geometrical Property of Cyclic Dunkl Displacements}

The following proposition shows us an interesting geometrical property of cyclic Dunkl displacements, namely 
that they turn out to be base space one-forms when evaluated in the canonical $\ad$-invariant generator $[-\e]$. This can thus be interpreted as an ingredient of the base space geometry. It is worth recalling here that in the theory of quantum characteristic classes, in a similar spirit, we consider the cohomology classes of the closed forms on the base, expressible in terms of the connection and its differential.

\begin{prop}
For every cyclic Dunkl 
displacement $\ddisp$ the element $\ddisp_0$ 
defined by 
\begin{equation}
\label{define-ddisp-0}
\ddisp_0:=\sum_{s\in \S}\ddisp(s)=\ddisp(-\e) \in \frak{hor}^1 (P) 
\end{equation}
is right invariant under the right co-action 
$\FWP$. 
Therefore $\ddisp_0$ belongs to the hodc 
$\Omega(M)$ of the "base space". 
\end{prop}
\begin{proof}
We compute 
\begin{align*}
\FWP ( \ddisp_0 ) &= \sum_{s \in \S} \FWP \big( \ddisp(s) \big) 
= \sum_{s \in \S} \sum_{k \in G} 
\ddisp (k s k^{-1}) \otimes \delta_k 
\\
&=  \sum_{k \in G} \sum_{s \in \S}
\ddisp (k s k^{-1}) \otimes \delta_k 
\\
&=  \sum_{k \in G} \sum_{s \in \S}
\ddisp ( s ) \otimes \delta_k 
\\
&=  \sum_{s \in \S}
\ddisp ( s ) \otimes  \sum_{k \in G} \delta_k 
\\
&= \ddisp_0 \otimes 1_{\cal{A}} 
\end{align*}
using 
the identity  \eqref{F-hat-of-ddisp} in the second equality 
and the property  $g \S g^{-1} = \S$ in the fourth equality. 
This proves that $\ddisp_0$ is right invariant. 
The last statement in the proposition follows immediately 
from the definition of $\Omega(M)$. 
\end{proof}

\subsection{The (Quantum) Curvature}

\noindent 
The {\em (quantum) curvature} 
$r_\omega : \cal{A} \to \frak{hor}^2 (P)$ 
of a quantum connection $\omega$ 
is defined~as 
\begin{equation}
\label{define-quantum-curvature}
r_\omega(a) =  d_P \big(\omega \big( \pi(a) \big) \big) 
+ \omega ( \pi(a^{(1)}) ) \, \omega ( \pi(a^{(1)}) ) 
\end{equation}
for all $a\in\cal{A}$,  
where $\phi(a) = a^{(1)} \otimes a^{(2)}$ 
in Sweedler's notation.  

\begin{theorem} 
The curvature $r_\omega$ 
of the cyclic Dunkl connection 
$\omega = \tilde{\omega}+\ddisp$ is given by 
$r_\omega =  r_{\tilde{\omega}}$. 
\end{theorem}

\begin{remark}
This theorem says that the cyclic Dunkl displacement 
$ \ddisp $ does not change the non-commutative (or quantum)
geometry described by the quantum curvature. 
Thinking of $ \ddisp $ as a certain type of 
perturbation, we can say 
that the curvature is invariant under such 
a perturbation. 
\end{remark}

\begin{proof} 
The hypothesis is that $ \ddisp $ 
is a cyclic Dunkl displacement with respect 
to $ \tilde{\omega} $. 
According to Equation~(3.10) in \cite{DS} as applied 
in this context, for all $ k \in G $ we have 
\begin{align}
\label{r-omega-curvature}
   r_\omega (\delta_k) &= r_{\tilde{\omega}+\ddisp} (\delta_k) 
   \\
   &=  
   r_{\tilde{\omega}} (\delta_k) 
   + D_{\tilde{\omega}} \big( \ddisp ([\delta_k]) \big) 
   + \ddisp ([\delta_k^{(1)}]) \ddisp ([\delta_k^{(2)}]) 
   \nonumber 
   \\
   &= 
   r_{\tilde{\omega}} (\delta_k) 
   + D_{\tilde{\omega}} \big( \ddisp ([\delta_k]) \big) 
   + \sum_{g h = k} 
   \ddisp ([\delta_g]) \ddisp ([\delta_h]). 
   \nonumber
\end{align} 
We now analyze the images $ \pi ( \delta_k ) = [ \delta_k ] $ 
of the basis $ \{ \delta_k ~|~ k \in G \} $ of $ \cal{A} $ 
under the quantum germs map $ \pi $. 
By Section 13.4 in \cite{Part-II} 
we have these three exclusive and exhaustive cases: 
\begin{itemize}

\item 
$  [\delta_s] $ for $ s \in \S \,  $,
which form 
a basis of $ \ginv $, and so $  [\delta_s] \ne 0 $.

\vskip 0.2cm
\item
$  [ \delta_e ] = - \sum_{s \in \S}  [ \delta_s ] \ne 0$. 

\vskip 0.2cm
\item
$ [ \delta_k ] = 0 $ for $ k \notin \S \cup \{ e \} $. 

\end{itemize}

By the third property we see that 
$ D_{\tilde{\omega}} \big( \ddisp ([\delta_k]) \big) = 0 $
for all $ k \notin \S \cup \{ e \} $. 
And by the first two properties and the 
Closed-ness property of $ \ddisp $ we see that 
$ D_{\tilde{\omega}} \big( \ddisp ([\delta_k]) \big) = 0 $
for all $ k \in \S \cup \{ e \} $. 
Consequently, 
$ D_{\tilde{\omega}} \big( \ddisp ([\delta_k]) \big) = 0 $
for all $ k \in G $. 
So \eqref{r-omega-curvature} becomes 
\begin{equation}
\label{getting-simpler}
r_\omega (\delta_k) = 
 r_{\tilde{\omega}} (\delta_k) 
   + \sum_{g h = k} \ddisp ([\delta_g]) \ddisp ([\delta_h]).
\end{equation}

It is worth observing that our curvature 
formula \eqref{getting-simpler}
can be rewritten for every $a\in\cal{A}$ in a more compact form  as
$r_\omega(a)=
r_{\tilde{\omega}}(a)+
m_{\Omega(P)} 
\left(\ddisp^{\sim} \left( \phi(a) \right) \right)$, 
where $ \phi $ is the co-multiplication of $ \cal{A} $ 
and $ m_{\Omega(P)}  $ is the multiplication in
$ \Omega(P) $.

Notice that the sum in \eqref{getting-simpler}
 is over all pairs $ (g,h) $ such 
that $ g h = k $, where $ k $ is the group element 
appearing on the left side. 
But by the three properties given just above a term 
in this summation can only be non-zero if both 
$ g \in \S \cup \{ e \} $ as well as 
$ h \in \S \cup \{ e \} $. 
In turn this breaks down into four 
mutually exclusive and exhaustive cases 
for the terms in the summation, 
recalling that $ e \notin \S $. 
We now examine these cases. 
Throughout $ k \in G $ is a given element. 

\noindent 
Case~1:  $ g \in \S $ and $ h \in \S $ 
with $ g h = k $, that is,
we are summing over the set
\begin{equation}
\label{define-M-k}
\cal{M}_k := 
\{ (g,h) \in \S \times \S ~|~ g h =k \} \subset \S \times \S. 
\end{equation}
We write the set $ \cal{M}_k $ as a (disjoint!) union 
of certain orbits in $ \S \times \S $ of the 
flip-over operator $ \sigma $, namely
\begin{equation}
\label{M-k-is-union-of-orbits}
\cal{M}_k = \bigcup \, \{ \cal{O} ~|~ \cal{O} {\rm~is~an~orbit~of~} 
\sigma {\rm ~satisfying~} {\rm Inv}(\cal{O}) = k \}. 
\end{equation}
This is so since clearly $ \cal{O} \subset \cal{M}_k$ 
whenever $ \mathrm{Inv}(\cal{O}) =k $. 
Conversely, if $ (g,h) \in \cal{M}_k $, then the 
$ \sigma $-orbit $ \cal{O} $ of $ (g,h) $ 
lies in $ \cal{M}_k $ and $ \mathrm{Inv}(\cal{O}) =k $. 

But for {\em any} orbit 
$ \cal{O} \leftrightarrow (q_1, \dots , q_n) $ 
of $ \sigma $
we have that
$$
\ddisp(\cal{O}):= 
\sum_{j =1}^n \ddisp ([q_j] ) \, \ddisp ([q_{j+1}]) = 0
$$
by the cyclic property of $ \ddisp $. 
 
Consequently, by the two previous equalities we see that 
$$
 \sum_{(g,h) \in \cal{M}_k} \!\!\!\! \ddisp([\delta_g]) \, \ddisp([\delta_h]) 
 \,\, = 
\!\!\! \sum_{\cal{O} \,:\, {\rm Inv}(\cal{O}) =k} 
\!\!\!\!\!\! \ddisp (\cal{O}) = 0. 
$$

\noindent 
Case~2:  
$ g = e $ and $ h \in \S $. 
We then have $ k = g h = h \in \S $, and so this case 
does not occur for $ k \notin \S $. 
But for $ k \in \S $ we do get one term, namely  
$$ 
\ddisp ([\delta_e]) \ddisp ([\delta_k]). 
$$

\noindent 
Case 3: $ g \in \S $ and $ h = e $. 
Then $ k = g h = g \in \S $ and so again, as in the previous 
case, there are no such terms if $ k \notin \S $. 
However, for $ k \in \S $ we get again just one term 
which now is
$$ 
\ddisp ([\delta_k])  \ddisp ([\delta_e]). 
$$

\noindent 
Case 4: $ g = e $ and $ h = e $. 
So $ k = g h = e \notin \S $. 
So the case is vacuous if $ k \ne e $. 
Note that when $ k = e \notin \S $ we will have 
in general terms from Case~1 as well. 
If $ k = e $, then 
this case gives us exactly one term in the summation, 
namely
$$
   \ddisp ([\delta_e]) \ddisp ([\delta_e]) = 0,
$$
since the product here is actually the wedge product 
of the de Rham differential calculus of elements 
in $ \hor^1 (P) = \cal{D}^1 $, the $ 1 $-forms on $ E $.  

Now we have to add up the results from these four cases. 
If $ k \in \S  $ then we add the results from Cases 1, 2
and 3 getting 
$$
\sum_{g h = k} \ddisp ([\delta_g]) \ddisp ([\delta_h]) =
\ddisp ([\delta_e]) \, \ddisp ([\delta_k])  + 
 \ddisp ([\delta_k]) \, \ddisp ([\delta_e]) = 0,
$$
since again the product in $ \hor^1 (P) $ is anti-commutative.  
If $ k \notin \S$, then we add up the results from 
Cases 1 and 4, and again we get zero. 

Therefore by substituting this into 
\eqref{getting-simpler} we find for all $ k \in G $ that 
$$
 r_\omega (\delta_k) = r_{\tilde{\omega}} (\delta_k). 
$$
Since the $ \delta_k $'s form a vector space basis 
of $ \cal{A} = \cal{F}(G) $, we conclude 
that $  r_\omega = r_{\tilde{\omega}} $. 
\end{proof}

We have an immediate consequence: 
\begin{corollary}
Under the hypothoses of the previous theorem 
we have that the cyclic Dunkl connection 
$ \omega $ has curvature zero provided that 
the initial connection $ \tilde{\omega} $ 
has curvature zero. 
\end{corollary}

Because of this Corollary, the next result is relevant. 
\begin{theorem}
The connection $ \tilde{\omega} $ defined  
in \eqref{define-omega-tilde} has zero curvature.  
\end{theorem}
\begin{proof}
This is a straightforward calculation. 
For $ a \in \cal{A} $ we have 
\begin{align*}
\tilde{r}_{\omega} (a) &= d_P (\omega(\pi(a))) 
+ \omega(\pi(a^{(1)}) \, \omega(\pi(a^{(2)}) 
\\
&= d_P (1_{\cal{B}} \otimes \pi(a)) 
+ (1_{\cal{B}} \otimes \pi(a^{(1)}) ) \, 
  (1_{\cal{B}} \otimes \pi(a^{(2)}) ) 
\\
&= 1_{\cal{B}} \otimes d^{\wedge} \pi(a) 
+  1_{\cal{B}} \otimes \pi(a^{(1)}) \pi(a^{(2)}) 
\\
&= 1_{\cal{B}} \otimes 
   \big (d^{\wedge} \pi(a) 
          + \pi(a^{(1)}) \pi(a^{(2)}) 
   \big) = 0.
\end{align*}
In the last equality we used the Maurer-Cartan identity.
(See Section~10.3 in \cite{Part-II}.) 
We also used the definitions given in Appendix~B of 
$ d_{P} $ and of the multiplication in $ \Omega (P) $.  
\end{proof}

\subsection{Multiplicativity of Cyclic Dunkl  Connections}

Another very important property of cyclic 
Dunkl connections is that they are 
always multiplicative relative 
to the acceptable hodc $\ginv^{\sim}$. 
Here is the general definition: 

\begin{definition}
Suppose that $\Gamma^{\square}$ 
is an acceptable hodc 
extending a bicovariant, $ * $-covariant fodc 
$d : \cal{A} \to \Gamma$. 
Let $P = (\mathcal{B}, \mathcal{A}, F)$ be a QPB
with hodc $ \Omega(P) $. 
Then we say that a 
quantum connection $\omega : \ginv \to \Omega^1 (P)$ is 
{\em multiplicative relative to the hodc $\Gamma^{\square}$}
if $ \omega $ extends to a unital, 
multiplicative morphism
$
\omega^{\square} : \Gamma^{\square}_{\mathrm{inv}} 
\to \Omega(P).
$
\end{definition}

\begin{remark}
\label{multiplicative-remark}
In the special case of 
this paper  $P = (\mathcal{B}, \mathcal{A}, F)$  is
a QPB for the finite group $ G $, 
where $ \mathcal{A} = \mathcal{F} (G) $. 
and we define the hodc to be 
$ \Omega(P) := \mathcal{D} \otimes \ginv^{\square} $ 
instead of using \eqref{define-Omega-P}.

As another comment, we note that if $ \omega^{\square} $
exists, then it is unique and consequently if 
$ \omega $ is a $ * $-morphism, then $ \omega^{\square} $
also is a $ *$-morphism.

If $\Omega(P) = \mathcal{D} \otimes \ginv^{\wedge}$, 
where $ \Guni $ is the universal envelope of $ \Gamma $, 
then the quantum connection 
$\omega : \ginv \to \Omega^1 (P)$ 
is multiplicative if and only if 
$$  
\omega ( \pi(a^{(1)}) ) \, \omega ( \pi(a^{(2)}) ) 
 = 0 \in \Omega^2 (P) 
$$
for all $a \in \cal{R}$, the right ideal in 
$\ker \, \e \subset \mathcal{A}$ 
used to define the fodc $ \Gamma $. 
Here we are using Sweedler's notation 
$\phi(a) = a^{(1)} \otimes a^{(2)}$. 
This non-trivial result is shown in \cite{MichoQPB1}. 
\end{remark}

\begin{theorem}
The connection $ \tilde{\omega} $ 
defined in \eqref{define-omega-tilde} is multiplicative 
relative to the universal differential envelope $ \Guni $ of $\Gamma$.  
\end{theorem}
\begin{proof}
Using the previous remark, we take 
$a \in \cal{R} \subset \cal{A}$ 
and calculate using  \eqref{define-omega-tilde} 
and the definition of multiplication in 
$ \Omega(P) $ as defined in Appendix~B to get 
\begin{multline*}
\tilde{\omega} ( \pi(a^{(1)}) ) \, \tilde{\omega} ( \pi(a^{(2)}) ) = 
 (1_{\cal{B}} \otimes \pi(a^{(1)}) ) \, 
  (1_{\cal{B}} \otimes \pi(a^{(2)}) ) 
\\
= 
 1_{\cal{B}} \otimes \pi(a^{(1)}) \pi(a^{(2)}) 
= - 1_{\cal{B}} \otimes d^{\wedge} \pi(a) = 0, 
\end{multline*}
since 
$ \cal{R} \subset \ker \, \pi $. 
Again, we used the Maurer-Cartan identity. 
\end{proof}

\subsection{The Main Theorem}

\noindent 
We now collect the results proved above 
into the main theorem of this paper: 
\begin{theorem}
Let $ P =( C^\infty (E), \cal{F} (G), F ) $ be a QPB 
with finite structure group 
$ G $ and right co-action $ F $ induced by a right 
action of $ G $ on $ E $. 
Let $ ( \Omega(P), \Gamma^\wedge , \FWP ) $ be the 
hodc associated to an fodc
$ d : \mathcal{A} \to \Gamma $, as defined above. 
Let $ \tilde{\omega} $ be the regular quantum connection 
defined in \eqref{define-omega-tilde}. 

Then, for every cyclic Dunkl displacement 
$ \ddisp : \ginv \to \hor^1 (P) $ 
with respect to $ \tilde{\omega} $, the 
cyclic Dunkl connection $ \omega = \tilde{\omega} + \ddisp $ has curvature zero, 
$ r_{\omega} \equiv 0 $. 

Moreover, the coordinate cyclic Dunkl operators 
$ \{ D_{\omega}^{j} ~|~ j = 1, \dots, n = \dim \, E \} $, 
defined in \eqref{gen-dunkl-operators} 
as the coordinates of the covariant 
derivative $ D_{\omega} $
of the quantum connection $ \omega $, 
commute among themselves. 
\end{theorem}
\begin{proof}
Only the last statement remains to be proved. 
First off, we note 
for all homogeneous elements 
$ \varphi \in \cal{D}^k$ that 
we obtain a general formula relating 
the square of the covariant derivative 
with the curvature:
\begin{align*}
D_{\omega}^{2} (\varphi) &=
d_{P} (D_{\omega} \varphi) - 
(-1)^{\partial (D_{\omega}\varphi)} 
(D_{\omega} \varphi)^{(0)} 
\omega([ (D_{\omega}\varphi)^{(1)} ])
\\
&=
d_{P} (D_{\omega} \varphi) - (-1)^{1 +\partial \varphi} 
(D_{\omega} \varphi^{(0)})  
\omega([ \varphi^{(1)} ])
\\
&=
d_{P} (D_{\omega} \varphi) + (-1)^{\partial \varphi} 
(D_{\omega} \varphi^{(0)})  
\omega([ \varphi^{(1)} ])
\\
&= d_{P}^{2} \varphi - (-1)^{\partial \varphi} 
d_{P} \big( \varphi^{(0)} \omega([ \varphi^{(1)} ]) \big)
\\
& \,\,\,\,\,\, + (-1)^{\partial \varphi} 
\big\{  
d_{P} \varphi^{(0)} - (-1)^{\partial \varphi^{(0)}} 
\varphi^{(00)} \omega ([\varphi^{(01)}])
\big\} 
\omega([ \varphi^{(1)} ]) 
\\
&= d_{P}^{2} \varphi - (-1)^{\partial \varphi} 
d_{P} \big( \varphi^{(0)} \big) \omega([ \varphi^{(1)} ]) 
- (-1)^{\partial \varphi + \partial \varphi^{(0)}} 
\varphi^{(0)} d_{P} \big( \omega([ \varphi^{(1)} ]) \big)
\\
& \,\,\,\,\,\, + (-1)^{\partial \varphi} 
\big\{  
d_{P} \varphi^{(0)} - (-1)^{\partial \varphi^{(0)}} 
\varphi^{(00)} \omega ([\varphi^{(01)}])
\big\} 
\omega([ \varphi^{(1)} ]) 
\\
&= 
-\varphi^{(0)} d_{P} \big( \omega([ \varphi^{(1)} ]) \big) 
- \varphi^{(00)} \omega ([\varphi^{(01)}])
\omega([ \varphi^{(1)} ]) 
\\
&= 
-\varphi^{(0)} d_{P} \big( \omega([ \varphi^{(1)} ]) \big) 
- \varphi^{(0)} \omega ([\varphi^{(11)}])
\omega([ \varphi^{(12)} ]) 
\\
&= 
-\varphi^{(0)} 
\Big(  
d_{P} \big( \omega([ \varphi^{(1)} ]) \big) 
+ \omega ([\varphi^{(11)}])
\omega([ \varphi^{(12)} ]) 
\Big)
\\
&= 
-\varphi^{(0)} r_{\omega} (\varphi^{(1)}), 
\end{align*}
where we used the definition 
\eqref{define-covariant-derivative}
of $ D_{\omega} $ thrice, 
the covariance of $ D_{\omega} $ 
in the second equality
(Theorem 12.12 in \cite{Part-II}), 
$ \partial \varphi = \partial \varphi^{(0)} $, 
the graded Liebniz rule for~$ d_{P} $, 
$ d_{P}^{2} = 0 $, the co-action property
in Sweedler's notation and the 
definition \eqref{define-quantum-curvature} 
of the curvature $ r_{\omega} $. 
So we have $  D_{\omega}^{2} (\varphi) = 0 $, 
since $ r_{\omega} \equiv 0 $. 

But on the other hand a direct computation 
using \eqref{define-Dunkl-gradient} 
for $ \varphi \in C^\infty (E) $ gives 
\begin{multline*}
D_{\omega}^{2} (\varphi) = 
\sum_{k=1}^{n} D_{\omega} 
\left( D_{\omega}^{k} (\varphi \, d x_{k} ) \right) 
=
\sum_{k=1}^{n} D_{\omega} 
\left( D_{\omega}^{k} (\varphi) \right)  d x_{k} 
\quad \begin{pmatrix}\text{more about}\\ \text{this step later}\end{pmatrix}
\\
=
\sum_{j,k=1}^{n} 
D_{\omega}^{j} D_{\omega}^{k} (\varphi) 
\, d x_{j} \wedge d x_{k} 
=
\sum_{1 \le j<k \le 1}^{n} 
\left( D_{\omega}^{j} D_{\omega}^{k} (\varphi) -
D_{\omega}^{k} D_{\omega}^{j} (\varphi) \right) 
d x_{j} \wedge  d x_{k} .
\end{multline*}
For the last equality we used that the horizontal 
forms are the complexified de~Rham differential forms. 
Now the forms  $ d x_{j} \wedge d x_{k} $ for $ j < k $
are linearly independent, and so it follows 
from  $  D_{\omega}^{2} (\varphi) = 0 $ that 
$$
 D_{\omega}^{j} D_{\omega}^{k} (\varphi) = 
 D_{\omega}^{k} D_{\omega}^{j} (\varphi)
$$
for {\em all} $ j,k \in \{ 1, \dots, n \} $ 
and for all $ \varphi \in \cal{D}^0 = C^\infty (E) $. 

However, we still owe the reader more details about the 
second step in the above calculation. 
Among other things this depends on the fact 
that $ d x_{k} $ is a horizontal form, that 
is an element in $ \Omega(M) $ or in other 
words that it is $ \FWP $ invariant. 
(See \eqref{define-omega-M}). 
So we examine some relations for a homogeneous form
$ \w \in \Omega(M) $ and later specialize 
to the case when $ \w = d x_{k} $. 
We first note for all $ \theta \in \ginv $
that
\begin{equation}
\label{omega-theta-slash-w}
 \omega (\theta) \w = 
 (-1)^{\partial \w} \, \w \, \omega (\theta), 
\end{equation}
which follows from the definition of 
multiplication in $ \Omega (P) $. 
Next, we compute  
for any $ \varphi \in \mathcal{D}^k \cong \hor^{k}(P) $
that 
\begin{align*}
D_{\omega} (\varphi \w) &= d_{P} (\varphi \w) 
- (-1)^{\partial \varphi + \partial \w} 
(\varphi \w)^{(0)} \omega ( [ (\varphi \w)^{(1)} ])
\\
&= d_{P} (\varphi \w) 
- (-1)^{\partial \varphi + \partial \w} 
\varphi^{(0)} \w \, \omega ( [ \varphi^{(1)} ])
\\
&= d_{P} (\varphi \w) 
- (-1)^{\partial \varphi} 
\varphi^{(0)} \omega ( [ \varphi^{(1)} ]) \w 
\\
&= d_{P} (\varphi) \, \w 
+ (-1)^{\partial \varphi} \varphi \, d_{P} \w  
- (-1)^{\partial \varphi} 
\varphi^{(0)} \omega ( [ \varphi^{(1)} ]) \, \w 
\\
&=
D_{\omega} (\varphi) \, \w 
+ (-1)^{\partial \varphi} \varphi \, d \w  
\end{align*}
where we used the definition of $ D_{\omega} $, 
the formula for the co-action $ \FWP $ 
acting on $ \varphi \w $ in the second equality
(see Appendix~B), 
\eqref{omega-theta-slash-w} in the third
equality,
the Leibniz rule for $ d_{P} $, the fact that  
$ d_{P} $ on horizontal forms reduces to the 
de~Rham differential $ d $
and finally the definition of $ D_{\omega} $ again. 
To conclude we take $ \w = d x_{k} $ 
as we indicated earlier. 
So we get $ d \, \w = d d x_{k} = 0 $
and therefore
$$
D_{\omega} (\varphi \, d x_{k})  =  
D_{\omega} (\varphi ) \, d x_{k},
$$
which justifies the second step in the above argument.  
\end{proof}

\begin{remark} 
The main result of \cite{DS} can be understood as a consequence 
of the Dunkl QCD given here in \eqref{define-coxeter-ddisp} 
as being a cyclic Dunkl displacement. 
\end{remark}

\subsection{Some More on Multiplicativity}

\begin{theorem} 
Every cyclic Dunkl connection 
$ \omega = \tilde{\omega} + \ddisp $ 
with $ \tilde{\omega} $ multiplicative  
relative to $ \Guni $
is itself multiplicative relative to $ \Guni $. 
\end{theorem}

\begin{proof}
Let 
$\omega^{\otimes}\colon\ginv^{\otimes}
\rightarrow \Omega(P)$ 
be the unital multiplicative extension of $\omega$. 
This is defined analogously to the definition 
given just before Theorem~\ref{cyclic-dunkl-qcd} of 
$ \ddisp^\otimes $. 
According to Remark \ref{multiplicative-remark}
we have to show that 
\begin{equation}
\label{have-to-show}
\omega ( [a^{(1)}] ) \omega ( [a^{(2)}] ) = 
m_{\Omega(P)} \, \omega^{\otimes 2} 
(\pi \otimes \pi) \phi (a) = 0
\end{equation}
holds for 
all $ a \in \mathcal{R} $. 
So we take $ a \in \mathcal{R} $ and 
use $ \omega = \tilde{\omega} + \ddisp $ 
to get 
\begin{align*}
&m_{\Omega(P)} \, \omega^{\otimes 2} 
(\pi \otimes \pi) \phi (a) = 
\omega ( [a^{(1)}] ) \omega ( [a^{(2)}] )
\\
&= 
\tilde{\omega} ( [a^{(1)}] ) 
\tilde{\omega} ( [a^{(2)}] )
\!+\! \tilde{\omega} ( [a^{(1)}] ) \ddisp ( [a^{(2)}] ) 
\!+\! \ddisp ( [a^{(1)}] ) \tilde{\omega} ( [a^{(2)}] ) 
\!+\! \ddisp ( [a^{(1)}] ) \ddisp ( [a^{(2)}] )
\\
&= 
\tilde{\omega} ( [a^{(1)}] ) \ddisp ( [a^{(2)}] ) 
+ \ddisp ( [a^{(1)}] ) \tilde{\omega} ( [a^{(2)}] ) 
+ \ddisp ( [a^{(1)}] ) \ddisp ( [a^{(2)}] ), 
\end{align*}
where we used $ \tilde{\omega} $ is 
multiplicative (Remark \ref{multiplicative-remark})
as well as the compact notation 
$ [\cdot] = \pi(\cdot) $. 
We note that $ \mathcal{R} $ has a basis 
$
\mathcal{B}_{\mathcal{R}} := 
\{\delta_{g} ~|~ g \notin \mathcal{K}_{\mathcal{R}} 
\}
$
by Corollary 13.1 in \cite{Part-II},
where 
$ \mathcal{K}_{\mathcal{R}} := 
\{ h \in G ~|~ f(h)=0 
\mathrm{~for~all~}f \in \mathcal{R} \} $. 
Since \eqref{have-to-show} is linear in 
$ a \in \mathcal{R} $, it suffices to prove it 
for the basis elements 
$ \delta_{g} \in \mathcal{B}_{\mathcal{R}} $. 
Now for $ a = \delta_{g} $ we have
\begin{align*}
a^{(1)} &\otimes a^{(2)} = \phi(\delta_{g}) 
=
\sum_{ hk = g } \delta_{h} \otimes \delta_{k} 
= \!\!\!\!
\sum_{ (h, k) \in \mathcal{M}_{g} } 
\delta_{h} \otimes \delta_{k} 
\\
&= \!\!
\sum_{\mathcal{O}: \mathrm{Inv}(\mathcal{O}) = g}
\sum_{ (h, k) \in \mathcal{O} } 
\delta_{h} \otimes \delta_{k}, 
\end{align*} 
where $ \mathcal{O} $ is a $ \sigma $-orbit and 
$ \mathrm{Inv}(\mathcal{O}) $ is defined in 
\eqref{define-group-invariant}. 
Also, $ \mathcal{M}_{g} $
is defined in \eqref{define-M-k} and 
satisfies \eqref{M-k-is-union-of-orbits}. 
We will use explicit summation notation 
(here with $ h, k \in G $)
instead of Sweedler's when the former is handier. 
It follows that 
$$
\ddisp ( [a^{(1)}] ) \ddisp ( [a^{(2)}] )~~~~ =
\sum_{\mathcal{O}: \mathrm{Inv}(\mathcal{O}) = g}
\,
\sum_{ (h, k) \in \mathcal{O} } 
\ddisp(\delta_{h})  \ddisp (\delta_{k}) = 0, 
$$
since $ \ddisp  $ satisfies the Cyclic Property and 
so the inner sum already is $ 0 $.  
So the equation above reduces to 
\begin{equation}
\label{reduced-eqn}
m_{\Omega(P)} \, \omega^{\otimes 2} 
(\pi \otimes \pi) \phi (a) = 
\tilde{\omega} ( [a^{(1)}] ) \ddisp ( [a^{(2)}] ) 
+ \ddisp ( [a^{(1)}] ) \tilde{\omega} ( [a^{(2)}] ).  
\end{equation}
Next, we analyze the first term on the right side
using $ \tilde{\omega} $ regular 
(see \eqref{define-regular-connection}) to get   
\begin{equation}
\label{omega-tidle-is-regular}
\tilde{\omega} ( [a^{(1)}] ) \ddisp ( [a^{(2)}] ) = 
\sum_{hk = g}
(-1)^{1} \ddisp ([\delta_{k}])^{(0)} 
\tilde{\omega} \big( [\delta_{h}] \circ \ddisp ([\delta_{k}])^{(1)} \big). 
\end{equation}
And next, using that $ \ddisp $ is a QCD in the
second equality, we calculate 
\begin{align*}
\ddisp ([\delta_{k}])^{(0)} &\otimes \ddisp ([\delta_{k}])^{(1)} = 
\FWP(\ddisp ([\delta_{k}]) ) 
= 
(\ddisp \otimes id) \, \mathrm{ad} ([\delta_{k}])
\\
&=
\ddisp ([\delta_{k}^{(2)}]) \otimes \
\kappa (\delta_{k}^{(1)}) \, \delta_{k}^{(3)}
=
\sum_{lmn=k} 
\ddisp ([\delta_{m}]) \otimes \
\kappa (\delta_{l}) \, \delta_{n}, 
\end{align*}
where $ l,m,n \in G $. 
We also used the formula \eqref{ad-ad-relation}
for ad. 
Then \eqref{omega-tidle-is-regular} becomes
\begin{align*}
&\tilde{\omega} ( [a^{(1)}] ) \ddisp ( [a^{(2)}] ) = 
- \!\!\! \sum_{hklm=g}
\ddisp ([\delta_{l}]) \, 
\tilde{\omega} \big( [\delta_{h}] \circ 
\kappa (\delta_{k}) \, \delta_{m} \big)
\\
&= 
- \!\!\! \sum_{hklm=g}
\ddisp ([\delta_{l}]) \, 
\tilde{\omega} \big( [\delta_{h}] \circ 
\delta_{k^{-1}} \, \delta_{m} \big)
\\
&= 
- \!\!\!\!\!\!\! \sum_{hm^{-1}lm=g}
\!\!\!\ddisp ([\delta_{l}]) \, 
\tilde{\omega} \big( [\delta_{h}] \circ \delta_{m} \big)
= 
- \sum_{lh=g}
\ddisp ([\delta_{l}]) \, 
\tilde{\omega} \big( [\delta_{h}] \big)
=
- \ddisp ( [a^{(1)}] ) \tilde{\omega} ( [a^{(2)}] ).
\end{align*}
In the second equality 
we used $ \kappa (\delta_{k}) = \delta_{k^{-1}} $ 
as the reader can verify. 
To get the third equality we summed on $ k $
and for the fourth we summed on $ m $ and 
used the properties of the $ \circ $ operation
mentioned just after
\eqref{Woronowicz-sigma-quadratic-formula}. 
And in the last equality we reverted back to 
Sweedler's notation. 
Substituting back into \eqref{reduced-eqn} we get
the desired result.  
\end{proof}

\begin{corollary} 
	Every cyclic Dunkl connection 
	$ \omega = \tilde{\omega} + \ddisp $ 
	with $ \tilde{\omega} $ multiplicative
	is extendible (in a necessarily unique way) 
	to a unital algebra morphism 
$$
\omega^{\sim}\colon\ginv^{\sim}\rightarrow\Omega(P). 
$$ 
\end{corollary}

\begin{remark} If we use the calculus $\Gamma^\sim$ as the  
hodc over the fodc $\Gamma$, then the above proposition 
can be restated as the multiplicativity of Dunkl connections. In particular, they are always multiplicative for the 
universal differential algebra 
$\Guni$ over $\Gamma$. 

So cyclic Dunkl connections provide 
a nice class of examples 
for non-regular (in general) but nevertheless 
multiplicative connections.  
For example, the Dunkl operators as originally defined 
in \cite{dunkl-1989} do not satisfy the Leibniz rule 
as is shown in \cite{DS}. 
Then by Theorem 12.14 of \cite{Part-II} 
this implies that the 
corresponding Dunkl connection is not regular, 
but by what we have shown here they are multiplicative. 
\end{remark}

\begin{remark} 
For multiplicative connections $\omega$ 
there is a very nice interpretation of the curvature 
as a measure of the deviation 
of $\omega^{\sim}\colon\ginv^{\sim}\rightarrow 
\Omega(P)$ from being a differential algebra morphism. 
Indeed, for every $a\in\cal{A}$ we have 
$$ 
\bigl\{d\omega^{\sim}-\omega^{\sim}d\bigr\}\pi(a)=
d\omega\pi(a)+\omega\pi(a^{(1)})\omega
\pi(a^{(2)})=r_\omega(a), 
$$
where we used the Maurer-Cartan formula 
$d \pi(a) = - \pi(a^{(1)}) \pi(a^{(2)})$ and 
where we write $\phi(a) = a^{(1)} \otimes a^{(2)}$ 
in Sweedler's notation. 

In this case, there are no residual curvature terms
for $a\in\cal{R}$ and 
the curvature map 
$r_\omega\colon\mathcal{A} \rightarrow\hor^2(P)$
naturally projects down to 
$r_\omega\colon\ginv\rightarrow\hor^2(P)$. 
As we know from the general theory 
\cite{MichoQPB2,MichoQPB3}
the residual curvature terms are a manifestation of an 
interesting purely quantum phenomena, where a specific 
quadratic combination
$\omega(\pi(a^{1)})\omega(\pi(a^{(2)})$ 
of `vertical' elements given by values of the connection 
form, surprisingly turns out to be horizontal 
for $a\in\cal{R}$. 
The presence of these curvature terms can be understood as 
the obstacle to the multiplicativity of the connection. 
\end{remark}

\section{Example: Complex Reflection Groups}
\label{sec-6}

\noindent \indent 
As an important special example 
of this construction, 
we consider complex reflection groups 
and their associated Dunkl operators. 
These operators were introduced in the paper \cite{dunkl-opdam}.

\subsection{Basic Definitions}

A {\em complex reflection} $s$ is a unitary transformation 
acting on a complex, 
finite dimensional vector space $ V $ 
with $\dim V =n  \ge 1$ and 
with a Hermitian inner product 
such that $s$ has finite order in the unitary group $\mathrm{U}(V)$ of $V$ 
and exactly one of the eigenvalues of $s$ is not
equal to $1$ and has multiplicity $1$. 
In particular, $s \ne e$, the identity.

Suppose that $S^\prime$ 
is a non-empty set of
complex reflections acting on $V$ and that 
$S^\prime$ generates a {\em finite} subgroup $G$ of $\mathrm{U}(V)$. 
Then we say that $G$ is 
{\em finite complex reflection group} acting on $V$. 
Let $\S$ be the smallest set containing $S^\prime$ 
such that $\S$ is closed under conjugation 
by arbitrary elements $g \in G$ and $S^{-1} = S$. 
So 
$S = 
\bigcup \{ g^{-1} s g, g^{-1} s^{-1} g ~|~ g \in G, s \in S^\prime \}$. 
It follows that $e \notin S$ and $\S$ is non-empty. 
So $\S$ determines a unique bicovariant, $*$-covariant fodc 
$\Gamma$ for the
Hopf algebra $\mathcal{A} = \mathcal{F} (G)$.
In particular, let $d : \mathcal{A} \to \Gamma$ 
denote the differential of this fodc. 

For every one-dimensional complex 
subspace $L$ of $V$ let $\xi_L$ be 
the {\em (singular) 
classical differential complex $1$-form} 
on $V$ defined by 
\begin{equation}\label{L-dunkl}
\xi_L (x)=
\frac{ 1 }{\langle x, \alpha\rangle} 
\sum_{j=1}^{n} \alpha_{j} d \overline{x}_{j} =
\frac{ \tilde{\alpha} }{\langle x, \alpha\rangle}, 
\end{equation}
where $0 \ne \alpha = (\alpha_{1}, \dots , \alpha_{n}) \in L$ 
is called a {\em representative vector} of $\xi_L$
and $x\in V$ 
satisfies $\langle x, \alpha\rangle \neq 0$, that is, 
$ x \notin L^{\perp} $, the hyperplane orthogonal 
to $ L $. 
In this section we use the notation 
$ \tilde{\alpha} = \sum_{j=1}^{n} \alpha_{j} d \overline{x}_{j} $ and the $ d \overline{x}_{j} $'s are
the standard $ (0,1) $ one-forms associated to some given
orthonormal basis of $ V $. 

\begin{remark} 
It is easy to see that the above definition \eqref{L-dunkl}  
does not depend on 
the representative vector $\alpha$, since our convention is that 
the inner product is linear in its second entry. 
And it also does not depend on the choice of the 
orthonormal basis. 
We include singular anti-meromorphic forms here, 
since  the scalar product 
$ \langle x , \alpha \rangle  $
is zero
for all points $x\in L^\bot \ne \emptyset$, 
the orthogonal complement of $L$. 
If $\dim \,V \ge 2$, then $L^\bot \ne 0$ as well. 
The exterior differential graded *-algebra structure remains 
well-defined for such singular $1$-forms, 
although with implicitly defined domain restrictions.
\end{remark}

\subsection{Preliminary Results}

\begin{prop} 
\label{3-L-I}
For every three linearly dependent 
$ 1 $-dimensional subspaces $X$, $Y$ and $Z$ of $V$ we have
\begin{equation}\label{superposition-3}
\xi_X\xi_Y+\xi_Y\xi_Z+\xi_Z\xi_X = 0 . 
\end{equation}
\end{prop}

\begin{proof} 
The property is trivial if $X=Y$, $Y=Z$ or $Z=X$. 
Hence we assume that 
$X\neq Y\neq Z\neq X$. 
Choose representative vectors $\alpha\in X$, $\beta\in Y$ and 
$\gamma\in Z$. 
These give us $3$ distinct vectors, 
since $X$, $Y$ and $Z$ are $3$ distinct 
one-dimensional subspaces. 
Every pair taken from $\{ \alpha, \beta, \gamma \}$ forms 
a basis of the 2-dimensional subspace spanned by  
all $3$ of them.  
We now compute  
\begin{align*}
\xi_X\xi_Y+\xi_Y\xi_Z \!&=\! 
\frac{\tilde{\alpha}}{\langle x, \alpha\rangle}
\frac{\tilde{\beta}}{\langle x, \beta\rangle}
{+}\frac{\tilde{\beta}}{\langle x, \beta\rangle}
\frac{\tilde{\gamma}}{\langle x, \gamma\rangle}
\!\!=\!\!\frac{\tilde{\alpha}\, \tilde{\gamma}}{\langle x, \alpha\rangle
\langle x, \beta\rangle\langle x, \gamma\rangle}
\frac{\langle x, \gamma\rangle-a\langle x, \alpha\rangle}{b}
\\
&= 
\frac{\tilde{\alpha}\, \tilde{\gamma}}{\langle x, \alpha\rangle\langle 
x, \gamma\rangle}=\xi_X\xi_Z = - \xi_Z\xi_X, 
\end{align*}
where $\gamma=a\alpha+b\beta$ and $a,b\in\Bbb{C}\setminus\{0\}$, using the graded 
commutativity of classical $ 1 $-forms in 
the very last step.  
\end{proof}

\begin{theorem} 
Let $n \ge 1$ be an integer. 
Then for every family $W_1$, \dots, $W_n$ 
of one-dimensional subspaces of $V$, 
every three of which are linearly dependent, we have that 
\begin{equation}
\label{W1-Wn}
\xi_{W_1}\xi_{W_2}+\xi_{W_2}\xi_{W_3}+
\cdots+\xi_{W_{n-1}}\xi_{W_n}+\xi_{W_n}\xi_{W_1}=0.
\end{equation}
\end{theorem}
\begin{proof} The hypothesis on the family $W_1$, \dots, $W_n$ is vacuously true 
for every family in the cases $n=1$ and $n=2$. 
But in those cases the conclusion is trivially true. 
The case $n=3$ was proved in Proposition~\ref{3-L-I}. 
For $n \ge 4$ one argues 
inductively applying the cyclic property 
\eqref{superposition-3} for the three spaces $W_1, W_n, W_{n+1}$ 
and taking \eqref{W1-Wn} as the induction hypothesis.  
\end{proof}

All of the subspaces $W_1$, \dots, $W_n$ in the previous theorem 
are contained in the common subspace 
$U = W_1 + \cdots + W_n$ of $V$ 
and the dimension of $U$ is $1$ or $2$. 
We say that this subspace $U$ is the {\em container} 
of $W_1$, \dots, $W_n$. 

Let us note that $G$ acts faithfully on $V$ 
by unitary transformations, and 
every element of the set $\S$ is a complex reflection. 
For every $s \in \S$, let $W_s \subset V$ denote 
the eigenspace of $s$ whose eigenvalue is not equal to $1$. 
So $W_s$ is one-dimensional. 
Then we define
\begin{equation}
\label{define-Omega}
       \Omega := \{ W_s ~|~ s \in \S \},
\end{equation}
the set of these eigenspaces.  

We also define 
$\sw\colon \S\rightarrow \Omega$ by $\sw(s):= W_s$ 
for all $s \in \S$. 
Then $\sw$ is a surjection, though it need not 
be an injection. 
It is clear that $G$ acts naturally on $\Omega$ by 
$ W_s \cdot g := W_{g^{-1} s g} $ 
for $s \in \S$ and $ g \in G $. 
Also,
$ G $ acts on $ \S $ by 
$s \cdot g := g^{-1} s g \in \S$. 
These are right actions. 
Then $\sw$ is covariant with respect to these actions of $ G $. 
The cyclic structure on $\S$ projects to a 
cyclic structure on $\Omega$, namely 
$ \sigma : \Omega \times \Omega \to \Omega \times \Omega $ is 
given by $ \sigma (W_s , W_t) = ( W_{s t s^{-1}} , W_s) $ 
for $ s,t \in \S $. 

Let us now observe that if $\alpha$ and $\beta$ are eigenvectors 
corresponding to complex reflections $u$ and $v$ 
(respectively)
and whose eigenvalues are $ \ne 1 $, 
then we claim that 
\begin{equation}\label{complex-colinearity}
\beta-u(\beta)=a\alpha, 
\end{equation}
for some $a\in \Bbb{C}$. 
Indeed, we can decompose 
$\beta=\beta_{\parallel}+\beta_{\bot}$ into parallel and 
orthogonal components to $\alpha$, and so  
$$u(\beta)=c\beta_{\parallel}+\beta_{\bot},$$ 
where $ u (\alpha) = c \alpha $ with $ c \ne 1 $.
From this 
the colinearity condition 
\eqref{complex-colinearity} follows from
$ \beta-u(\beta)=\beta_{\parallel}+\beta_{\bot} -
(c\beta_{\parallel}+\beta_{\bot}) = (1-c) 
\beta_{\parallel} $ 
with $ 1 -c \ne 0$. 
Using $ \beta_{\parallel} = b \alpha $ for 
some $ b \in \Bbb{C} $, then gives 
\eqref{complex-colinearity}. 
As we shall now see, this condition implies that every 
cyclic line in $\Omega$ possesses a one or 2-dimensional 
container, and hence the cyclic relation \eqref{W1-Wn} 
holds. We shall actually prove a more general result, 
regarding  groups of Coxeter type which we now define.

\subsection{Groups of Coxeter Type}

\begin{definition}
A finite group $G$, realized as a subgroup of the unitary group $\mathrm{U}(V)$, is of {\em Coxeter type}, if there is a (finite) generating subset $\St$ of $G$ and a function $\sw\colon 
\St\rightarrow \mathrm{CP}(V)$, 
the complex projective space of $V$ (namely 
the space whose elements are the one-dimensional 
subspaces of $ V $), 
such that:

\smallskip\bla{i} If $u,v\in\St$ then $uvu^{-1}\in\St$. 
This property implies that $\St$ is 
invariant under conjugations by arbitrary elements from $G$.  
In other words $ \St $ splits into one or more entire 
conjugacy classes. 

\smallskip\bla{ii} For $u,v\in\St$ we have $\sw(uvu^{-1})=u[\sw(v)]$. In particular, if 
$u$ and $v$ commute (for example, $u=v$) we see that $\sw(v)$ is $u$-invariant, and thus $u$ acts on the subspace $\sw(v)$ 
as multiplication by a unitary scalar. 

\smallskip\bla{iii} The space $u[\sw(v)]$ is contained in the linear span of $\sw(u)$ and $\sw(v)$. 
In other words if $0 \ne \alpha\in \sw(u)$ 
and $0 \ne \beta\in\sw(v)$, 
then 
\begin{equation}\label{generalized-colinearity}
u(\beta)=a\alpha+b\beta
\end{equation}
for some $a,b\in\Bbb{C}$. 
\end{definition}

Obviously, every complex reflection group is of Coxeter type. 
In this definition we are requiring  
that the "generalized reflections" $u \in \St$ transform 
the "root vector" representatives from $\sw(v)$ in a very local way, namely 
the transformed vector is always a linear combination of the initial vector $ \beta $
and the reflection transformation vector $ \alpha $.

\subsection{The Natural Cyclic Structure}

We can now proceed exactly as above by first 
defining the family $\Omega:=\mathrm{Ran} \, \mu$, 
a finite subset of  $ \mathrm{CP}(V)$, 
and then the natural cyclic structure 
induced by $ \mu $ on it. 
\begin{prop}\label{2-container} 
If $ G $ is of Coxeter type, then 
every cyclic line 
in $\Omega$ possesses 
a one or 2-dimensional container. 
Thus the above cyclic relation \eqref{W1-Wn} holds.  
\end{prop}

\begin{proof} 
We consider  $u$, $v$ 
in $\S$ ("generalized reflections")
and their associated cyclic line $\ell_{u,v}$. 
If $ u = v $, the result is trivial. 
So we assume that $ u \ne v $. 
The point preceding $ u $ on $\ell_{u,v}$ 
is $uvu^{-1}$. Let $\gamma\in \sw(uvu^{-1})\setminus\{0\}$ and choose non-zero vectors 
$\alpha$ and $\beta$ from $\sw(u)$ and $\sw(v)$, respectively. 
Since 
$\sw(uvu^{-1}) = u [\sw(v)] $, the third condition implies that  
$\gamma=u(\beta) = a \alpha + b \beta$ 
for some $a,b\in\Bbb{C}$. 

Now we can proceed inductively to conclude that every 
one-dimensional subspace associated to the points on the 
cyclic line $\ell_{u,v}$ is contained in the subspace spanned by $\alpha$ and $\beta$. 
\end{proof}

\subsection{}

{\bf The Quantum Principal Bundle
	and \\ Its Cyclic Quantum Displacement}
\vskip 0.2cm 
Let $E$ be the dense open subset of $V$  
consisting of all the vectors with trivial $G$-stabilizer. 
By definition $G$ acts freely on $E$. 
Furthermore, every vector $x\in E$ is not 
orthogonal to any of the subspaces $ W_s $ in $\Omega$ 
as defined in \eqref{define-Omega}. 
As we have seen 
$ P = (C^\infty(E), \cal{F}(G), F) $ is then 
a QPB, where $ F $ is the pull-back of the right action 
of $ G $ on $ E. $ 
Let $ \Omega(P) $ be the associated hodc which was 
constructed as in \cite{DS} and reviewed 
in Subsection~\ref{QPB} and Appendix~B.

Let us use the same symbol $\xi_W$ for the corresponding 
restricted differential, complex $1$-form on $E$. 
Notice that $\xi_W$ when restricted to $E$ has no 
singularities. 
In other words $\xi_W$ is 
a smooth (that is, $ C^\infty $) section, 
defined on all of $E$, of the 
complexified cotangent bundle 
$ T^*(E) \otimes \mathbb{C} $. 
In accordance with the above comments $\xi_W$ is 
identified as a horizontal $1$-form in $\Omega(P)$, 
namely $ \xi_W \in \hor^1 (P)$. 
Let us define $\lambda\colon\S\rightarrow\hor^1(P)$ as 
$$
\ddisp(s)=\nu(s) \, \xi_{\sw(s)},  
$$
where $\nu\colon\S\rightarrow\Bbb{C}$ is any function 
which is 
constant on each conjugation class of $\S$.  

\begin{theorem} 
The map $\ddisp$ is a cyclic Dunkl displacement for $P$, 
where $P$ is viewed as a quantum principal bundle with 
structure group $G$ and is 
equipped with the canonical connection $\tilde{\omega}$ 
defined in \eqref{define-omega-tilde}.
\end{theorem}

\begin{proof} 
We first prove that the form $\xi_W$ is closed with 
respect to the covariant derivative $ D_{\tilde{\omega}} $. 
But $ D_{\tilde{\omega}} = d$, 
the complexified de Rham differential. 
Recall that 
$
\xi_W (x) =\tilde{\alpha} / \langle x , \alpha \rangle 
$, 
where $ \alpha \in W \setminus \{ 0 \} $. 
Then we calculate that 
\begin{align*}
   &d \, \xi_W 
   = d \Big( \sum_j \dfrac{\alpha_j}{\langle x , \alpha \rangle} d \bar{x}_j 
        \Big)
   = \sum_j d \Big[
              \dfrac{\alpha_j}{\langle x , \alpha \rangle}
           \Big] \wedge \, d \bar{x}_j 
   \\
   &= \sum_j \sum_k \dfrac{\partial}{\partial \bar{x}_k} 
   \Big[
              \dfrac{\alpha_j}{\langle x , \alpha \rangle}
           \Big] \, d \bar{x}_k \wedge d \bar{x}_j 
   = - \sum_j \sum_k 
   \dfrac{\alpha_j \alpha_k}{\langle x , \alpha \rangle^2} \, 
   d \bar{x}_k \wedge d \bar{x}_j 
   = 0, 
\end{align*}
where we used that
$
\dfrac{\partial}{\partial {x}_k} 
\Big[
\dfrac{\alpha_j}{\langle x , \alpha \rangle}
\Big] = 0,
$
since the function inside the brackets is anti-holomorphic
on its domain of definition. 
This proves that $\xi_W$ is closed. 
Since $ \nu (s) $ is simply a complex number, it is immediate
that $d \ddisp(s)= d (\nu(s) \, \xi_{ \sw(s)} ) = 0$, that is
$\ddisp(s)$ is closed for all $ s \in \S$. 

The cyclic property is a direct consequence of 
Proposition \ref{2-container}.  For the covariance property we see on the one hand that
\begin{align*}
F_g^\wedge ( \ddisp(s) ) &= F_g^\wedge ( \nu(s) \, \xi_{\sw(s)} ) 
= \nu(s) \, F_g^\wedge (  \xi_{\sw(s)} ) 
= \nu(s) \, (\id \otimes g) F^\wedge (  \xi_{\sw(s)} ) 
\\
&= \nu(s) \, (\id \otimes g) 
\sum_{k \in G} k \cdot \xi_{\sw(s)} \otimes \delta_k 
= \nu(s) \, \sum_{k \in G} k \cdot \xi_{\sw(s)} \otimes \delta_k (g) 
\\
&= \nu(s) \, (g \cdot \xi_{\sw(s)}). 
\end{align*}
On the other hand, by using 
$ W_s \cdot g := W_{g^{-1} s g} $, we compute
\begin{align*}
   \ddisp( g s g^{-1}) &= 
   \nu ( g s g^{-1}) \, \xi_{\sw (g s g^{-1})} 
   = 
   \nu (s) \, \xi_{\sw (g s g^{-1}) } 
   \\
   &= 
   \nu (s) ( \xi_{\sw(s) \cdot g^{-1}} )  
   = \nu(s) \, (g \cdot \xi_{\sw(s)}). 
\end{align*}
In the last equality we used 
$ g \cdot \xi_{W} = \xi_{W \cdot g^{-1}}$ which holds 
since the action by $ g $ is an orthogonal 
transformation.  

This finishes the proof of the 
covariance property and of the theorem. 
\end{proof}

\subsection{The Covariant Derivative}

Let us now compute the covariant derivative of the Dunkl connection $\omega=\tilde{\omega}+\lambda$. 
According to the general theory, 
for $\varphi \in \mathcal{D}^k$ we have that 
\begin{multline*}
        D_{\omega} (\varphi) = 
         D_{ \tilde{\omega} } (\varphi) - (-1)^k\varphi^{(0)}\lambda\pi(\varphi^{(1)})=
         \dP(\varphi)-(-1)^k \sum_{g \in G} (\varphi_g  \lambda\pi(\delta_g) )\\
         =\dP(\varphi)+(-1)^k\sum_{s \in \S} (\varphi 
         - \varphi_s) \lambda [s] =\dP(\varphi)+(-1)^k\sum_{s\in\S} (\varphi-\varphi_s)\nu(s)\xi_{\sw(s)}. 
\end{multline*}
This expression has the same abstract structure as the previously mentioned formula 
\eqref{D-omega-varphi}. 
\begin{remark} 
It is worth recalling that we chose to deal with anti-meromorphic displacements, the meromorphic version is obtained by the simple complex conjugation of the relevant objects. 
We can write $\dP=\partial_P+\bar{\partial}_P$
in the complex case, the decomposition into the holomorphic and anti-holomorphic differentials,  satisfying $\partial_P^2=\bar{\partial}_P^2=\partial_P\bar{\partial}_P+\bar{\partial}_P\partial_P=0$. Each 
of these differentials can be taken as the representative of the differential calculus on $P$. Of course, in this case 
we would lose the *-compatibility, but as already mentioned, this property was not essential for the present considerations. 
It is also important to observe that our expression holds in the very general context of 
all groups $ G $ of Coxeter type. 
\end{remark}

\subsection{The 
Covariant Derivative Viewed as 
Dunkl Operators}

In order to make an explicit link with the 
Dunkl operators introduced in \cite{dunkl-opdam}, let us fix for every $s \in \S$, 
an element $\alpha_s\in \sw(s)$, and use the formula \eqref{L-dunkl}. For the classical differential let us use the anti-holomorphic differential $\bar{\partial}_P$. We have 
\begin{equation}
D_\omega(\varphi)=\bar{\partial}_P(\varphi)+
(-1)^k\sum_{s\in\S} (\varphi-\varphi_s)\nu(s)\frac{\tilde{\alpha}_s}{\langle x,\alpha_s\rangle}
\end{equation}
which shows that our operators effectively accomodate, as a special case and modulo trivial modifications and 
a change of notation, those of \cite{dunkl-opdam}. In particular, in \cite{dunkl-opdam} the set $\S$ is 
naturally labeled by the hyperplanes associated to the complex reflections.
\begin{remark}
It is also worth noticing that with such change $\dP\rightsquigarrow \partial_P, \bar{\partial}_P$, the holomorphic (anti-holomorphic) forms on $P$ are fully preserved under the Dunkl covariant derivative action.  If in addition $G$ 
is a complex reflection group, then the corresponding polynomial forms will be preserved, as the 
singular contributions of the type $\langle x, \alpha \rangle^{-1}$ will always be canceled out. 
\end{remark}

\section{Example: Cuntz Algebras}

A rich class of truly quantum principal bundles equipped with cyclic Dunkl operators can be constructed 
from Cuntz algebras. We shall sketch here a kind of a universal 
construction for the full differential structure on these quantum principal bundles. 

If a finite-dimensional unitary space $V$ is given, its Cuntz algebra $\cal{O}(V)$ is defined as 
the *-algebra generated by the linear space $V$ and relations
\begin{equation}
\psi_i^*\psi_j=\delta_{ij}\qquad \sum_{k=1}^n\psi_k^{\phantom{*}}\psi_k^*=1
\end{equation}
where $\psi_1\dots\psi_n$ is an arbitrary orthonormal basis of $V$ and $n=\dim V\geq 2$. The algebra $\cal{O}(V)$ is independent of the choice of this orthonormal basis---it depends only on the scalar product on $V$. 

The Cuntz algebras \cite{Cuntz} naturally appear in the study of partial isometries in Hilbert spaces. Their enveloping  
C*-algebras are simple, and play a fundamental role in constructive algebraic quantum 
field theory \cite{DR}, where they provide a key structure to describe the superselection sectors, and related to 
this, the duality theory for compact groups \cite{DR2}. The Cuntz algebras are 
also (together with their matrix siblings) the building blocks for the quantum classifying spaces associated to 
compact matrix quantum groups \cite{M-Qclsp}. 

If $G$ is an arbitrary compact matrix quantum group unitarily acting on $V$, via a unitary representation 
$u\colon V\rightarrow V\otimes\cal{A}$ so that 
$$u(\psi_j)=\sum_{i=1}^n\psi_i\otimes u_{ij}$$
then $u$ uniquely extends to a *-homomorphism $F\colon\cal{O}(V)\rightarrow\cal{O}(V)\otimes\cal{A}$ which is 
the action of $G$ on $\cal{O}(V)$.  

If the representation $u$ is faithful and such that the conjugate representation $\bar{u}$ is contained as a subrepresentation of a tensor power $u^{\otimes m}$, then the action $F$ will be free (in the dual sense). 

Indeed, in this case the 
whole algebra $\cal{A}$ is generated by the matrix entries $u_{ij}$ (without involving the conjugates $u^*_{ij}$). On the other hand the above defining relations give 
$$ \psi^*_i F(\psi_j)=1\otimes u_{ij} $$ 
which can be iterated to prove that the relations of the form $\Sum_\alpha q_\alpha F(b_\alpha)=1\otimes a$ 
hold for all the elements of $\cal{A}$, which is precisely (the dual form of) the freeness condition for $F$.  

Interestingly, if $G$ is any finite classical group, this condition is always satisfied. In other words we have a 
quantum principal bundle $P$ based on $\cal{O}(V)$. The "base space" algebra $\cal{V}$ is given by the 
invariants of this action and it is highly non-commutative.  

Let $\cal{O}^+(V)$ be the subalgebra generated by $V$. If $S\subseteq G\setminus{\e}$ is the defining 
set for a bicovariant *-calculus $\Gamma$ over $G$ then we can always construct a map $\zeta\colon\ginv\rightarrow\cal{O}^+(V)$ which intertwines the adjoint action $\ad$ and $F$ and such that its values on 
$S$ are all non-zero. This is a direct consequence of the above generating property for the representation $u$ 
and the definition of the Cuntz algebras.  

On the other hand, by universality, the action $F$ extends uniquely to a homomorphism of graded differential *-algebras $F\colon\Omega\cal{O}(V)\rightarrow \Omega\cal{O}(V)\otimes\cal{A}$, 
where $\Omega\cal{O}(V)$ is the universal differential envelope of the algebra $\cal{O}(V)$. 

We would like to obtain here a cyclic Dunkl displacement. A natural candidate is the differential of $\zeta$ 
in $\Omega\cal{O}(V)$. If we thus define $\lambda(s):=d\zeta(s)$, the Covariance and Closed-ness properties 
in the definition of cyclic Dunkl displacements are automatically satisfied. The only property which is possibly 
not fulfilled is the Cyclic Property. But this is easily `remedied' by observing that the second-order 
elements of the form 
$\lambda(\ell)=\lambda(s_1)\lambda(s_2)+\cdots+\lambda(s_{n-1})\lambda(s_n)+\lambda(s_n)\lambda(s_1)$ 
are all exact and transform covariantly, so they generate a graded differential *-ideal $\cal{I}$ 
in $\Omega\cal{O}(V)$ such that $F(\cal{I})\subseteq\cal{I}\otimes \cal{A}$. 
If we now set $\hor(P):=\Omega\cal{O}(V)/\cal{I}$ and project $F$ and $\lambda$ 
from $\Omega\cal{O}(V)$ to $\hor(P)$, then such a projected $\lambda$ is a cyclic Dunkl displacement for 
$\hor(P)$. 

We can then proceed with the general constructions to craft the complete calculus, as described 
in Section~5. It is interesting to observe that
here $\Omega\cal{O}(V)$ is used to build the horizontal forms only, an appropriate 
full calculus on $P$ is constructed as a 
twisted tensor product $\Omega(P)\leftrightarrow\hor(P)\otimes\ginv^\wedge$. The projected differential 
$d\colon\hor(P)\rightarrow\hor(P)$ is interpreted as the covariant derivative of an associated  
`initial' zero-curvature connection $\tilde{\omega}$.  

\section{Concluding Observations}
\label{sec-7}

This paper deals with the cyclic Dunkl operators
exclusively from the 
point of view of purely algebraic aspects of 
non-commutative geometry. 
In a future paper we will present the analytic properties 
of these new operators, such as a generalized 
Dunkl transform. 
We also intend to analyze in detail 
examples involving the groups of the Coxeter type, beyond 
the complex reflection groups. The geometry of the quantum 
principal bundle is essentially given by the choice of 
$\S$, which determines 
the differential calculus. 

Another interesting class of examples, which we leave for a future study, comes from appropriately quantized Euclidean spaces on which a finite group $G$ acts, essentially without changing 
the action on basic coordinates. Among other 
purely quantum phenomena in this context, we can mention an automatic freeness of the group action--as there are simply no points to possibly manifest themselves as
elements in a non-trivial stabilizer. 
Also, as we mentioned earlier the case when the 
tangent bundle of $ E $ is trivial merits closer
examination, something we plan to do in a future 
paper.

\section*{Appendix A: Cyclic Structures}
\label{sec-8}

Here we shall explain how the basic geometry of oriented cycles, 
that naturally emerges in the context of
differential calculi on finite groups, can be formalized, 
so that we can recover the group and 
the calculus from a simple set of axioms. 

Let $\Omega$ be a finite set, equipped with a family $\cal{T}$ 
of cyclically oriented subsets of $ \Omega $.  
Every such set in $\cal{T}$ is representable 
as some ordered $n$-tuple $(\omega_1, \dots, \omega_n)$ with $n \ge 1$ 
mutually distinct elements 
$\omega_1 , \dots , \omega_n \in \Omega$. 
Here we identify those $n$-tuples that can be 
obtained one from another by using cyclic permutations. 

The first property that we shall assume here 
is that the elements of $\cal{T}$ behave like "oriented lines". 
Specifically, for every ordered pair $(x, y)\in\Omega\times\Omega$ 
with $x\neq y$ we assume there exists a unique 
$\ell_{x,y}\leftrightarrow(\omega_1,\dots, \omega_n)$ 
in $\cal{T}$ such that $\omega_1=x$, $\omega_2=y$ 
and $\omega_j \in \Omega$ for all $ 3 \le j \le n $. 

In this case we necessarily have a discrete `line' 
with $n \ge 2$ points. 
To this we add a kind of normalization property 
by also assuming that there exists a line 
$\ell_{x,x}=(x) \in \cal{T}$ 
for every $x\in\Omega$. 
So in this case the `line' has exactly one point. 

We can think of this as an oriented version of one 
of Euclid's axioms: 
Every oriented pair of points determines a unique 
oriented line on which they lie. 

\begin{definition} 
Every pair $(\Omega, \cal{T})$ satisfying the above properties 
will be called a {\it cyclic space}. 
\end{definition}

We are mainly interested in the case $\Omega=\S$, 
a subset of a finite group $G$ closed under conjugation 
by all $ g \in G $ and $\cal{T}$ 
being associated to the orbits of the action of 
the flip-over operator $\sigma$ on $\S\times \S$. 
However, there are different realizations of this simple scheme. 
In this regard 
it seems natural to introduce the following definition. 
\begin{definition} 
Let $\mathfrak{h}$ be an algebra and 
$(\Omega,\cal{T})$ a cyclic space. 
Then we define a 
{\em Dunkl representation} of $(\Omega,\cal{T})$ in 
$\mathfrak{h}$ 
to be a map $\xi\colon\Omega\rightarrow\mathfrak{h}$ satisfying 
\begin{equation}
\label{A5}
\xi(w_1)\xi(w_2)+\xi(w_2)\xi(w_3)+\cdots+\xi(w_{n-1})\xi(w_n)
+\xi(w_n)\xi(w_1)=0
\end{equation}
for every line 
$\ell\leftrightarrow(w_1,\dots,w_n)$ in $\cal{T}$. 
In particular we always have $\xi(w)^2=0$ for every 
$w\in\Omega$.  
Also, if a line consists of two distinct points 
$w_1$ and $w_2$ only 
(which is equivalent to saying 
that $ w_1 $ and $ w_2 $ commute 
when $ \Omega \subset G $ as described 
in the main body of this paper), 
then \eqref{A5} says 
that $\xi(w_1)$ and $\xi(w_2)$ anti-commute in $ \mathfrak{h} $. 
\end{definition}

In the main part of this paper the example of a 
Dunkl representation in the algebra 
$\mathfrak{h} = \hor(P) $ was given. 
In that example it happens to be the case that $ \mathfrak{h} $ 
is actually a $ * $-algebra. 

For example, we can consider the Fano plane, 
consisting of the $7$ imaginary units in 
the classical non-associative algebra of octonions $\Bbb{O}$. 
In this case $\cal{T}$ consists of $21$ elements. 
The elements of $\cal{T}$ are the $7$ lines $\ell_{x,x}$ 
for each of the $7$ points $x$ of the 
Fano  plane plus $14$ oriented cycles 
which correspond to the $7$ lines of the Fano plane, 
each line containing 3 points and taken with both 
of the two possible cyclic orientations.  
For this example to work out, one has to prove that 
for each of the $7 \cdot 6 = 42$ ordered pairs of points $x \neq y$
in the Fano plane there exists a unique oriented cycle
$(x,y,w_3)$ in $\cal{T}$. 
For more on the Fano plane and octonions see \cite{baez}. 

On the other hand, as we shall see below, 
if we add a couple of simple additional properties, 
this entire context becomes equivalent to that 
of the quantum differential calculus as presented 
in the main part of this paper. 

Given a cyclic space $(\Omega, \cal{T})$ we can introduce 
a natural action in the following way. 
We start with an ordered pair $(x,y) \in \Omega\times\Omega$ 
with $x\neq y$
and consider the uniquely determined element 
$\ell_{x,y}  \in \cal{T}$ represented by the ordered $n$-tuple 
$(x,y,\omega_3,\dots, \omega_n)$. 
Then for $n \ge 3$ we define $x\wact y:=\omega_3$, while for
$n = 2$ we define $x\wact y:= x$. 
Finally, for the diagonal elements $(x,x)\in\Omega\times\Omega$
we define $x\wact x=x$. 
In short, we have defined a function
$\wact\colon\Omega\times\Omega\rightarrow\Omega$.  
We can think that $y$ `acts' upon $x$ from the right 
in the expression $x\wact y$. 

\begin{prop} 
Let $(\Omega, \cal{T})$ be a cyclic space. 
Then we have this cancellation property: 
If $a\wact x=b\wact x$, then $a=b$.  
\end{prop}

\begin{proof} 
Let $w:= a\wact x=b\wact x$. 
Then $\ell_{x,w} \in \cal{T}$ is uniquely determined. 
We first consider the case when $x \neq w$. 
Therefore, $\ell_{x,w}$ is represented by the $n$-tuple
$(x,w, \dots, w_{n-1}, w_n)$ as well as by 
$(w_n, x,w, \dots, w_{n-1})$. 
But $a\wact x = w$ means that $\ell_{a,x}$ is represented 
by $(a,x,w, \dots )$ for some $m$-tuple. 
But the uniqueness of $\ell_{x,w}$ forces $\ell_{x,w} = \ell_{a,x}$. 
In particular, it follows that $a=w_n$. 
Similarly, $b\wact x = w$ implies $b=w_n$. 
Consequently, $a = b$ as desired. 

For the case when $x = w$ we have $a\wact x=x$.
This forces $a = x$. 
Similarly, 
$b\wact x=x$ forces $b = x$ in this case.
It follows that $ a = b$. 
\end{proof}

To every $x\in\Omega$ we can associate a map 
$(\cdot)\wact x\colon\Omega\rightarrow\Omega$ 
by mapping $a \in\Omega$ to $a \wact x \in\Omega$. 
In accordance with the above cancellation property, 
all these maps are injective. 
But since $\Omega$ is finite, they are all actually 
permutations of $\Omega$.  
We let $ \mathrm{Perm}(\Omega) $
denote the finite group of permutations 
of the finite set $ \Omega $. 
Let us now assume that a symmetric left cancellation property holds. 
Namely
$$
(\forall a\in \Omega) \quad a \wact x=a\wact y
\quad \Rightarrow \quad x=y. 
$$
This means the elements of $x\in\Omega$ are faithfully 
represented by the permutations 
$(\cdot)\wact x$ of $\Omega$. 
To this we shall add the following:

\vskip 0.2cm \noindent
{\bf Non-triviality Assumption}: 
Every such permutation $(\cdot)\wact x$ 
is non-trivial, in other words, 
for every $x\in\Omega$ there exists at least one 
$a\in\Omega$ such that $a\wact x\neq a$. 

\begin{remark} 
In our context of quantum differential calculus on finite groups, 
this assumption is not necessarily the case, 
although it holds in the most interesting examples, 
where the group $G$ acts faithfully by 
conjugation on the basis set $\S$ of $ \ginv $.  
\end{remark}

Our next assumption deals with the compatibility 
between cyclical lines and the defined action. 
We postulate that 
\begin{equation}\label{lines-action}
(z\wact x)\wact y=(z\wact y)\wact (x\wact y)
\end{equation}
for every $x,y,z\in\Omega$. 

\begin{remark}
Algebraically, 
this equation says the right action $  \wact y $ 
distributes over the binary operation $ \wact $. 
Geometrically, it says that the ordered sequence of 
consecutive points $z, x, z \wact x$ 
on their uniquely determined
cyclic line transforms 
under the right action $ ( \cdot ) \wact y $
into the ordered sequence of 
consecutive points 
$ z \wact y, x \wact y, (z \wact x) \wact y $ 
on their uniquely determined cyclic line. 

This property trivializes for $x=y$ or $x=z$. 
On the other hand, in the special case 
$y=z$ it reduces to 
$$
(y\wact x)\wact y=y\wact (x\wact y)
$$
for every $x,y\in\Omega$, which is an  
associativity property for any two elements 
$ x $ and $ y $. 
\end{remark}

\begin{prop} 
Let $(\Omega, \cal{T})$ be a cyclic space 
that satisfies \eqref{lines-action}. 
Then for every cyclic line 
$\ell \in \cal{T}$ with $\ell \leftrightarrow (w_1, \dots, w_n)$  
the composition 
$$ 
x \mapsto (x \wact w_k) \wact w_{k+1} 
$$ 
of the right actions of any two 
cyclically consecutive 
elements $(w_k$, $ w_{k+1}) $ 
is a permutation of $\Omega$ that does not depend on $k$, but 
only depends on $ \ell $. 
Here $1 \le k \le n$ and $ w_{n+1} = w_1 $. 
\end{prop} 

\begin{proof} 
The assertion is trivially true if $ n = 1 $. 
For $ n \ge 2 $ 
the statement of the theorem is equivalent to saying that
\begin{equation}
\label{wact-equivalence}
    (x \wact w_k) \wact w_{k+1} = (x \wact w_1) \wact w_2
\end{equation}
holds for all $ x \in \Omega $ and all $1 \le k \le n$. 
The case $ k = 1 $ is trivial. 
For $ k =2 $ we compute 
$$
    (x \wact w_2) \wact w_3 =  (x \wact w_2) \wact ( w_1 \wact w_2) 
    = (x \wact w_1) \wact w_2, 
$$
where we used first used $ w_{3} = w_{1} \wact w_{2} $ 
and then  
the identity \eqref{lines-action} in 
the second equality. 
Then \eqref{wact-equivalence} in general follows by an induction 
that uses \eqref{lines-action} again.  
\end{proof}

Let $G$ be the subgroup of $ \mathrm{Perm}(\Omega) $ 
generated by all the permutations
 $(\cdot)\wact x$, where $x\in\Omega$. 
That is, 
$ G:= 
\big\langle \, (\cdot)\wact x ~|~ x \in \Omega \,  \big\rangle
\subset \mathrm{Perm}(\Omega) $, 
and so $G$ is a finite group. 
But $\Omega$ can also be viewed 
via $\Omega\ni x\rightsquigarrow (\cdot)\wact x\in G$ 
as a subset of $G$. 
So we can extend the right action 
$\wact : \Omega \times \Omega \to \Omega$ 
to a right action 
$\wact\colon \Omega\times G\rightarrow \Omega$. 
In particular for $ x,y,z \in \Omega $ we have 
$(z \wact x) \wact y = z \wact (x y)$, where 
we are using the identifications of $x$ and $y$ 
as the permutations $ ( \cdot ) \wact x $ and 
$ ( \cdot ) \wact y $, respectively, and therefore
taking $x y$ to mean their product in $\mathrm{Perm}(\Omega) $. 

\begin{prop} 
Let $(\Omega, \cal{T})$ be a cyclic space 
that satisfies \eqref{lines-action} as well as 
the Non-triviality Assumption.  
Then in terms of the above identification, we have 
$$
x\wact y=y^{-1}xy 
$$
for every $x,y\in\Omega$. 
In particular, the set $\Omega$ is the disjoint union 
of conjugation classes of 
$G$ and $\e\notin \Omega$. 
\end{prop}
\begin{proof} 
We have to prove that 
$$ z\wact (y^{-1}xy)=z\wact (x\wact y)$$
for every $x,y,z\in\Omega$. 
Without a lack of generality we can replace $z$ by $z\wact y$, 
in which case the equality to prove becomes
$$ z\wact (xy)=(z\wact y)\wact (x\wact y) $$
because of $(z\wact y)\wact (y^{-1}xy)=z\wact (xy)$. 
However, this is just a rewritten form of the
property \eqref{lines-action}. 

We see that $\Omega$ is invariant under conjugations 
by elements of $\Omega$ and, since $\Omega$ generates $G$, 
it follows that $\Omega$ is invariant under all conjugations
by elements in $ G $, 
that is, it is a disjoint union 
of conjugation classes of $ G $. 
The fact that $\e\notin\Omega$ is nothing but another way 
of expressing  
the non-triviality assumption for the elements of $\Omega$.  
\end{proof}

In particular, the cyclical lines are precisely the projected orbits 
of the action of the inverse of the canonical flip-over operator 
$\sigma^{-1}$ on $\Omega\times\Omega$. 

\begin{remark} 
It is worth observing that the condition \eqref{lines-action} 
can be restricted only for those three elements $x$, $y$ and $z$ 
belonging to a single line. 
This would include the appropriate non-associative structures, 
like finite Moufang loops, in which every two elements generate 
an associative subgroup. 
As explained in \cite{u2} such non-associative objects can still 
be viewed as diagrammatic groups, within a more general framework 
of diagrammatic categories and collectivity structures. 
\end{remark}

\section*{Appendix B: Operations in $\Omega (P)$} 
\label{sec-9}

We include here for the reader's convenience 
the definitions of the four basic 
operations on $\Omega (P)= \cal{D} \otimes \ginv^{\wedge}$. 
In the following for $\varphi \in \mathfrak{D}$ 
we use Sweedler's notation 
$$
_{\mathfrak{D} }\Phi (\varphi) 
= \varphi^{(0)} \otimes \varphi^{(1)} 
\in \mathfrak{D} \otimes \mathcal{A} 
$$
for the right co-action $ _{\mathfrak{D} }\Phi $ 
defined in \eqref{define-Coxeter-co-action}. 

Then for 
$\psi \otimes \theta , \varphi \otimes \eta \in \Omega (P)
= \cal{D} \otimes \ginv^{\wedge}
$ 
with $\deg (\theta)=k$ and $\deg (\varphi) = j$ 
we define their product by 
\begin{equation}
\label{define-twisted-mult}
   (\psi \otimes \theta )(\varphi \otimes \eta):= 
   (-1)^{j k} \psi \varphi^{(0)} \otimes 
   (\theta \circ \varphi^{(1)}) \eta.
\end{equation}
This bilinear expression defines a linear map, 
also called the multiplication, and 
denoted as 
$  m_{\Omega (P)} : \Omega (P) \otimes \Omega (P) \to \Omega (P)$. 

The $*$-operation is defined for $\deg(\varphi) = j$ 
and $\deg (\theta)=k$ by 
$$
   (\varphi \otimes \theta)^* := 
   (-1)^{j k}  \varphi^{(0)*} \otimes (\theta^* \circ \varphi^{(1)*}).
$$

The differential $ d_P $ in $\Omega (P)$ is defined 
for $\deg (\varphi) = j$ by
$$
    d_P (\varphi \otimes \theta) := 
    D(\varphi) \otimes \theta 
    + (-1)^{j} \varphi^{(0)} \otimes \pi (\varphi^{(1)} ) \theta
    + (-1)^{j} \varphi \otimes d^\wedge (\theta),
$$
where 
$d^\wedge :  \, \ginv^\wedge 
\to \, \ginv^\wedge$ 
is the restriction of the differential $d^\wedge$ 
defined on the acceptable algebra $\Gamma^\wedge$. 
Also, $ D $ is the complexified de Rham differential. 

Finally, 
there is a right co-action 
$ \FWP : \Omega (P) \to \Omega (P) \otimes \Gamma^\wedge $
of $\Gamma^\wedge$ on $\Omega(P)$ 
that is explicitly defined by 
$$
\FWP (\varphi \otimes \theta) := 
\varphi^{(0)} \otimes \theta^{(0)} \otimes \varphi^{(1)} \theta^{(1)}.
$$Since 
$\hatphi : \Gamma^\wedge \to \Gamma^\wedge \otimes \Gamma^\wedge$
restricts to 
$\hatphi : \, \ginv^\wedge 
\to  \, \ginv^\wedge \otimes \Gamma^\wedge$, 
for $\theta \in \,\! \ginv^\wedge$ 
we also are using Sweedler's notation 
$$
\hatphi (\theta) = \theta^{(0)} \otimes \theta^{(1)} \in  
\,\! \ginv^\wedge \otimes \Gamma^\wedge. 
$$ 

The right co-action $ \FWP $ extends $ F $ 
and is a differential, unital, 
degree zero $*$-morphism of graded  algebras. 

Hence $\Omega(P)$ is a graded differential, unital $*$-algebra. 
The differential $d_P$ satisfies  
the graded Leibniz rule with respect to 
the product \eqref{define-twisted-mult}
on $\Omega(P)$, 
is covariant with respect to the co-action $\FWP$
and is a $*$-morphism. 
And therefore 
$( \Omega (P), \Gamma^\wedge, \FWP )$ is an hodc which extends
the QPB $  P = (C^\infty (E) , \cal{F} (G) , F ) $. 

It is worth mentioning that each of these four operations 
involve a `twisting' coming from the right co-action 
$ _{\mathfrak{D} }\Phi $. 
If this co-action is trivial, (i.e.,
$ _{\mathfrak{D} }\Phi (\varphi) = \varphi \otimes 1 $ 
for all $\varphi \in \mathfrak{D}$) , 
then these operations reduce to tensor product formulas, and
so in this particular case it is correct 
to think that the structure of the total space
is that of a tensor product. 
But in general these operations are 
not tensor products and so the total 
space is not simply a tensor product 
despite what definition 
\eqref{define-Omega-P}
 might otherwise suggest.

\section*{Appendix C: A Technical Proof}
\label{sec-10}

In this Appendix we present a rather long 
technical proof of \eqref{sigma-eta-vartheta} 
in order to have a more 
complete presentation without interrupting the 
flow of the main body of this paper. 

The motivation for considering the expression 
$\sigma (\eta \otimes \vartheta) $ for both entries 
being left invariant starts with the observation 
that the braiding operation $\sigma $
is the flip {\em provided} that $\eta $ 
is left invariant and $\vartheta $ is right 
invariant, namely 
$\sigma (\eta \otimes \vartheta) = \vartheta \otimes \eta $ 
in this case.  
This property of $\sigma $ is essentially its 
definition. 
But in classical differential geometry one considers 
either left invariant forms without ever mentioning 
right invariant forms (the usual convention) 
or, on the other hand, only 
right invariant forms without ever mentioning 
left invariant forms. 
So it becomes a matter of curiosity to understand 
how $\sigma $ acts in the `un-mixed' case when 
both forms are left invariant. 
And the resulting identity then turns out 
to have its own utility. 

We start out by establishing some notation. 
See Chapter~5 of \cite{Part-II} for more details 
on this notation and related properties. 
Throughout we take $\Gamma$ to be a bicovariant fodc over 
a Hopf algebra $\cal{A}$, though the result 
holds in the more general setting of 
Chapter~5 of \cite{Part-II}. 
We let $\{ \omega_i ~|~ i \in I \}$ be a basis 
of the vector space $\ginv$ of left invariant 
forms in $\Gamma $. 
And we also let $\{ \eta_i ~|~ i \in I \}$ be a basis 
of the vector space of right invariant 
forms in $\Gamma $. 
These bases are related by 
$$
    \omega_i = \sum_{j \in I} \eta_j R_{ji} 
\quad \mathrm{and} \quad 
    \eta_j = \sum_{i \in I} \omega_i \kappa (R_{ij} )
$$
for unique elements $R_{ij} \in \cal{A} $. 
(These identities are inverses of each other.) 

Since we are assuming that both $\eta$ and $\vartheta $
are left invariant, we expand them in terms of the basis 
of $\ginv$ as
$
   \eta = \sum_{j \in I} \lambda_j \omega_j 
$
and
$
   \vartheta = \sum_{i \in I} \sw_i \omega_i  
$, 
where $\lambda_j , \sw_i \in \mathbb{C}$. 

We have this identity for the right adjoint 
co-action $\ad$ and the right canonical co-action $ _\Gamma \Phi $,  
which are equal when evaluated 
on left invariant elements, namely that 
$
  \ad (\omega_i) = \, _\Gamma \Phi (\omega_i) =
  \sum_{j \in I} \omega_j \otimes R_{ji} 
$. 
From this it immediately follows that 
$$
\ad (\vartheta) = 
\vartheta^{(0)} \otimes \vartheta^{(1)} = 
\sum_i \sw_i \ad (\omega_i) = 
\sum_{ij} \sw_i \omega_j \otimes R_{ji} = 
\sum_{ij} \omega_j \otimes \sw_i R_{ji}.  
$$
Some other identities that we will use are:  
\begin{equation*}
 \omega_i \, b = \sum_j ( f_{ij} \ast b ) \, \omega_j, 
 \quad
 \eta_i \, b = \sum_j (b \ast g_{ij}) \eta_j, 
 \quad
 \phi (R_{ij}) = \sum_k R_{ik} \otimes R_{kj}. 
\end{equation*}
Here $b \in \cal{A}$ and $f_{ij}, g_{ij} : \cal{A} \to \mathbb{C}$
are a doubly indexed families of linear functionals known 
as the {\em structure representations}. 
(In this context $f_{ij} = g_{ij} $ 
a fact we do not need.) 
Also, the symbol $\ast$ refers to two different 
{\em convolution products} between elements of $\cal{A}$ and linear 
functionals on $\cal{A}$. 

Here is the derivation of \eqref{sigma-eta-vartheta}. 
We take $\eta,  \vartheta \in \ginv $ and compute 
using all of the above identities as follows: 

\begin{align*}
&\sigma ( \eta \otimes \vartheta ) = 
\sum_{il} \lambda_l \, \sw_i \, \sigma ( \omega_l \otimes \omega_i )
= \sum_{ijl} \lambda_l \, \sw_i \, \sigma ( \omega_l \otimes \eta_j R_{ji} ) 
\\
&= \sum_{ijkl} \lambda_l \, \sw_i \, \sigma ( \omega_l \otimes (R_{ji} \ast g_{jk} ) \, \eta_k )  
= \sum_{ijkl} \lambda_l \, \sw_i \, \sigma \big( \omega_l (R_{ji} \ast g_{jk} ) \otimes \eta_k \big)   
\\
&= \sum_{ijklm} \lambda_l \, \sw_i \, \big( f_{lm} \ast (R_{ji} \ast g_{jk} ) \big) \, 
\sigma ( \omega_m  \otimes \eta_k )\\
&= \sum_{ijklm} \lambda_l \, \sw_i \, \big( (f_{lm} \ast R_{ji}) \ast g_{jk} ) \big) \, 
( \eta_k \otimes \omega_m  ) 
\\
&= \sum_{ijklm} \lambda_l \, \sw_i \, \big( (f_{lm} \ast R_{ji}) \ast g_{jk} ) \, 
 \eta_k \big) \otimes \omega_m 
=  \sum_{ijlm} \lambda_l \, \sw_i \,  
 \big( \eta_j  (f_{lm} \ast R_{ji}) \otimes \omega_m \big)   
 \\
&= \sum_{ijlm} \lambda_l \, \sw_i \,  
 \big( \eta_j \otimes (f_{lm} \ast R_{ji}) \, \omega_m \big) 
= \sum_{ijl} \lambda_l \, \sw_i \,  
 \eta_j \otimes   \, \omega_l \, R_{ji} 
\end{align*}
\begin{align*}
&= \sum_{ij} \sw_i \,  
 \eta_j \otimes   \, \sum_l \lambda_l \, \omega_l \, R_{ji} 
= \sum_{ij} \sw_i \,  
 \eta_j \otimes \eta \, R_{ji} 
\\ 
&= \sum_j  
 \eta_j \otimes \eta \sum_i \sw_i \, R_{ji} 
= \sum_{jk}   
  \omega_k \, \kappa (R_{kj})  \otimes \eta \sum_i \sw_i \, R_{ji} 
\\
&= \sum_{jk}   
  \omega_k \otimes \kappa (R_{kj}) \, \eta \sum_i \sw_i \, R_{ji} 
= \sum_{k}   
  \omega_k \otimes \sum_i \sw_i \, \sum_j \kappa (R_{kj}) \, \eta \, R_{ji} 
  \\
&= \sum_{k}   
  \omega_k \otimes \sum_i \sw_i \, ( \eta \circ R_{ki} ) 
  = \sum_{ik}   
  \omega_k \otimes ( \eta \circ \sw_i \, R_{ki} ) 
  = \vartheta^{(0)} \otimes ( \eta \circ \vartheta^{(1)} ).  
\end{align*}

We also used the identity 
$\vartheta \circ a = \kappa (a^{(1)}) \, \vartheta \, a^{(2)}$ 
and an associativity property of the convolution operations.


\begin{thebibliography}{99}
\label{bibliography}


\bibitem{baez}
Baez,~J.C.: The Octonions. 
Bull. Am. Math. Soc. \textbf{39}, 145--205 (2002).  

\bibitem{Cuntz}
Cuntz,~J.: Simple C*-algebras Generated by Isometries. Commun.~Math.~Phys. \textbf{57}, 173--185 (1977). 

\bibitem{DR}
Doplicher,~S., Roberts,~J.~E.: Endomorphisms of C*-algebras, Cross Products and Duality for Compact Groups. 
Ann.~Math. \textbf{130}, 75--119 (1989). 

\bibitem{DR2}
Doplicher,~S., Roberts,~J.~E.: Why There is a Field Algebra with a Compact Gauge Group Describing the 
Superselection Structure in Particle Physics. Commun.~Math.~Phys. \textbf{131}, 51--107 (1990). 

\bibitem{dunkl-1989}
Dunkl,~C.F.: Differential difference 
operators associated to reflection groups. 
Trans. Am. Math. Soc. \textbf{311},  
167--183 (1989).

\bibitem{dunkl-opdam}
Dunkl,~C.F., Opdam,~E.M.: 
Dunkl Operators for Complex Reflection Groups.  
Proc. London Math. Soc. {\bf 86}, 70--108 (2003).

\bibitem{MichoQPB1}
\Dj ur\dj evich, M.:
Geometry of Quantum Principal Bundles I.  
Commun. Math. Phys. {\bf 175}, 475--521  (1996).

\bibitem{MichoQPB2}
\Dj ur\dj evich,~M.: Geometry of Quantum Principal Bundles II, Extended Version. 
Rev. Math. Phys. \textbf{9}, No. 5, 
531--607 (1997). 

\bibitem{MichoQPB3}
\Dj ur\dj evich,~M.: 
Geometry of Quantum Principal Bundles III. 
Alg. Groups Geom. {\bf 27},  247--336 (2010). 

\bibitem{u2}
\Dj ur\dj evich, M.: Geometric Quantum Groups as Realizations 
of Diagrammatic Symmetry Category. 
Alg. Groups Geom. {\bf 33}, 369--426  (2016). 

\bibitem{M-Qclsp}
\Dj ur\dj evich, M.: Quantum Classifying Spaces and Universal Quantum
Characteristic Classes. Publications of Stefan Banach
Mathematical Center, {\bf 40}, 315--327 (1997).

\bibitem{DS}
\Dj ur\dj evich~M., Sontz~S.B.:  
Dunkl Operators as Covariant Derivatives
in a Quantum Principal Bundle. 
SIGMA \textbf{9}, 040, 29 pages, (2013). \\
http://dx.doi.org/10.3842/SIGMA.2013.040. 

\bibitem{etingof}
Etingof,~P.: Calogero-Moser systems and representation theory,  
Z\"urich Lectures in Advanced Mathematics,
European Mathematical Society (EMS), Z\"urich, 
(2007).

\bibitem{grove-benson}
Grove~L.C., Benson~C.T.:  
Finite Reflection Groups, 2nd. Ed.,
Springer-Verlag, New York, (1985).

\bibitem{hump}
Humphreys~J.E.: 
Reflection Groups and Coxeter Groups,  
Cambridge University Press, Cambridge, Cambridge,  (1990). 

\bibitem{rosler}
R\"osler, M.: Dunkl operators: theory and applications.  
In: Koelink, E., Van Assche, W. (eds.) pp, 93-135, 
Orthogonal Polynomials and Special Functions, 
(Leuven, 2002), 
Lecture Notes in Math., Vol. 1817, pp.~93--135, 
Springer, Berlin (2003).

\bibitem{rosler-voit}
R\"osler, M., Voit, M.: 
Markov processes related with Dunkl operators, 
Adv. in Appl. Math. {\bf 21},  575--643 (1998).

\bibitem{sontz-heat}
Sontz,~S.B.: On Segal-Bargmann analysis for finite Coxeter groups 
and its heat kernel. 
Math. Z. {\bf 269}, 9--28  (2011). 

\bibitem{Part-II} 
Sontz,~S.B.: Principal Bundles, The Quantum Case. 
Universitext, Springer, Cham (2015).  

\bibitem{spivak}
Spivak,~M.: A Comprehensive Introduction to 
Differential Geometry, Vol. 2, 3rd Ed., 
Publish or Perish, Inc., Houston, (1999). 

\bibitem{Wtwisted} 
Woronowicz~S.L.: Twisted SU(2) group. 
An example of a noncommutative differential calculus, 
Publ. Res. Inst. Math. Sci., Kyoto University 
23, 117--181 (1987). 

\bibitem{Wdiff}
Woronowicz,~S.L.: Differential Calculus on Compact Matrix Pseudogroups
(Quantum Groups), Commun. Math. Phys. 122, 
 125--170 (1989). 

\end{thebibliography}
\end{document}